\documentclass{article}
\usepackage[letterpaper, margin=1.1in]{geometry}

\usepackage{soul}
\usepackage{url}
\usepackage[hidelinks]{hyperref}
\usepackage[T1]{fontenc}
\usepackage[utf8]{inputenc}
\usepackage[small]{caption}
\usepackage[shortlabels]{enumitem}
\usepackage{graphicx}
\usepackage{amsmath,amssymb,amsfonts,amsthm}
\usepackage{mathtools}
\usepackage{booktabs}
\usepackage{framed}
\usepackage{xspace}
\usepackage{algorithm}
\usepackage[noend]{algpseudocode}
\usepackage[dvipsnames]{xcolor}
\usepackage{todonotes}

\usepackage{cancel}

\usepackage{tikz}
\usepackage{todonotes}
\usetikzlibrary{shapes,shapes.callouts,shadows,arrows,backgrounds,%
matrix,patterns,arrows,decorations.pathmorphing,decorations.pathreplacing,%
positioning,fit,calc,decorations.text%
}

\theoremstyle{plain}
    \newtheorem{theorem}{Theorem}
    \newtheorem{corollary}[theorem]{Corollary}
    \newtheorem{lemma}[theorem]{Lemma}
    \newtheorem{proposition}[theorem]{Proposition}
    \newtheorem{claim}{Claim}[theorem]
    
    \newtheorem*{theorem*}{Theorem}
    \newtheorem*{corollary*}{Corollary}
    \newtheorem*{lemma*}{Lemma}
    \newtheorem*{proposition*}{Proposition}
    \newtheorem*{claim*}{Claim}
    \newtheorem*{conjecture*}{Conjecture}

\theoremstyle{definition}

    \newtheorem*{definition*}{Definition}
    \newtheorem*{example*}{Example}

\theoremstyle{remark}

    \newtheorem*{remark*}{Remark}
    \newtheorem*{note*}{Note}

\newenvironment{claimproof}
{\noindent {\em Proof of claim:} }
{\hfill $\diamond$ \smallskip}

\DeclarePairedDelimiter\abs{\lvert}{\rvert}
\DeclarePairedDelimiter\norm{\lVert}{\rVert}
\DeclarePairedDelimiter\ceil{\lceil}{\rceil}

\mathchardef\mhyphen="2D

\def\A{{\bf A}}

\DeclareSymbolFont{bbold}{U}{bbold}{m}{n}
\DeclareSymbolFontAlphabet{\mathbbold}{bbold}
\newcommand*{\reals}{\mathbb{R}}
\newcommand*{\integers}{\mathbb{Z}}
\newcommand*{\rationals}{\mathbb{Q}}
\newcommand*{\naturals}{\mathbb{N}}

\newcommand*{\range}[2]{\{#1,\dots,#2\}}
\newcommand*{\csp}{\textsc{CSP}}
\newcommand*{\instance}{\mathcal{I}}
\newcommand*{\variables}{V}
\newcommand*{\constraints}{C}
\newcommand*{\declareInstance}{\instance = (\variables, \constraints)}
\newcommand*{\nxnindset}{n \times n \textsc{ Independent Set}}
\newcommand*{\unitallen}{{\bf A}_{\rm ua}\xspace}
\newcommand*{\dcsp}[1][d]{#1\mhyphen\csp}
\newcommand*{\dcspchi}[1][d]{#1\mhyphen\csp^{\chi}}

\newcommand*{\disjtemp}[1]{{\bf D}_{#1}}
\newcommand*{\stpfull}{\bf{S}\xspace}

\newcommand*{\allin}[1]{\mathcal{U}({#1})\xspace}

\newcommand*{\compare}[1]{\odot_{#1}}

\algrenewcommand{\algorithmiccomment}[1]{\hfill$//$ \textit{#1}}

\newcommand*{\Accept}{\textbf{accept\xspace}}
\newcommand*{\Reject}{\textbf{reject\xspace}}
\newcommand*{\algtab}{\hskip\algorithmicindent}

\newcommand{\cc}[1]{{\mbox{\textnormal{\textsf{#1}}}}\xspace}  %% Complexity class
\newcommand{\Nat}{\mathbb{N}}
\newcommand{\bigoh}{O}
\newcommand{\PP}{\cc{P}}
\newcommand{\NP}{\cc{NP}}

\newcommand{\PSPACE}{\cc{PSPACE}}
\newcommand{\FPT}{\cc{FPT}}
\newcommand{\XP}{\cc{XP}}
\newcommand{\Weft}{{\cc{W}}}
\newcommand{\W}[1]{{\Weft}{{[#1]}}}
\newcommand{\paraNP}{\cc{pNP}}

\newcommand{\hy}{\hbox{-}\nobreak\hskip0pt}

\newcommand{\intP}[1]{\lfloor #1 \rfloor}
\newcommand{\ratP}[1]{{\rm frac}(#1)}

\def\inst{{\cal I}\xspace}

\def\D{{\bf D}\xspace}

\newcommand*{\var}[1]{{\rm var}_\chi(#1)\xspace}
\newcommand*{\con}[1]{{\rm con}_\chi(#1)\xspace}
\newcommand*{\scope}[1]{{\rm scope}(#1)\xspace}
\newcommand*{\cd}[1]{{\it CD}(#1)\xspace}
\newcommand*{\cdf}{{\it CD}\xspace}

\newcommand*{\num}[1]{{\rm num}(#1)\xspace}

\newcommand*{\tw}[1]{{\rm tw}(#1)\xspace}

\usepackage{boxedminipage}
\newcommand{\probfont}[1]{{\sc #1}\xspace}

\newcommand{\MPSS}{\probfont{MPSS}}

\def\SubsetSum{\probfont{Subset Sum}}

\newcommand{\pbDef}[3]{%
\noindent
\begin{center}
\begin{boxedminipage}{0.98 \columnwidth}
#1\\[5pt]
\begin{tabular}{l p{0.70 \columnwidth}}
Input: & #2\\
Question: & #3
\end{tabular}
\end{boxedminipage}
\end{center}
}

\newcommand{\pbDefP}[4]{%
\noindent
\begin{center}
\begin{boxedminipage}{0.98 \columnwidth}
#1\\[5pt]
\begin{tabular}{l p{0.70 \columnwidth}}
Input: & #2\\
Param.: & #3\\
Question: & #4
\end{tabular}
\end{boxedminipage}
\end{center}
}

\newcommand{\pg}{P}
\newcommand{\ig}{I}

\newcommand{\rva}{\alpha}
\newcommand{\rca}{\beta}

\newcommand{\rcd}{D_B}
\newcommand{\ria}{\tau}
\newcommand{\RRR}{\mathcal{R}}
\newcommand{\SSS}{\mathcal{S}}
\newcommand{\emptyass}{\emptyset}
\mathchardef\mhyphen="2D

\def\reals{\mathbb{R}}
\def\integers{\mathbb{Z}}
\def\rationals{\mathbb{Q}}
\def\naturals{\mathbb{N}}

\newcommand{\SB}{\{\,}
\newcommand{\SM}{\;{|}\;}
\newcommand{\SE}{\,\}}

\newcommand{\mvec}[1]{\bar{#1}}

\begin{document}

\title{Algorithms and Complexity of Difference Logic\thanks{This article improves and extends results from two conference papers~\cite{Dabrowski:etal:kr2020,Dabrowski:etal:aaai2021}.
For the purpose of open access, the authors have applied a Creative Commons Attribution (CC BY) licence to any Author Accepted Manuscript version arising.}}

%\title{Complexity of the Satisfiability\\ Problem for Difference %Logic}
\author{
Konrad K. Dabrowski\thanks{School of Computing, Newcastle University, UK, \texttt{konrad.dabrowski@newcastle.ac.uk}} \and
Peter Jonsson\thanks{Department of Computer and Information Science, Link{\"o}ping University, Sweden, \texttt{peter.jonsson@liu.se}} \and
Sebastian Ordyniak\thanks{School of Computing, University of Leeds, UK, \texttt{sordyniak@gmail.com}} \and
George Osipov\thanks{Department of Computer and Information Science, Link{\"o}ping University, Sweden, \texttt{george.osipov@pm.me}}
}
\date{\today}

\maketitle

\begin{abstract}
{\em Difference Logic} (DL) is a fragment of linear arithmetics
where atoms are constraints $x+k \leq y$ for variables $x,y$ (ranging over ${\mathbb Q}$ or ${\mathbb Z}$) and integer $k$.
We study the complexity of deciding the truth of existential DL sentences.
This problem appears in many contexts: examples include verification, bioinformatics, telecommunications, and spatio-temporal reasoning in AI.
We begin by considering
sentences in CNF with rational-valued variables. We restrict the allowed clauses via
two natural parameters: {\em arity} and {\em coefficient bounds}.
The 
problem is \NP-hard for most choices of these parameters. As a response to this, we refine our understanding by analyzing the time complexity and
the parameterized complexity 
(with respect to well-studied parameters such as
primal and incidence treewidth).
We obtain a comprehensive picture of the complexity landscape
in both cases. Finally, we generalize our results to integer domains and
sentences that are not in CNF. 
\end{abstract}
\bigskip

\noindent
{\bf Keywords:} Difference Logic, Algorithms and Complexity, Fine-grained Complexity, Parameterized Complexity, Treewidth.

\newpage

\tableofcontents

\newpage

\section{Introduction}
\label{sec:introduction}

We have divided this introductory section into four parts.
In the first one (Section~\ref{sec:background}), we present difference logic (DL)
and some of its applications, and we describe our
approach for studying the complexity of DL.
In short, the satisfiability problem for DL is
almost always \NP-hard and a more fine-grained analysis
becomes necessary; we will thus study the time complexity
of DL together with its parameterized complexity under
natural structural parameters.
Our time complexity results are discussed
in Section~\ref{sec:intro-time} while our parameterized results are
discussed in Section~\ref{sec:intro-parameterized}.
Finally, an outline of the article is given in
Section~\ref{sec:outline}.

\subsection{Background}
\label{sec:background}

{\em Difference Logic} is a fragment of linear arithmetics
where atoms are constraints of the form $x+k \leq y$ for variables $x,y$ (with some numeric domain such as ${\mathbb Q}$ or ${\mathbb Z}$) and some integer $k$.
The {\em satisfiability} problem 
for DL is the computational problem of deciding the truth
of sentences
\[\exists x_1,\dots,x_n . \phi\]
where $\phi$ is a quantifier-free formula over variable set
$\{x_1,\dots,x_n\}$.
The satisfiablity problem for
conjunctions of difference atoms
is solvable in polynomial time (by, for instance, the Floyd-Warshall algorithm), while adding various logical features often leads to computational hardness. We note, for instance, that the satisfiability problem is \NP-hard (since the satisfiability problem for propositional logic is \NP-hard~\cite{Cook:stoc71}) and that the problem of deciding the truth of an arbitrary formula is \PSPACE-hard
(since deciding the truth of quantified propositional formulas is \PSPACE-hard~\cite{Stockmeyer:Meyer:stoc73}).
These complexity results hold both for rational and integer variable domains.
DL is a well-studied formalism due to its many applications: the archetypal example is from verification where
timed automata have natural connections with
DL~\cite{Alur:cav99,Niebert:etal:ftrtft2002}.
Other
important applications include
the channel assignment problem
(which is a central problem in telecommunications~\cite{Audhya:etal:wcmc2011,Kral:dam2005}),
unit interval problems (with applications in bioinformatics
and graph theory~\cite{Golumbic:etal:aam94,pe1997satisfiability}),
and problems in connection with answer set programming~\cite{Lierler:Susman:tplp2017,Niemela:amai2008}---all of these
can be viewed as restricted variants of DL.
Applications like these and the relative simplicity of DL have made it into
one
of the most ubiquitous theories in the context
of {\em satisfiability modulo theories} (SMT)~\cite{Barrett:etal:SMT,Candeago:etal:sat2016,Nieuwenhuis:etal:jacm2006}.
DL is also interesting from a complexity-theoretic point of view. One example is
the {\em max-atom} problem (see the
paper by Bezem et al.~\cite{Bezem:etal:lpar2008} or Section 6 in~\cite{Bodirsky:Mamino:survey})
that can be viewed as a severely restricted version of DL.
This problem is polynomial-time equivalent to problems such as mean pay-off games,  scheduling under and-or precedence constraints, and finding solutions
to certain classes of equations.
The max-atom problem is intriguing
since it is known to be in \NP $\cap$ co\NP\ but no
polynomial-time algorithm has yet been identified.

DL is of major importance in AI but the connections
are in general not clearly spelled out in the literature.
Spatial-temporal reasoning is a fundamental task in AI and
one of the most influential formalisms is
the {\em simple temporal problem} (STP) that was
first
proposed in an AI context by Dechter et al.~\cite{Dechter:etal:ai91}.
It is a constraint satisfaction problem (CSP) over a constraint language 
with relations
$$\{ (x,y) \in \rationals^2  :  x - y \in [l,u] \}$$
where 
$\ell,u \in \rationals \cup \{-\infty, + \infty\}$
and $[\ell,u]$ denotes a closed interval.
We refer to constraints using such relations as \emph{simple constraints}.
The close relationship to DL is obvious.
The STP formalism is often generalized so that the
intervals may be half-closed, open, or a single point.
%Dechter et al.~\cite[Sec. 7]{Dechter:etal:ai91} point
%out that this makes the formalism easier to use
%in connection with other temporal formalisms such
%as the Point Algebra~\cite{Vilain:Kautz:aaai86} and Allen's %Algebra~\cite{allen1983maintaining}, and that
Dechter et al.~\cite[Sec. 7]{Dechter:etal:ai91} point
out that this generalization apparently do not have any adverse effects
and,
in particular, the resulting CSP is still solvable
in polynomial time. 
Even though STPs have proven to be immensely useful in AI, their expressive power is limited.
Thus, a common way of obtaining increased expressibility is to introduce
disjunctions in various
ways~\cite{barber2000reasoning,Dechter:etal:ai91,Oddi:Cesta:ecai2000,Stergiou:Koubarakis:ai2000}. From the DL perspective, this is equivalent to considering 
DL formulas on
conjunctive normal form and restricting the set of allowed 
clauses in various ways.
The resulting formalisms
are highly relevant in an AI context. Well-known examples can be found in
automated planning~\cite{Gerevini:etal:jair2006,BrentVenable:ijcai2005}
and multi-agent systems~\cite{Bhargava:Williams:aamas2019,Boerkoel:Durfee:aaai2013}. Stergiou \& Koubarakis~\cite[Sec. 7]{Stergiou:Koubarakis:ai2000},
Tsamardinos \& Pollack~\cite{Tsamardinos:Pollack:ai2003} and Peintner et al.~\cite{Peintner:etal:cp2007}
discuss various other applications, and
Zavatteri et al.~\cite{Zavatteri:etal:constraints2023}
have recently presented a
large-scale
evaluation of software for solving DTPs.

We traditionally view a computational problem as
intractable if it is \NP-hard. \NP-hardness rules out polynomial-time
algorithms (assuming \PP $\neq$ \NP), but it does not say anything
about the time complexity of the best possible algorithm.
Recent advances in complexity theory allow us to prove
conditional lower bounds via restricted reductions from
complexity-theoretic conjectures that are stronger than the
\PP $\neq$ \NP conjecture.
This methodology has
enabled proving close-to-optimal bounds on time
complexity for a multitude of problems
assuming suitable conjectures, cf.
the textbook by Gaspers~\cite{Gaspers:ETA}.
The goal of this article is to analyze the satisfiability problem for DL
following this methodology.
Our time complexity results reveal that many severely restricted variants of 
DL cannot be solved in a reasonable amount of time under
the Exponential-Time Hypothesis (ETH).
This computational hardness makes it worthwhile to
use
parameterized complexity
for analyzing DL with
restricted interactions between variables and constraints.
%We show that DL instances on conjunctive normal form where the coefficients
%are bounded by $poly(n)$
%can be solved in $n^{f(w)}$ time. Here, $n$ denotes the problem size,
%$w$ is the treewidth of the incidence graph and~$f$ is a computable function;
%in other words, this problem is in the complexity class \XP 
%and it can be solved in polynomial time whenever $w$ is fixed.
%We complement this result by showing that 
%DL
%restricted to the coefficients $0$ and~$1$ is not fixed-parameter tractable
%with respect to treewidth, 
%i.e. it does not admit a $f(w) \cdot poly(n)$ time algorithm for any %computable function $f$, under standard complexity assumptions.

We need some definitions and notation to facilitate the discussion of the problems that we will study. In the sequel, we restrict ourselves to the satisfiability problem for DL over rational numbers where the input is
in CNF, and we study various ways of restricting the allowed clauses. 
We will return to DL without these restrictions
in Section~\ref{sec:extensions}.
The restriction to CNF formulas allows us to view the satisfiability problem for DL as a constraint satisfaction
problem where the constraint language correspond to the allowed clauses.
Our clause restrictions will be based on
two parameters: {\em arity} and {\em coefficient bounds}. The arity bounds the number of distinct
variables that may appear in a clause.
It is closely connected to the {\em length}
of a clause, i.e. the maximum number of literals,
since
if a clause has length $k$, then its arity is at most $2k$.
% \todo{GO: is length used anywhere in the article? PJ: No, I do not think so. Length is often used in connection with clauses (think of $k$-SAT).}
The coefficient bound simply equals the
maximum over the absolute values of constants appearing
in clauses.

We continue by introducing the maximally expressive
constraint language $\D$. 
We consider intervals over $\rationals$ with endpoints in $\integers \cup \{-\infty,+\infty\}$.
The intervals may be open, closed, half-closed, or a single point.
Let ${\mathbb I}$ denote the set of these intervals and
let $\D$ contain all relations
\[
    \textstyle \{ ({x_1},\dots,{x_t}) \in \rationals^t  : \bigvee_{\ell=1}^{m} x_{i_\ell} - x_{j_\ell} \in I_{\ell} \}
\]
for arbitrary $t,m \geq 1$ where $i_\ell, j_\ell \in \range{1}{t}$ and $I_\ell \in {\mathbb I}$ 
for all $1 \leq \ell \leq m$. 
We remark that one may equally well use the reals instead of
the rationals as the underlying domain.
The CSP for $\D$ is known
 as the {\em disjunctive temporal problem} (DTP) in the AI literature.
It is easy to verify that CSP$(\D)$ is in \NP\ since the STP is solvable
in polynomial time.
Given a relation $R \in \D$, let $K(R)$ denote the set of numerical bounds appearing in $R$, e.g. for 
$$R= \{(x,y,z) \in \rationals^3 : (-\infty < x - y \leq 3) \lor (0 \leq x - z < 6) \}$$ we have 
$K(R) = \{3, 0, 6\}$.
If $X$ is a set of relations, then the definition of $K$ extends naturally: $K(X) = \bigcup_{R \in X} K(R)$.
Let ${\bf A} \subseteq {\bf D}$ and define $\num{{\bf A}} = \max \{ \abs{a} : a \in K({\bf A}) \}$, i.e. $\num{{\bf A}}$ is the
least
upper bound on absolute values of all numerical bounds appearing in the relations of ${\bf A}$.
%For an instance $\declareInstance$ of $\csp(\D)$, let $n = \abs{V}$ denote the %number of variables and let the numerical upper bound be $k = %\num{\constraints}$.
We 
 let~$\D_{a,k}$ (where $a,k \in {\mathbb N} \cup \{\infty\}$) 
 denote the class of relations of arity at most $a$ and with
 $\num{\D_{a,k}} \leq k$.

We illustrate the basic definitions with an example: consider 
{\em Allen's interval algebra}~\cite{allen1983maintaining} 
restricted so that the intervals are only allowed to have unit length.
This formalism (which is referred to as the {\em unit Allen algebra)} 
has, for example, applications in bioinformatics
and graph theory~\cite{Golumbic:etal:aam94,pe1997satisfiability}.
Given a closed interval $I$, we let
$I^-$ and $I^+$ denote 
the left and the right endpoint, respectively.
We let $\unitallen$ denote a binary structure based on
the following relations:
\begin{alignat*}{3}
    &I \{p\} J &&\qquad I \text{ precedes } 
    J &&\qquad I^+ < J^- \\
    &I \{m\} J &&\qquad I \text{ meets } 
    J &&\qquad I^+ = J^- \\
    &I \{o\} J &&\qquad I \text{ overlaps } 
    J &&\qquad I^- < J^- \text{ and } J^- < I^+ \text{ and } I^+ < J^+ \\
    &I \{e\} J &&\qquad I \text{ equals } 
    J &&\qquad I^- = J^- \text{ and } I^+ = J^+
\end{alignat*}
Note that relations $p, m, o$ admit converses 
$p^{-1}, m^{-1}, o^{-1}$ while the relation $e$ is symmetric.
We let the structure
$\unitallen$ contain every disjunction of the basic relations.
Formally, let $\mathbb{U}$ denote the set of all unit intervals 
on the real line.
$\unitallen$ contains
$ \{ (I,J) \in \mathbb{U}^2 : \bigvee_{r \in S} I \{r\} J \} $
for every $S \subseteq \{p,m,o,e,o^{-1},m^{-1},p^{-1}\}$.
%$\csp{(\unitallen)}$ is NP-complete~\cite{pe1997satisfiability}.
Observe that every basic relation in 
the unit Allen algebra can be expressed 
as a simple relation in $\disjtemp{2,1}$ over 
the left endpoints of the intervals, i.e.
\begin{align*}
    I \{p\} J &\quad \iff 
    \quad I^- - J^- \in (-\infty, -1),  \\
    I \{m\} J &\quad \iff 
    \quad I^- - J^- \in \{-1\},  \\
    I \{o\} J &\quad \iff 
    \quad I^- - J^- \in (-1, 0),  \\
    I \{e\} J &\quad \iff 
    \quad I^- - J^- \in \{0\},
\end{align*}
and similarly for the converse relations.
Moreover, every simple relation in 
$\disjtemp{2,1}$ can be expressed as 
a basic relation of the unit Allen algebra
since the correspondence is one to one.
This reasoning naturally extends to
taking disjunctions of simple/basic relations.  
Thus, $\csp(\unitallen)$ and $\csp(\disjtemp{2,1})$
are the same computational problem,
and any upper/lower bound that applies to
one of the problems also applied to the other.

Let us now summarize the computational complexity of $\csp(\D_{a,k})$.
The polynomial-time solvability of $\csp(\D_{2,0})$ follows from the fact that
the relations in $\D_{2,0}$ equal the
point algebra~\cite{Vilain:Kautz:aaai86}. 
It is well known that $\csp(\D_{k,0})$ for $k \geq 3$ is \NP-hard 
(this follows,
for instance, from an easy reduction from the \probfont{Betweenness} problem~\cite{gj79}). Finally, $\csp(\D_{2,1})$ (and thus $\csp(\unitallen)$) are \NP-hard via a straightforward
reduction from \probfont{3-Colourability}; \NP-hardness for $\csp(\D_{2,k})$, $k > 1$, is a direct
consequence.
These results are presented
in Table~\ref{tb:computational-summary}---we immediately see that there is an
conspicuous lack of polynomial-time solvable cases.
In the rest of this article, we will refine our understanding of the complexity of $\csp(\D_{a,k})$ by first analyzing its time complexity
and continue with its parameterized complexity.
We discuss these results in Sections~\ref{sec:intro-time} and~\ref{sec:intro-parameterized}, respectively.

\begin{table} \centering
\begin{tabular}{|c|c|c|c|} \hline
       & $k=0$               & $1 \leq k < \infty$         & $k$ unbounded \\ \hline

$a = 2$ & $\in \PP$   &  \NP-complete  & \NP-complete    \\ \hline
$a \geq 3$ & \NP-complete        &  \NP-complete   & \NP-complete  \\ \hline
\end{tabular}

\caption{Summary of computational complexity landscape for $\csp(\D_{a,k})$.}
\label{tb:computational-summary}
\end{table}

\subsection{Time Complexity}
\label{sec:intro-time}

%We use $\disjtemp{2}^{\leqslant}$ and %$\disjtemp{\omega}^{\leqslant}$ for the constraint languages %that never use strict inequalities.
We prove the following results concerning the time complexity of DTPs.
Our lower bounds are based on the {\em Exponential Time Hypothesis} (ETH)
by Impagliazzo et al.~\cite{impagliazzo2001problems}, i.e.
the {\sc 3-Satisfiability} problem cannot be solved
in $2^{o(n)}$ time, where $n$ is the number of variables.
We let $\D_{a,k}^{\leq}$ denote the subset of $\D_{a,k}$
where the relations are defined by only using closed intervals.
\begin{enumerate}
\item
  $\csp(\D)$ is solvable in $2^{O(n(\log{n}+\log{k}))}$ time (Corollary~\ref{thm:generalupperbound}).
\item
  $\csp(\D_{2,k})$ is solvable in $2^{O(n \log \log n)}$ time (Theorem~\ref{thm:gammaonetime}).
\item
  $\csp(\D_{4,0})$ and $\csp(\D^{\leq}_{3,1})$ are not solvable in $2^{o(n \log n)}$ time
  (Theorems~\ref{thm:omegakresult} and~\ref{thm:D3,1-hardness}).
\item
  $\csp(\D^{\leq}_{2,\infty})$ is not solvable in $2^{o(n (\log n + \log k))}$ time
  (Theorem~\ref{thm:binary-hardness}).
\item
  For every $c>1$, there exist  $k \geq 0$ and 
  ${\bf A} \subseteq \D^{\leq}_{2,k}$ such that $\csp({\bf A})$
  cannot be solved in $O(c^n)$ time (Theorem~\ref{thm:d2klowerbound}).
\end{enumerate}

We additionally use a result by 
Eriksson and Lagerkvist~\cite[Section~3]{Eriksson:Lagerkvist:ijcai2021}.

\begin{theorem}[\cite{Eriksson:Lagerkvist:ijcai2021}] \label{thm:D30}
$\csp(\D_{3,0})$ is solvable in $2^{O(n)}$
but not in $2^{o(n)}$ time (if the ETH is true).
\end{theorem}

The results are summarized in Table~\ref{tb:time-summary} and we see that the upper and lower bounds
are reasonably close.
The lower bounds hold for constraint languages that do not use strict inequalities except for
$\csp(\D_{a,0})$, $a \geq 2$; an instance
of $\csp(\D^{\leq}_{a,0})$ is always satisfiable by assigning each variable  value 0.
The results concerning $\csp(\D_{2,k})$ indicates that there is no uniform single-exponential
algorithm for $\csp(\disjtemp{2,k})$. 
The result does not, however, rule out the possibility that 
$\csp(\disjtemp{2,k})$ can be solved in $2^{c_k \cdot n}$ time, 
where $c_1,c_2,\ldots$ is an increasing sequence.
All results in Table~\ref{tb:time-summary} remain intact if we restrict the variables to take integer values only (see
Section~\ref{sec:integerdomains}).
We remark that our main goal is in delineating single-exponential vs super-exponential
running times, for which the ETH is a reasonable starting point.
To obtain more fine-grained lower bounds, e.g. 
rule out concrete constants in the bases of exponential functions,
one typically needs to rely on stronger hypotheses like the {\em strong} ETH~\cite{Calabro:etal:iwpec2009}.

\smallskip

Our algorithm for $\csp(\D)$ is based on proving
a {\em small solution property}: every satisfiable instance of $\csp(\D)$
has a solution that assigns sufficiently small
values to the variables. 
The small solution property is not very common in infinite-domain CSPs but it is, for instance, known to hold for 
the max-atom problem \cite{Bezem:etal:lpar2008} and 
the CSP problem for {\em unit two variables per inequality} relations~\cite{Seshia:etal:jsbmc2007}.
Our proof utilizes certain ordering properties inherent in $\D$ together with a method for handling the integer and fractional part of the variables independently; this approach is distinctly different compared to the proof
    techniques used in~\cite{Bezem:etal:lpar2008} and~\cite{Seshia:etal:jsbmc2007}.
With the aid of this
result, we can enumerate a suitable collection of assignments and check whether at least one of them satisfies all constraints in the instance.
The small solution property will be important once again when we consider the parameterized setting (see Section~\ref{sec:intro-parameterized}).
Our algorithm for $\csp(\D_{2,k})$ is based on a non-trivial
divide-and-conquer approach. The relations in
$\D_{2,k}$ exhibit even stronger ordering properties
than the relations in $\D$ and this allows us to
show that any solution for an instance $\instance$ of
$\csp(\disjtemp{2,k})$ suggests a natural split of the whole instance
into either two or three subinstances sharing only a small number of
variables. Hence, our algorithm enumerates all
possible decompositions into two or three subinstances with small variable overlap and recurses on
those for every possible assignment of the shared variables. An immediate consequence of this
algorithm is the following result 
(since $\csp(\unitallen)$ and $\csp(\disjtemp{2,1})$
are the same computational problem).

\begin{proposition}
    \label{prop:unitallen-time-ub}
    $\csp{(\unitallen)}$ is solvable in $2^{O(n \log \log n)}$ time.
\end{proposition}

\smallskip

Our lower bounds are based on a mixture of related ideas.
We exploit the lower bound on the $(k \times k)$-\probfont{Independent Set}
problem by Lokshtanov,~Marx~and~Saurabh~\cite{lokshtanov2018slightly}, and 
the lower bound by Traxler~\cite{traxler2008time}. 
The latter result concerns binary CSPs over finite domains,
where the complexity is measured with respect to 
the number of variables.
Intuitively, Traxler shows that, under the ETH, 
the complexity of binary CSPs grows 
together with the domain size.
We illustrate the main technical idea by an example.
Suppose an instance of a CSP over the domain $\{1,2,3\}$ 
has two variables $v_1, v_2$ and two unary constraints:
$v_1 \in \{1,2\}$ and
$v_2 \in \{2,3\}$.
One can reduce it to a CSP over the domain
$\{1,2,3\}^2$
with a single constraint
$v \in \{1,2\} \times \{2,3\}$, where
$\{1,2\} \times \{2,3\} = \{ (1,2), (1,3), (2,2), (2,3) \}$.
Here variable $v$ encodes the pair of variables $(v_1, v_2)$.
Applying the same idea, one can reduce any
instance of binary CSP over domain $d$ with $n$ variables
to a binary CSP over domain $d^r$ with roughly $n/r$ variables,
for any constant $r$.
Thus, with increased domain size, 
the number of variables required to 
encode the same set of constraints decreases.
Lokshtanov~et~al. push this idea to the limit,
where the domain size and the number of variables 
are roughly equal.
A helpful technical tool that we use 
in adapting these results to $\csp(\D_{a,k})$
are {\em Sidon sets}. 
A set $S$ of natural numbers is called a Sidon set if 
all pairwise sums of its elements are distinct, i.e. the equation 
$a+b=c+d$ with $a,b,c,d \in S$ is only solvable when $\{a,b\}=\{c,d\}$. 
Sidon sets are also used 
in our lower bound proofs in the parameterized case.

% \todo[inline]{PJ: Improve and expand paragraph on lower bounds.}

% \todo[inline]{PJ: Our lower bound result only implies that
% unit Allen is not
% solvable in subexponential time.}

\begin{table} \centering
\begin{tabular}{|c|c|c|c|} \hline
{\bf Upper bounds}        & $k=0$               & $1 \leq k < \infty$         & $k$ unbounded \\ \hline

$a = 2$ & $\in \PP$   &  $2^{O(n \log \log n)}$  & $2^{O(n(\log n + \log k))}$   \\ \hline
$a = 3$ & $2^{O(n)}$ & $2^{O(n \log n)}$  & $2^{O(n(\log n + \log k))}$ \\ \hline
$a \geq 4$ & $2^{O(n \log n)}$    &  $2^{O(n \log n)}$    & $2^{O(n(\log n + \log k))}$ \\ \hline \hline

{\bf Lower bounds}        & $k=0$               & $1 \leq k < \infty$         & $k$ unbounded \\ \hline

$a = 2$ & $-$            &  $(*)$    & $2^{o(n (\log n + \log k))}$  \\ \hline
$a = 3$ &   $2^{o(n)}$          &    $2^{o(n \log{n} )}$     &  $2^{o(n (\log n + \log k))}$   \\ \hline
$a \geq 4$ &  $2^{o(n \log n)}$      &  $2^{o(n \log n )}$     & $2^{o(n (\log n + \log k))}$ \\ \hline
\end{tabular}

\caption{Summary of time complexity landscape for $\csp(\D_{a,k})$. $(*)$ means that for every $c>1$, there exists $k \geq 0$ and
    ${\bf A} \subseteq \disjtemp{2,k}$ such that $\csp({\bf A})$
    cannot be solved in $O(c^n)$ time.}
\label{tb:time-summary}
\end{table}

We conclude this section with a few words about related problems
from the literature.
The $\csp(\D)$ problem can be expressed in the {\em existential theory of the reals} ($\exists \reals$).
An $\exists \reals$-formula is a Boolean combination of atomic
predicates of the form $p(x_1, \dots, x_n) \odot 0$,
where $p$ is a real polynomial and $\odot \in \{<,\leq,=,\geq,>,\neq\}$.
Renegar's algorithm~\cite{renegar1992computational}
decides the satisfiability problem for
$\exists \reals$-formulas
in $L \log L \log \log L \cdot (md)^{O(n)}$ time
where
$L$ is the number of bits needed to represent the coefficients in the polynomials, 
$m$ is the number of polynomials in the sentence, 
$d$ is maximum among total degrees of the polynomials, and 
$n$ is the number of variables.
Observe that instances of $\csp(\D)$ can be written as
$\exists \reals$-formulas by replacing atomic formulas of the form
$x - y \leq a$ with $p(x,y) \leq 0$ where $p(x,y) = x - y - a$.
An instance ${\cal I}$ of $\csp(\D)$ with
$n$ variables and $k = \num{\cal I}$
can have $O(n^2k)$ atomic formulas:
there are $\binom{n}{2}$ pairs of variables
and $O(k)$ possible bounds can be expressed on their difference.
We are allowed to use disjunctions, which can be applied to
an arbitrary subset of the $O(n^2k)$ atomic formulas.
Thus, cast as a $\exists \reals$-formula,
$I$ has $m \leq 2^{O(n^2k)}$ polynomials of degree $d = 1$.
This leads to a $2^{O(n^3k)}$-time algorithm for $\csp(\D)$
and, consequently, $2^{O(n^3)}$ time for $\csp(\disjtemp{\infty,k})$. 
For binary constraint languages, Renegar's algorithm yields
better results with $2^{O(n (\log n + \log k))}$ time for $\csp(\disjtemp{2,\infty})$ 
and $2^{O(n \log n)}$ time for $\csp(\disjtemp{2,k})$
because only $O(n^2 k)$ disjunctive formulas are available.
In fact, the running time for $\csp(\disjtemp{2,\infty})$ obtained this way matches our result. 
However, we claim that our algorithm represents 
a very simple and natural approach to solving this problem.
While asymptotically the result are the same,
Renegar's algorithm solves a much more general problem,
and the hidden constants in its running time are astronomical
(see e.g. the practical evaluation in~\cite{hong1991comparison}).

Let us turn our attention to lower bounds.
Soca{\l}a~\cite{socala2016tight} shows that the {\em channel assignment} problem cannot be solved
in $2^{o(n \log n)}$ time under the ETH. This problem can be viewed as CSP$({\bf A}_{\rm ca})$
where ${\bf A}_{\rm ca}$ contains the relation $\{(x,y) \in {\mathbb N}^2 :  |x-y| \geq a\}$ for every $a \in {\mathbb N}$. This result implies that $\csp(\D^{\leq}_{2,\infty})$ is not solvable in
$2^{o(n \log n)}$ time
but it does not directly imply our stronger
$2^{o(n (\log n + \log k))}$ lower bound for integer solutions (that is derived
by combining Theorem~\ref{thm:binary-hardness} and Lemma~\ref{lem:integer-ak}). 
Unit two variables per inequality (UTVPI) relations
are
defined as
$\{(x,y) \in \integers^2 : ax+by \geq c\}$
where $a,b \in \{-1,0,1\}$ and $c \in \integers$.
This is a well-studied and interesting generalization of
$\csp(\D^{\leq}_{2,\infty})$ over the integers;
Schutt and Stuckey
write the following~\cite[p. 514]{Schutt:Stuckey:informs2010}.
\begin{quotation}
Unit two-variable-per-inequality (UTVPI) constraints form one of the largest class of integer
constraints which are polynomial time solvable (unless \PP\ = \NP). There is considerable interest
in their use for constraint solving, abstract interpretation, spatial databases, and theorem
proving.
\end{quotation}
Seshia et al.~\cite{Seshia:etal:jsbmc2007} have presented an algorithm
for
checking the satisfiability of first-order formulas without universal quantification
over UTVPI constraints.
This algorithm runs in $2^{O(n(\log n+\log k))}$
time (where, as usual, $n$ is the number of variables and $k$ is the coefficient bound).
Our lower bound result for $\csp(\D^{\leq}_{2,\infty})$ over the integers
shows that this algorithm is
essentially optimal with respect to running time.

\subsection{Parameterized Complexity}
\label{sec:intro-parameterized}

We have seen
that $\csp(\D)$
(and many severely restricted variants) cannot be solved in single-exponential 
% \todo{GO: or even single exponential? PJ: I changed the sentence.} 
time  under
the Exponential-Time Hypothesis (ETH).
This motivates the search for efficiently solvable subproblems.
To this end, we use the framework of 
{\em parameterized complexity}~\cite{DowneyFellows13,book/FlumG06,book/Niedermeier06},
where the run-time of an algorithm is studied with respect to a parameter
$p\in\Nat$ and the input size~$n$.
The idea is that the parameter describes the structure of
the instance in a computationally meaningful way.
Here, the most favorable complexity class is \FPT
(\emph{fixed-parameter tractable}),
which contains all problems that can be decided 
in $f(p)\cdot n^{O(1)}$ time, where $f$ is a computable
function.
The next best option is the
complexity class \XP, which contains all problems decidable
in $n^{f(p)}$ time, i.e. the problems solvable in polynomial time
when the parameter $p$ is bounded.
Clearly, $\FPT \subseteq \XP$ and
this inclusion is strict 
(see e.g.~\cite[Cor. 2.26]{book/FlumG06}).
It is significantly 
better if a problem is in \FPT than in \XP
since the order of the polynomial factor in the former case does not depend on the parameter $p$.
Finally, the class \paraNP\ contains all problems that can be decided 
in $f(p)\cdot n^{O(1)}$ time by a non-deterministic algorithm for some computable function $f$.
It is known that
a problem is \paraNP-hard (under {\em fpt-reductions}; see Sec~\ref{sec:lb}) if it is \NP-hard for some constant value of the parameter. 
Problems that are \paraNP-hard are considered to be
significantly harder than those in \XP since a problem that is \paraNP-hard cannot be in \XP unless \PP = \NP.

A prominent method for identifying tractable fragments of CSPs 
is to restrict variable-constraint interactions (see, for instance, the
survey by Carbonnel and Cooper~\cite[Sec.~5]{Carbonnel:Cooper:constraints2016}); 
these are referred to as {\em structural restrictions}
and are commonly studied via the primal and incidence graphs associated with instances of the CSP.
The {\em primal graph} has the variables as its vertices 
with any two joined by an edge if they occur together 
in a constraint. 
The {\em incidence graph} is the bipartite graph with 
two disjoint sets of vertices corresponding to the variables and the constraints, respectively.
A constraint vertex and a variable vertex are joined by an edge if the variable occurs in the scope of the constraint.
The {\em treewidth} of such graphs has been used extensively. 
It is, for example, known
that the finite-domain CSP is in \FPT with the parameter $w+d$
if $w$ is the primal treewidth and~$d$ is the domain size~\cite{Gottlob:etal:ai2002}, 
while this is not true
(under standard complexity assumptions) if 
$w$ is the incidence
treewidth~\cite{Samer:Szeider:jcss2010}.

We now describe
our parameterized results.
It is known that
the primal treewidth is bounded
from below by the incidence treewidth~\cite{Kolaitis:Vardi:jcss2000} for
arbitrary CSP instances.
Thus, we present algorithms for $\csp(\D_{a,k})$ parameterized by incidence treewidth
and lower bounds with respect to primal treewidth.
We exhibit an \XP algorithm
for CSP$(\D_{\infty,k})$ when $k \in {\mathbb N}$. This is a
bottom-up dynamic programming algorithm along a nice
tree-decomposition of the incidence graph that exploits
the fact that CSP$(\D)$ has the small solution property. The algorithm runs
in time $(nk)^{O(w)}$ where $w$ is the treewidth of the
incidence graph. 
% \todo{GO: changed omega to w to avoid conflicts with omega categorical. Was there a reason for omega as opposed to w? PJ: I think some of the basic definitions were copied from some other paper.}
One may note that
CSP$(\D)$ is in \XP\ whenever the numeric values occurring in the instance
are bounded by a polynomial in the number of variables.

We complement this algorithmic result by proving that $\csp(\D_{2,k})$ for $1 \leq k < \infty$
is \Weft[1]-hard when parameterized by primal treewidth and thus not in \FPT under standard complexity-theoretic assumptions. This 
shows that significantly faster algorithms for 
$\csp(\D_{a,k})$ with $k < \infty$
are unlikely. 
This \Weft[1]-hardness result carries over to
$\csp(\D^{\leq}_{a,k})$ when $k \geq 1$ almost without extra effort; note that the condition on $k$ is required since
$\csp(\D^{\leq}_{a,0})$ is trivially in $\PP$.
The reduction is from a novel multi-dimensional variant
of the well-known \textsc{Subset Sum} problem, which we show to be
\Weft[1]-hard.
Many important problems from the AI literature such as Allen's Algebra and RCC8 are in \FPT~\cite{Dabrowski:etal:ai2023}
so even $\csp(\D_{2,1})$ is a substantially harder problem.
We finally show that $\csp(\D_{2,\infty})$ is \paraNP-hard, i.e. the
problem becomes much harder when the numeric values are unbounded.
If a language $L$ is in \NP, then all parameterized languages 
$L' \subseteq L \times {\mathbb N}$ are
members of \paraNP so $\csp(\D_{2,\infty})$ is a \paraNP-complete problem.
We summarize our results in Table~\ref{tb:summary}. 
All results for $k \geq 1$ can be found in this article,
while the result for $k=0$ was proven
by Dabrowski et al.~\cite{Dabrowski:etal:ai2023}.
We note that the results still hold if we restrict ourselves to integer 
variable domains (see Corollary~\ref{cor:w1-hard} and Section~\ref{sec:integerdomains}).
The results outlined above immediately implies the following
since $\csp(\unitallen)$ and $\csp(\disjtemp{2,1})$
are the same computational problem.

\begin{proposition}
    \label{prop:unitallen-param}
    $\csp{(\unitallen)}$ with parameter treewidth of incidence graph
    is in \XP and it is \W{1}-hard with parameter treewidth of primal graph.
\end{proposition}

We conclude this section by discussing some related algorithms from the literature.
Bodirsky \& Dalmau~\cite{Bodirsky:Dalmau:jcss2013} 
and Huang et al.~\cite{Huang:etal:ai2013} 
proved that $\csp({\bf A})$ is in \XP (with treewidth of
the primal graph as parameter)
for $\omega$-categorical ${\bf A}$ and
binary constraint languages ${\bf A}$ 
that have the {\em atomic network amalgamation property}~(aNAP),
respectively.
Huang et al. write that their algorithm is fixed-parameter tractable,
but this is due to non-standard terminology; according to their Theorem 6,
the algorithm runs in
$O(w^3n \cdot {\rm e}^{w^2 \log n}) = n^{O(w^2)}$ time.
These two general results apply to many interesting
problems: $\omega$-categoricity is a fundamental
property in the study of infinite-domain CSPs
and many AI-relevant CSPs have this property
(cf. the book by Bodirsky~\cite{Bodirsky:book}).
Similarly, the aNAP and other amalgamation properties
are highly important in this context, too.
However, these properties do not hold for 
the  constraint language $\D$ 
or even the fragment $\D_{2,1}$, as we will show next.

The theorem by Engeler, Ryll-Nardzewski, and 
Svenonius~(see e.g.~\cite[Theorem 6.3.1]{Hodges:1997:SMT:262326})
implies that if ${\bf A}$ is an $\omega$-categorical constraint language,
then for all $n > 1$, there are finitely many nonequivalent 
formulas over ${\bf A}$ with $n$ free variables.
This is not true for $\D_{2,1}$: 
consider the infinite sequence of formulas $\phi_2(x,y),\phi_3(x,y),\dots$
defined as follows:
\[\phi_k(x,y) \equiv \exists z_1,\dots,z_k. \; x = z_1 \land y = z_k \land 
\bigwedge_{i=1}^{k-1} z_{i+1} - z_{i} = 1 \]
and note that $\phi_k(x,y)$ holds if and only if $y=x+k-1$.
If a structure ${\bf A}$ containing binary relations has aNAP, then 
for any pair of complete atomic
instances $(V_1, C_1)$ and $(V_2, C_2)$ 
of $\csp({\bf A})$ 
that have the same constraints
over the variables in $V_1 \cap V_2$,
their union $(V_1 \cup V_2, C_1 \cup C_2)$ is satisfiable.
An instance of~$\csp({\bf A})$ is {\em complete} if there is one constraint
for every pair of variables, and it is {\em atomic}
if no constraints involve disjunctions. 
Consider the instances
%KD: Changed from align* environment to avoid overfull hbox.
\begin{align*}
  \inst_1 &= (\{ x,a,y \}, \{ a - x = 1, y - a = 1, y - x \in (1,\infty) \}), \\
  \inst_2 &= (\{ x,b,y \}, \{ b - x = 1, y - b \in (0,1), y - x \in (1,\infty) \}).  
\end{align*}
$\inst_1$ and $\inst_2$ are complete, satisfiable, atomic instances of $\csp(\D_{2,1})$, and
they agree on their intersection.
However, their union is not satisfiable, 
since $\inst_1$ implies that $y - x = 2$, 
while $\inst_2$ implies that $y - x \in (1,2)$.

Dabrowski et al.~\cite{Dabrowski:etal:ai2023} have presented
a fixed-parameter tractable algorithm
for constraint languages having the {\em patchwork}
property~\cite{Lutz:Milicic:jar2007}; this is yet another
amalgamation property.
The applicability of this
algorithm can naturally be ruled out with the
aid of the hardness results presented in Section~\ref{sec:lb}.
It is also straightforward to verify directly
that the problems we study do not have
the patchwork property: in fact, the example
above for ruling out that $\D_{2,1}$ has
aNAP also shows that $\D_{2,1}$
does not have the patchwork property.

%We conclude that alternative methods are required
%to study parameterized complexity of $\csp(\D)$.

\begin{table} \centering
\begin{tabular}{|c|c|c|c|} \hline
 {\bf Upper bounds}       & $k=0$               & $1 \leq k < \infty$         & $k$ unbounded \\ \hline

$a = 2$ & $\in \PP $            &  $\in \XP$    & $\in \paraNP$   \\ \hline
$a \geq 3$ & $\in \FPT $          &  $\in \XP$    & $\in \paraNP$    \\ \hline
\end{tabular}

\bigskip

\begin{tabular}{|c|c|c|c|} \hline
 {\bf Lower bounds}       & $k=0$               & $1 \leq k < \infty$         & $k$ unbounded \\ \hline

$a = 2$ & $ - $            &   \Weft[1]-hard   & \paraNP-hard    \\ \hline
$a \geq 3$ & $ - $          &   \Weft[1]-hard    & \paraNP-hard    \\ \hline
\end{tabular}

\caption{Summary of parameterized complexity landscape for $\csp(\D_{a,k})$}
\label{tb:summary}
\end{table}

\subsection{Outline}
\label{sec:outline}

This article is based on two conference papers~\cite{Dabrowski:etal:kr2020,Dabrowski:etal:aaai2021}.
The major differences are that (1) this article
generalizes our earlier results on various temporal
formalisms to difference logic, (2) it gives a comprehensive
picture of the time complexity landscape, (3) the
proofs are both unified and significantly simplified by the addition
of multi-purpose results such as
Theorem~\ref{lem:compact-assign}, and (4) the results are extended to both general formulas
and variables with integer domains.
The article has
the following structure.
We present the necessary preliminaries in Section~\ref{sec:prelims}.
The upper and lower bounds on time complexity are collected in
Sections~\ref{sec:upper-bounds-time} and \ref{sec:lower-bounds-time}, respectively, while the parameterized upper and lower bounds
are collected in Sections~\ref{sec:ub} and \ref{sec:lb}, respectively.
We look at two generalizations of our results in
Section~\ref{sec:extensions}: formulas that are not
in conjunctive normal form are considered in Section~\ref{sec:generalformulas}
and problems where variables have integer domains in 
Section~\ref{sec:integerdomains}.
We conclude the article in Section~\ref{sec:discussion} with a discussion of
our results.

\section{Preliminaries}
\label{sec:prelims}

In this section we provide some prerequisites.
We present the basic language of difference logic in Section~\ref{sec:difflogic}
and give a compact overview of the constraint satisfaction problem
in Section~\ref{sec:csp}. Finally, Section~\ref{sec:sidon} contains a primer on Sidon sets
that we use as a tool for proving our lower bound results.

\subsection{Difference Logic}  
\label{sec:difflogic}

We begin with some basic logical terminology.
A {\em (relational) signature} $\tau$ is a set of symbols, 
each with an associated natural number called their {\em arity}. 
A {\em (relational) $\tau$-structure} $\A$ consists of
a set $D$ (the domain), together with relations $R^\A \subseteq D^k$ 
for each $k$-ary symbol $R \in \tau$. 
%A structure is {\em countable} if its domain is a countable set.
To avoid overly complex notation, we sometimes do not distinguish between
the symbol $R$ for a relation and the relation $R^\A$ itself.
We also allow ourselves to view relational structures as sets and, for instance,
write expressions like $R \in \A$.
Let $\A$ be a $\tau$-structure over a domain $D$.
We say that $\A$ has arity $a$ if every relation in $\A$ has arity at most $a$.

Let $\A$ be a $\tau$-structure.
First-order formulas $\phi$ over $\A$ (or, for short, $\A$-formulas)
are defined using the logical symbols of universal and existential
quantification, disjunction, conjunction, negation,
equality, bracketing, variable symbols, the relation symbols from $\tau$, and
the symbol $\bot$ for the truth-value false. 
First-order formulas over $\A$ can be used to define relations: 
for a formula $\phi(x_1,\ldots,x_k)$ 
with free variables $x_1,\ldots,x_k$, the corresponding relation $R$
is the set of all $k$-tuples $(t_1,\ldots,t_k) \in D^k$
such that $\phi(t_1,\ldots,t_k)$ is true in $\A$. 
In this case we say that $R$ is {\em first-order definable} in $\A$.
Our definitions of relations are always parameter-free, i.e. we do not allow
the use of domain elements within them. 

Certain types
of first-order formulas are particularly interesting for our
purposes. Let $\phi$ denote a first-order formula.

\begin{itemize}
\item
$\phi$ is a {\em sentence} if it has no free variables.

\item
$\phi$ is in \emph{conjunctive normal form} (CNF), if it is a a conjunction of disjunctions of
{\em literals}, i.e., atomic formulas or their negations. A disjunction of literals is called a {\em clause}. 

\item
$\phi$ is {\em quantifier-free} if it does not contain existential and/or universal quantifiers.

\item
$\phi$ is {\em existential} if $\phi=\exists x_1,\dots,x_n . \psi$
where $\psi$ is quantifier-free.
\end{itemize}
We let $\stpfull$ denote the relational structure representing
the atomic DL formulas, i.e. the infinite set
of relations
$$ \{ (x,y) \in \rationals^2 : \ell \compare{1} \: x - y \: \compare{2} u \} $$
for any $\ell \in \integers \cup \{-\infty\}$, $u \in \integers \cup \{\infty\}$ and $\compare{1}, \compare{2} \in \{<,\leq\}$. We will sometimes
consider a restricted set ${\bf S}^{\leq}$ where $\compare{1} = \compare{2} = \; \leq$.
The satisfiability problem for DL is the following problem.

\pbDef{\probfont{DL-Sat}}
{An existential first-order sentence $\phi$ over $\stpfull$.}
{Is $\phi$ true?}

 Note that we make (without loss of generality) the sensible assumption that the bounding values are integers (see e.g. the article by Tsamardinos \& Pollack~\cite{Tsamardinos:Pollack:ai2003}):
 real values cannot in general be written down with a finite number of bits, and rational numbers can be scaled
 in a suitable way.
We use the rationals as the value domain (also without loss of generality): if there is a solution to an instance of $\csp(\D)$ over
the reals, then there is also a solution over the rationals. 
While this is not of major
importance in this article, the differences between $\reals$ and $\rationals$ sometimes causes
confusion and/or technical problems.
We refer the reader to the literature for a more thorough discussion of representational
issues~\cite{Bodirsky:Jonsson:jair2017,Jonsson:Loow:ai2013}.

\subsection{Constraint Satisfaction}
\label{sec:csp}

We continue by defining the {\em constraint satisfaction problem} (\csp).
Let ${\bf A}$ denote a relational $\tau$-structure defined on a set $D$
  of values.
  The {\em constraint satisfaction problem} over ${\bf A}$ ($\csp({\bf A})$)
  is defined as follows:

\pbDef{\csp({\bf A})}
{A tuple $(\variables,\constraints)$, where $\variables$ is a set
      of variables and $\constraints$ is a set of constraints of the form
      $R(v_1, \dots, v_a)$, where $a$ is the
      arity of $R$, $v_1, \dots , v_a \in V$, and $R \in {\bf A}$.}
{Is there a function $f : V \rightarrow D$
          such that $(f(v_1), \dots , f(v_a)) \in R$ for every
          $R(v_1, \dots , v_a) \in \constraints$?} 
          
Observe that we do not require ${\bf A}$ to have finite signature or $D$ to be a finite set. 
The structure ${\bf A}$ is sometimes referred to as a {\em constraint language}, while
the function $f$ is a {\em satisfying assignment} or simply a {\em solution}.
If $c=R(x_1,\dots,x_a)$ is a constraint, then the set $\{x_1,\dots,x_a\}$
is the {\em scope} of $c$. We denote this set by $\scope{c}$.
A basic example of a CSP is 
the STP problem: it is easy to verify that it equals $\csp{(\bf S)}$.
Another example is the
max-atoms problem: it is 
conveniently defined as a CSP with the infinite constraint language ${\bf A}_{\max}$
containing the relations
$R_a = \{(x,y,z) \in {\mathbb Q}^3 : \max(x,y)+d \geq z\}$ for every $d \geq 0$.
We note that $R_d$ is quantifier-free definable in ${\bf S}$ since
\[\max(x,y)+d \geq z \Leftrightarrow (x+d \geq z) \vee (y+d \geq z).\]

One may view $\csp({\bf A})$ with ${\bf A} \subseteq {\bf D}$ as
a restricted \probfont{DL-Sat} problem. An {\bf A}-sentence
is {\em primitive positive} if it is of the form

\[\exists x_1,\dots,x_n . \psi_1 \wedge \dots \psi_l\]
where $\psi_1,\dots,\psi_l$ are atomic formulas over ${\bf A}$, i.e.
formulas (1) $R(y_1,\dots,y_a)$ with $R \in {\bf A}$, 
(2) $y_i=y_j$, or (3) $\bot$. 
Thus, $\csp({\bf A})$ can be viewed as \probfont{DL-Sat}
restricted to primitive positive ${\bf A}$-formulas
whenever the equality relation is in ${\bf A}$.
This assumption is harmless for the CSP problem: 
adding equality to the constraint language does not affect the complexity of the CSP
up to log-space reductions (cf. Lemma 1.2.6 in \cite{Bodirsky:book}).
This connection between CSP and \probfont{DL-Sat} will be 
exploited in Section~\ref{sec:extensions}.

To simplify the presentation, we sometimes use an
alternative notation for a disjunctive constraint
$\bigvee_{\ell=1}^{m} x_{i_\ell} - x_{j_\ell} \in I_\ell$ and write 
it as a set of simple constraints $\{ x_{i_\ell} - x_{j_\ell} \in I_\ell : \ell \in \range{1}{m} \}$.
Then, an assignment satisfies the disjunctive constraint whenever it satisfies at least one simple constraint in the corresponding set. This way of viewing disjunctions
simplifies, for instance, the treatment of certificates in Section~\ref{sec:enumeration}.

When considering CSPs with
infinite
constraint languages, it is important to specify how the relation symbols are represented in the input instances. 
In our case, it would (for instance) be
sufficient to represent the
relation symbol for a relation $R$ by a quantifier-free CNF definition of $R$ using atomic formulas of the form
$x-y \odot c$ with $\odot \in \{<,\leq\}$, and coefficients $c \in \integers$ represented in binary. Such a representation has certain pleasant features: one may, for instance,
check in polynomial time whether a given
rational tuple (where the numerator and denominator are viewed as integers represented in binary) is a member of $R$ or not. Note that there are no representational issues like these when considering finite
constraint languages.

For an instance $\instance$ of $\csp({\bf A})$, we write $\norm{\inst}$ for the number of bits required to represent
$\instance$.
We primarily measure time complexity in terms of $n$ (the number of variables). 
Historically, this has been the most common way of measuring time complexity: 
for instance, the vast majority of work concerning finite-domain \csp{s} 
concentrates on the number of variables.
One reason for this is that an instance may be much larger than the number of variables.
Consider an instance of the
propositional \textsc{SAT} problem,
i.e. a propositional logical formula in CNF.
Such a formula
may contain up to $2^{2n}$ distinct clauses
if repeated literals are disallowed,
so measuring in terms of the instance size may give far too
optimistic figures. It is thus more informative to know
that \textsc{SAT} can be solved in $O^*(2^{n})$ time\footnote{The $\bigoh^*(\cdot)$ notation hides polynomial factors.}  
instead of knowing that it is solvable in
$O^*(2^{\norm{\instance}})$ time.

The various constraint languages that we will consider were defined in Section~\ref{sec:introduction}.
We note that
disjunctive temporal relations are sometimes defined in a more general way which allows for
unary atomic relations $x \in I$ (as opposed to binary atomic relations $x-y \in I$).
The standard trick for handling unary relations
is to introduce a {\em zero variable} (see \cite{barber2000reasoning}).
Solutions to CSP$(\D)$ have the following property:
if $\varphi : V \rightarrow \rationals$ satisfies an instance $(\variables, \constraints)$, 
then so does $\varphi'(v) = \varphi(v) + c$ 
where $c \in \rationals$ is an arbitrary constant.
Thus, we can pick an arbitrary variable in $\variables$ 
and assume that its value is zero:
such a variable is called a {\em zero variable}.
We can now easily express unary constraints, e.g. the constraint $x - z \in (0,2]$ 
is equivalent to $x \in (0,2]$ if $z$ is the zero variable.
Adding a single zero variable 
does not affect the time complexity with more than a multiplicative factor.

%We conclude by observing that solutions to the temporal problems defined above
%are invariant under translation, i.e. 
%if $\varphi : V \rightarrow \reals$ satisfies an instance $(\variables, \constraints)$, 
%then so does $\varphi'(v) = v + c$ 
%for all $v \in \variables$ and an arbitrary $c$.
%Thus, we can pick any variable in $\variables$ 
%and assume that its value is zero without loss of generality.
%We call such a variable a \textit{zero variable}.
%Augmenting an instance of $\csp(\disjtemp{\omega})$ with a zero variable $z$
%allows us to express unary constraints, e.g. the constraint $x - z \in (0,2]$ 
%is equivalent to $x \in (0,2]$.

\subsection{Sidon Sets}
\label{sec:sidon}

Our lower bound results presented in Sections~\ref{sec:lower-bounds-time}
 and~\ref{sec:lb} use {\em Sidon sets}~\cite{sidon1932satz}.
 The study of Sidon sets is an important topic in additive number theory and elsewhere; see e.g.
the survey by O’Bryant~\cite{Obryant:ejc2004}
or the book by Halberstam and Roth~\cite{Halberstam:Roth:sequencebook}.
The terminology used in the literature may appear confusing: they are known under several names such as {\em Golomb rulers}, {\em Sidon sequences}, and $B_2$-sets, and the term Sidon set has 
different meanings in number theory and functional analysis.
A  Sidon set $S$ is a set of integers such that the sum of any pair of its elements is unique, i.e.  if $a + b = c + d$ for
$a, b, c, d \in S$, then $\{a, b\} = \{c, d\}$.
It is easier to work with differences in our proofs so we 
use the following equivalent condition: 
for all $a,b,c,d \in S$ such that $a \neq b$ and $c \neq d$, $a-b=c-d$ holds
if and only if $a=c$ and $b=d$.
This indicates one way of using Sidon sets: they allow us (under certain
conditions) to rewrite a
disjunction $x \neq a \vee y \neq b$
(where $x,y$ are variables and $a,b$ integers)
as a difference $x-y \neq c$ for some integer $c$.

The {\em order} of a Sidon set is the number of elements in it and the {\em length} is the difference between its maximal and minimal elements.
For example, $\{0,1,4,6\}$ is a Sidon set of order $4$ with length $6$.
We will use a particular way of constructing Sidon sets with length quadratic in their order.

\begin{proposition}[\cite{erdos1941problem}]
    \label{prop:golombconstruction}
    Let $p \geq n$ be an odd prime. Then 
    $$S_n = \left\{ pa + (a^2 \bmod p) : a \in \range{0}{n-1} \right\}$$ 
    is a Sidon set.
\end{proposition}

 We sometimes need
to ensure that the length of a Sidon set is bounded by a
polynomial in its order $k$. Indeed, Proposition~\ref{prop:golombconstruction} 
shows that there is a Sidon set
containing $k$ positive integers and whose largest element is at most $2p^2$,
where $p$ is the smallest prime number larger than or equal to $k$. This set can clearly be
constructed in polynomial time. Together with Bertrand's
postulate (see e.g. Chapter 2 in the book by Aigner and Ziegler~\cite{AignerZiegler18}) which states that for every natural
number $n$ there is a prime number between $n$ and $2n$, we see
that a Sidon set of order $k$ and length $8k^2$ can be generated in polynomial time.

\section{Upper Bounds on Time Complexity}
\label{sec:upper-bounds-time}

This section contains two main results:
a $2^{O(n(\log{n}+\log{k}))}$ time algorithm for
$\csp(\D)$ (Section~\ref{sec:generalupperbound}) and a
$2^{O(n \log \log n)}$ time algorithm for
$\csp(\D_{2,k})$ when $k < \infty$ is fixed (Section~\ref{sec:finitealgo}).
These results together with the lower bound results 
that are proved in Section~\ref{sec:lower-bounds-time} are
summarized in Table~\ref{tb:time-summary}.

\subsection{Upper Bound for $\csp(\D)$}
\label{sec:generalupperbound}

We will prove a {\em small solution property} for
$\csp(\D)$. Small solution properties are results that state that
a solvable instance of CSP$({\bf A})$ has a solution that
assign `small' values to the variables. Exactly what is meant
by `small' varies in different contexts. A concrete example is
provided by
Bezem et al.~\cite{Bezem:etal:lpar2008} for the max-atoms
problem that we encountered in Section~\ref{sec:background}:
every satisfiable instance $(V,C)$ of the max-atoms problem has
a solution $f:V \rightarrow \{0,\dots,p\}$ where

\[p = \sum_{\max(x,y)+d \geq z \in C} |d|.\]

Another example (from Section~\ref{sec:intro-time}) is
UTVPI relations. 
These are defined as
$\{(x,y) \in \integers^2 : ax+by \geq c\}$
where $a,b \in \{-1,0,1\}$ and $c \in \integers$.
Seshia et al.~\cite{Seshia:etal:jsbmc2007} prove that every satisfiable instance $(V,C)$ of the CSP over
UTVPI relations has a solution in the interval $\{-|V| \cdot k, \dots, |V| \cdot k\}$ 
where $k=\num{C}$.%\todo{PJ: The bound by Seshia et al. is $\{-|V| \cdot (k+1), \dots, |V| \cdot (k+1)\}$.}
This result implies that every satisfiable instance $(V,C)$ of
CSP$(\D^{\leq})$ has a solution 
$\{-|V| \cdot k, \dots, |V| \cdot k\}$ but it does not give a bound
for CSP$(\D)$---note, for instance, that CSP$(\D)$ is not guaranteed to have integer solutions,
e.g. $\{x - y \in (0,1)\}$.
Bezem et al.'s proof has a graph-theoretical flavor while 
Sesha et al.'s proof is based on a polyhedral approach. None of these methods
appear to be directly applicable to $\csp(\D)$: Bezem et al.'s proof use
intrinsic properties of the max-atom problem while Sesha et al.'s approach 
is built around the fact that solutions must assign integers to the variables.
Our proof strategy has more of an order-theoretic flavor.
Define the set 
$$\cd{n,k} = \left\{ z + \frac{q}{n} : z,q \in \naturals, \; 0 \leq z \leq (n-1)(k+1), \; {\rm and} \; 0 \leq q < n \right\}$$
for $n,k \in \naturals$.
This set will serve as a ruler, and we will show that
any satisfying assignment to an instance of $\csp(\D)$
with $n$ variables and numerical bound $k$
can be transformed into one that only chooses values
from the ruler. 
%The main approach behind the result that is shown in
%Theorem~\ref{lem:compact-assign} below is to show that any %satisfying assignment can be
%transformed into an assignment that uses only values in %$\cd{n,k}$. 
To achieve this we will split the assignment of each variable into the
integral part and the fractional part and show how to independently
transform these parts. Our starting point is the following
lemma which provides sufficient conditions for two assignments to
satisfy the same set of simple constraints. This lemma will also be
useful when proving the forthcoming Lemma~\ref{lem:independence}.

We let $\intP{x}$ denote the {\em floor} function 
(i.e. $\intP{x}$ is the largest integer less than or equal to the real number $x$) and we let
$\ratP{x}$ denote the fractional part of the non-negative real number $x$ (i.e. $\ratP{x}=x-\intP{x}$). 
To simplify the proof, we note that it is sufficient to concentrate on {\em unit} constraints,
which are defined as follows.
Let ${\bf T} \subseteq \disjtemp{2,k}$ be the constraint language with relations
\begin{align*}
  &\{ (x,y) \in \rationals^2 : x - y \in \{i\} \}, \\
  &\{ (x,y) \in \rationals^2 : x - y \in (i,i+1) \}, \text{and} \\
  &\{ (x,y) \in \rationals^2 : x - y \in (i, \infty) \}
\end{align*}
for all $i \in \integers$.
We refer to the relations in ${\bf T}$ as {\em unit} relations.
Consider constraint $x - y \in (-1,0] \cup [1,\infty)$.
An equivalent constraint can be enforced by 
a disjunction of unit constraints
$x - y \in (-1, 0) \lor x - y \in \{0\} \lor x - y \in \{1\} \lor x - y \in (1, \infty)$.
In a similar manner,
we can rewrite every disjunctive temporal constraint
as a disjunction of unit constraints. 
We are now ready to prove the main technical lemma.

\begin{lemma}\label{lem:ass-scon}
  Let $k$ be an integer and let $\phi_1 : V \rightarrow \rationals$ and
  $\phi_2 : V \rightarrow \rationals$ be two assignments of the
  variables in $V$ that satisfy
  the following two conditions:
  \begin{enumerate}
  \item For every $x,y \in V$, it holds that $\phi_1(x)-\phi_1(y)$ and
    $\phi_2(x)-\phi_2(y)$ have the same integer part up to $k+1$, i.e.
    $\min\{\intP{\phi_1(x)-\phi_1(y)},k+1\}=\min\{\intP{\phi_2(x)-\phi_2(y)},k+1\}$.\label{con:lemassscon1}
  \item For every $x,y \in V$, it holds that $\ratP{\phi_1(x)} \odot
    \ratP{\phi_1(y)}$ if and only if $\ratP{\phi_2(x)} \odot
    \ratP{\phi_2(y)}$ for every $\odot \in \{<,=,>\}$.\label{con:lemassscon2}
  \end{enumerate}
  Then, $\phi_1$ and $\phi_2$ satisfy the same simple constraints over
  $V$ with relations in $\disjtemp{2,k}$.
\end{lemma}
\begin{proof}
  To show the lemma, it is
  sufficient to show that $\phi_1$ and $\phi_2$ satisfy the
  same unit constraints.
  Suppose that $\phi_1$ and
  $\phi_2$ satisfy conditions~\ref{con:lemassscon1}
  and~\ref{con:lemassscon2}. We need to show that $\phi_1$ satisfies
  any of the unit constraints on two variables $x$ and $y$ if and only
  if so does $\phi_2$. We distinguish the following cases according to
  the three types of unit constraints given above.
  \begin{itemize}
  \item 
    If $\phi_1(x) - \phi_1(y) \in \{i\}$ for some $i \leq k$, then
    $\intP{\phi_1(x)-\phi_1(y)}=i$ and
    $\ratP{\phi_1(x)}=\ratP{\phi_1(y)}$. Therefore, 
    $\intP{\phi_2(x)-\phi_2(y)}=i$ and
    $\ratP{\phi_2(x)}=\ratP{\phi_2(y)}$, which implies
    that $\phi_2(x) - \phi_2(y) \in \{i\}$.
  \item
    If $\phi_1(x) - \phi_1(y) \in (i,i+1)$ for some $i <k$, then
    $\intP{\phi_1(x)-\phi_1(y)}=i$ and
    $\ratP{\phi_1(x)}>\ratP{\phi_1(y)}$. Therefore, 
    $\intP{\phi_2(x)-\phi_2(y)}=i$ and
    $\ratP{\phi_2(x)}>\ratP{\phi_2(y)}$, which implies
    that $\phi_2(x) - \phi_2(y) \in (i,i+1)$.
  \item
    If $\phi_1(x) - \phi_1(y) \in (i,\infty)$ for some $i \leq k$, then
    either $\intP{\phi_1(x)-\phi_1(y)}=i$ and
    $\ratP{\phi_1(x)}>\ratP{\phi_1(y)}$ or
    $\intP{\phi_1(x)-\phi_1(y)}>i$. In the former case,
    we have that $\intP{\phi_2(x)-\phi_2(y)}=i$ and
    $\ratP{\phi_2(x)}>\ratP{\phi_2(y)}$ and therefore
    $\phi_2(x) - \phi_2(y) \in (i,\infty)$.
    In the latter case,
    we have that $\intP{\phi_2(x)-\phi_2(y)}>i$ and therefore
    $\phi_2(x) - \phi_2(y) \in (i,\infty)$. 
  \end{itemize}
  This completes the proof.
\end{proof}

Lemma~\ref{lem:ass-scon} enables us to give a clear-cut proof of the
small solution property.

\begin{theorem} \label{lem:compact-assign}
  Every satisfiable instance $\inst = (V,C)$ of 
  $\csp(\D)$ has a solution
  $f : V \rightarrow \cd{\abs{V},\num{C}}$.
\end{theorem}
\begin{proof}
  Let $\inst = (V,C)$ be a satisfiable instance of $\csp(\D)$ with solution $g : V \rightarrow \rationals$.
  Let $n=\abs{V}$, $k=\num{C}$, our strategy is to take the assignment
  $g$ and construct a new assignment $f : V \to \cd{n,k}$
  that satisfies the same simple constraints as $g$ over $V$
  with relations in $\disjtemp{2,k}$.

  Index the variables $\{v_1,\dots,v_n\}$
  so that $g(v_i) \leq g(v_{i+1})$ for all $1 \leq i < n$.
  Then, split the values $g(v_i)$ into integral and fractional
  parts, i.e. define $z_i = \intP{g(v_i)}$ and 
  $q_i = \ratP{g(v_i)}$ for all $i$.
  Note that $0\leq q_i < 1$
  and 
  the integers $z_1, \dots, z_n$ are in non-decreasing order.
  % \todo{SO: Isnt the claim a bit too obvious for a proof, i.e., if $x\geq y$ for two real numbers then of course also $\intP{x} \geq \intP{y}$; do we really want to proof this?}
  % \begin{claim}
  %   \label{cl:z-non-decreasing}
  %   The integers $z_1, \dots, z_n$ are in non-decreasing order
  % \end{claim}
  % \begin{claimproof}
  %   Since $g(v_{i+1}) \geq g(v_{i})$, 
  %   we have $z_{i+1} + q_{i+1} \geq z_{i} + q_{i}$.
  %   Observe that $q_{i} - q_{i+1} > -1$ since
  %   $q_{i}, q_{i+1} \in (0,1)$, hence
  %   we have $z_{i+1} - z_{i} \geq q_{i} - q_{i+1} > -1$
  %   and $z_{i+1} > z_{i} - 1$ and $z_{i+1} \geq z_{i}$
  %   since $z_{i}$ and $z_{i+1}$ are integers. 
  % \end{claimproof}
  
  We recursively define the assignment $f(v_i) = c_i + d_i$ for all $i$,
  where $c_i$ is the integral part and $d_i$ is the fractional part of $f(v_i)$.
  Set $c_1 = 0$ and let $c_{i+1} = c_i + \min\{z_{i+1} - z_i, k + 1\}$ for all $1 \leq i < n$.
  Note that $c_1, \dots, c_n$ are sorted in non-decreasing order.
  Furthermore, let $\sigma : \{q_1, \dots, q_n\} \rightarrow \range{0}{n-1}$
  be an injective function 
  such that $\sigma(q_i) \odot \sigma(q_j) \iff q_i \odot q_j$
  for all $i$, $j$ and $\odot \in \{<,=,>\}$.
  One may view $\sigma$ as an order-preserving `scaling' of the fractional parts
  into the integers.
  Let $d_i = \frac{\sigma(q_i)}{n}$ for all $i$.
  Note that $c_n \leq (n-1)(k+1)$ and $0 \leq \sigma(q_i) \leq n-1$, 
  so $f$ maps the variables in $V$ into the set $\cd{n,k}$, as desired. Moreover, since
  $f$ and $g$ satisfy the conditions on $\phi_1$ and $\phi_2$ given in
  the statement of Lemma~\ref{lem:ass-scon}, we obtain that $f$ and
  $g$ satisfy the same simple constraints over $V$ with relations in
  $\disjtemp{2,k}$. Therefore, $f$ also satisfies $\inst$, as required.
\end{proof}

Theorem~\ref{lem:compact-assign} gives us straightforward upper bounds on the time complexity
of $\csp(\D)$ and many of its subclasses.

\begin{corollary} \label{thm:generalupperbound}
Every instance $\inst=(V,C)$ of $\csp(\D)$ can be solved in $2^{O(n(\log n + \log k))}$ time where
$n=|V|$ and $k=\num{C}$. In particular, every instance of $\csp(\D_{\infty,k})$ can be solved in $2^{O(n\log n)}$ time.
\end{corollary}
\begin{proof}
Enumerate all assignments $f:V \rightarrow CD(n,k)$ and check
if they satisfy $\inst$. Theorem~\ref{lem:compact-assign} implies that $\inst$ is satisfiable if and only if at least one such assignment is satisfying.
This takes 
$O^*(|CD(n,k)|^n)$ time in total. The set $CD(n,k)$ contains
$O(n^2k)$ elements so
$$O^*(|CD(n,k)|^n) = O^*((n^2k)^n)=2^{O(n \log n)} \cdot 2^{O(n \log k)} = 2^{O(n(\log n + \log k))}.$$

\end{proof}

\subsection{Upper Bound for \texorpdfstring{$\csp(\disjtemp{2,k})$}{Lg}}
\label{sec:finitealgo}

In this section we prove that $\csp(\disjtemp{2,k})$ 
can be solved in $2^{O(n \log \log n)}$ time.
Our algorithm for $\csp(\disjtemp{2,k})$ is based on a
divide-and-conquer approach, i.e., we split the instance into smaller
parts and solve them recursively. To achieve the splitting, we first
show that any solution for an instance $\instance$ of
$\csp(\disjtemp{2,k})$ suggests a natural split of the whole instance
into either two or three almost independent subinstances; here, almost
independent refers to the instances sharing only a small set of
variables. This will be exploited by the algorithm to enumerate all
possible decompositions into two or three subinstances and recurse on
those for every possible assignment of the shared variables. We
also need a subroutine that allows us to find all solutions with small
domain values. This is used to solve the middle instance in the case
that the instance decomposes into three subinstances.

Before we begin, we describe some polynomial time
preprocessing steps for $\csp(\disjtemp{2})$.
Suppose there are two constraints
for a pair of variables $x - y$:
\begin{align*}
  x - y &\in I_1 \lor \dots \lor x - y \in I_p, \\
  x - y &\in J_1 \lor \dots \lor x - y \in J_q, \\
\end{align*}
where $I_1, \dots, I_p$ and $J_1, \dots, J_q$ are intervals.
For both constraints to hold, 
there must exist $1 \leq i \leq p$ and $1 \leq j \leq q$
such that $x - y \in I_i \cap J_j$.
Thus, we can replace these two constraints with
\[ \bigvee_{i=1}^{p} \bigvee_{j=1}^{q} x - y \in I_i \cap J_j. \]
Applying this procedure exhaustively,
we obtain an instance with at most one constraint
for every pair of variables.
In the rest of the section we assume that all instances
$\declareInstance$ are preprocessed, and we
write $\sigma_{C}(x,y)$ to denote \emph{the} constraint in $C$ over variables
$x$ and $y$; if there are no constraints over $x$ and $y$, we let
$\sigma_{C}(x,y)$ be the set of all possible simple constraints
over $x$ and $y$.

The rest of this section is divided into three parts
(Sections~\ref{sec:enumeration}--\ref{ssec:algod2}):
the first two sections introduce certain subroutines that are needed
in the algorithm, and the
algorithm itself is presented and proven correct in the third section.

%\subsection{Enumeration of Certificates}
%
%The goal of this section is to present a method for enumerating
%compact representations of solutions to $\csp(\disjtemp{\omega})$ instances.
%Let $\declareInstance$ be an instance of $\csp(\disjtemp{\omega})$.
%Define $\allin{\constraints} = \bigcup_{C \in \constraints} C$ to be the set of all simple constraints in %$\constraints$.
%Let $\varphi : \variables \rightarrow \reals$ be an assignment to $\instance$.
%We identify $\varphi$ with the subset of constraints $F \subseteq \allin{\constraints}$ satisfied by %$\varphi$.
%This allows us to define an equivalence relation on the assignments: $\varphi_1 \sim \varphi_2$ if and only %if $F_1 = F_2$.
%This way, $F$ represents the entire class of assignments equivalent to $\varphi$. 
%We say that $F$ is a {\em certificate} of the satisfiability of $\instance$.
%An assignment $\varphi$ is satisfying if the certificate $F$ contains at least one simple constraint from %every $C \in \constraints$.
%Note that if $\varphi_1 \sim \varphi_2$ and $\varphi_1$ is a satisfying assignment, then so is $\varphi_2$.
%While there may be infinitely many satisfying assignments to $\instance$, the number of certificates is %finite: there are at most as many certificates
%as there are subsets of $\allin{\constraints}$.
%
%We begin this section by introducing ordered partitions.
%Next, we establish certain connections between certificates and ordered partitions,
%which form the basis for our enumeration algorithm presented at the end of the section.

\subsubsection{Certificates}
\label{sec:enumeration}

We begin by presenting an alternative method for enumerating
compact representations of solutions to $\csp({\bf D})$ instances in
terms of certificates. The main advantage of certificates compared to
representing solutions by assignments is that certificates allow us to
express the partial solution in terms of constraints. In particular, it allows
us to fix the behavior of certain variables by simply adding
additional constraints to the instance.
% This method will be an important part of
% our algorithms for $\csp(\disjtemp{2,k})$.

Let $\declareInstance$ be an instance of $\csp({\bf D})$.
Define $\allin{\constraints}$ to be the set of simple temporal constraints appearing as disjuncts in $C$.
For example, if $\constraints=\{(x-y \leq 0) \vee (x-y \geq 1), x-z \geq 1\}$, then $\allin{\constraints} = \{x-y \leq 0,x-y \geq 1, x-z \geq 1\}$.
Let $\varphi : \variables \rightarrow \reals$ be an assignment to $\instance$.
We identify $\varphi$ with the subset of constraints $F \subseteq \allin{\constraints}$ satisfied by $\varphi$.
This allows us to define an equivalence relation $\sim$ on the assignments where $\varphi_1 \sim \varphi_2$
holds if and only if $F_1 = F_2$.
This way, $F$ represents the entire class of assignments equivalent to $\varphi$. 
We say that $F$ is a {\em certificate} of the satisfiability of $\instance$.
An assignment $\varphi$ is satisfying if the certificate $F$ contains at least one simple constraint from every $c \in \constraints$.
Note that if $\varphi_1 \sim \varphi_2$ and $\varphi_1$ is a satisfying assignment, then so is $\varphi_2$.
While there may be infinitely many satisfying assignments to $\instance$, the number of certificates is finite since there are at most as many certificates
as there are subsets of $\allin{\constraints}$.

\begin{theorem}
  \label{thm:enumassign}
  The list of certificates to an instance $\declareInstance$
  of $\csp({\bf D})$ can be computed in
  $2^{O(n(\log{n} + \log{k}))}$ time, where
  $n = \abs{\variables}$ and $k = \num{\constraints}$.
\end{theorem}
\begin{proof}
By definition, every certificate $\instance' = (V', \constraints')$ 
for $\instance$ is a satisfiable instance of $\csp({\bf S})$
with $n$ variables and $\num{\constraints'} \leq k$.
Theorem~\ref{lem:compact-assign} implies that 
each $\instance'$ admits a satisfying assignment in $CD(n, k)$.
Thus, we can enumerate assignments $f : V(I) \to CD(n, k)$,
check whether it satisfies $I$, and if so,
collect the simple constraints in $\allin{C}$
satisfied by $f$ and output them as a certificate.
This requires $|CD(n,k)|^n \leq (n^2k)^n = 2^{O(n (\log{n} + \log{k}))}$ time.
\end{proof}

The algorithm underlying the previous theorem will be
used
as a subroutine in our algorithm for $\csp(\disjtemp{2,k})$. We will
refer to it as
\textsc{ListCert} in what follows.

\subsubsection{Instances with Bounded Span}
\label{sec:boundedspan}

We continue by examining a restricted version of $\csp(\disjtemp{2})$
(denoted $w\mhyphen\csp(\disjtemp{2})$) where solutions can only
take values in the interval $[0,w)$; we say that such a solution has
{\em span} $w$.
We show that this problem can be solved in $O^*({w^n})$ time.

\begin{lemma}
    \label{lem:w-disjtemp}
    $w\mhyphen\csp(\disjtemp{2})$ can be solved in $O^*({w^n})$ time.
\end{lemma}
\begin{proof}
    Let $\declareInstance$ be an instance of $w\mhyphen\csp(\disjtemp{2})$ and assume $\varphi : V \rightarrow [0, w)$ is a satisfying assignment.
    Without loss of generality, assume that all constraints in $C$ are represented as disjunctions of unit constraints.
    We split $\varphi$ into integral and fractional parts: $\varphi(x) = \varphi_{i}(x) + \varphi_{f}(x)$, where $\varphi_{i}(x) \in \range{0}{w-1}$ and $0 \leq \varphi_{f}(x) < 1$.
    Suppose we fix $\varphi_{i}$ and want to check whether any $\varphi_{f}$ extends $\varphi_{i}$ to a satisfying assignment.
    For every pair of distinct variables $x$ and $y$ we have $\varphi_{i}(x) - \varphi_{i}(y) = c$ for some integer $c$.
    There are only six nontrivial unit constraints that agree with this assignment,
    each of them expressible as a linear inequality or disequality:
    \begin{align*}
        &x - y \in (c-1, c) &&\longrightarrow &&\varphi_{f}(x) < \varphi_{f}(y) \\
        &x - y \in \{c\} &&\longrightarrow &&\varphi_{f}(x) = \varphi_{f}(y) \\
        &x - y \in (c, c+1) &&\longrightarrow &&\varphi_{f}(x) > \varphi_{f}(y) \\
        &x - y \in (c-1, c] && \longrightarrow && \varphi_{f}(x) \leq \varphi_{f}(y), \\
        &x - y \in [c, c+1) && \longrightarrow && \varphi_{f}(x) \geq \varphi_{f}(y), \\
        &x - y \in (c-1, c) \cup (c, c+1) && \longrightarrow && \varphi_{f}(x) \neq \varphi_{f}(y).
    \end{align*}
    These constraints together with the domain restriction $0 \leq \varphi_{f}(v) < 1$ for each $v$ yield a system of linear inequalities and disequalities that has a solution if and only if there is a fractional assignment $\varphi_{f}$ that extends $\varphi_{i}$ to a satisfying assignment.
    Feasibility of a system of linear inequalities and disequalities can be decided in polynomial time~\cite{Jonsson:Backstrom:ai98,Koubarakis:tcs2001}.
    There are $w^n$ possible functions $\varphi_{i} : \variables \rightarrow \range{0}{w-1}$ and checking whether an
    integer assignment $\varphi_{i}$ can be extended to a satisfying assignment requires polynomial time.
    Hence, the total running time of this algorithm is $O^*({w^n})$.
\end{proof}

We refer to the algorithm underlying the previous lemma as \textsc{SolveBounded}.

\subsubsection{Divide-and-Conquer Strategy}
\label{ssec:algod2}

Our algorithm for $\csp(\disjtemp{2,k})$ is based on a
divide-and-conquer approach: we split the instances into smaller parts
and solve them recursively. To achieve the splitting, we first show
that any satisfying assignment for an instance of
$\csp(\disjtemp{2,k})$ provides a natural split of the instance into
either two or three subinstances. These subinstances are almost independent
in the sense that they only share a small number of variables. 
By enumerating suitable values for the shared variables, we can thus solve
the original instance recursively.
To show that every
satisfying assignment allows one of the two splits, consider a
satisfying assignment $\varphi$ of an instance $\instance=(\variables,\constraints)$ of
$\csp(\disjtemp{2,k})$. We can assume that
the minimal value assigned by $\varphi$ is zero by translational invariance
and we additionally assume that $\varphi: \variables
\rightarrow [0,kw)$ for some $w \in \naturals$. Next, we 
divide the domain of $\varphi$ into intervals $[ki-k,ki)$ of span $k$
for every $i \in \range{1}{w}$. This implies a partition of the
variables in $V$ into disjoint (possible empty) subsets
$\{ V_i : i \in \range{0}{w+1} \}$ defined as follows:
\[ V_i = \{v \in V : \varphi(v) \in [ki-k,ki)\} \]
for all $i \in \range{1}{w}$, and $V_0 = V_{w+1} = \varnothing$. Also,
define the sets $L_i = \bigcup_{j=0}^{i-1} V_j$ and $R_i =
\bigcup_{j=i+1}^{w+1} V_j$ for all $i$. We start by showing that
we can split the instance at any $V_i$ to obtain two subinstances
that are independent up to their overlap at $V_i$. That is, let
$\instance[U] = (U,\constraints[U])$ be the subinstance of $\instance$
induced by the variables in $U$, where
$\constraints[U]$ is the subset of $\constraints$ containing only the
constraints with all variables in $U$. Then, the following lemma 
provides necessary and sufficient conditions that allows us to
transform solutions for the subinstances $\instance[L_i \cup V_i]$ and
$\instance[V_i\cup R_i]$ into a solution for the whole instance;
informally the lemma allows us to split the instance at any $V_i$.  

% A \textit{subinstance} of an instance $\declareInstance$ of $\csp(\disjtemp{})$ 
% induced by a subset of variables $U \subseteq V$ is an instance $\instance[U] = (U,\constraints[U])$, 
% where $\constraints[U] \subseteq \constraints$ contains the constraints that only involve the variables in $U$.

\begin{lemma}
    \label{lem:independence}
    Let $\declareInstance$ be an instance of $\csp(\disjtemp{2,k})$,
    where $V$ is equal to the disjoint union of the sets $X$, $Y$, and $Z$.
    Assume $\varphi_{1}$ and $\varphi_{2}$ are satisfying assignments
    to the subinstances $\instance[X \cup Y]$ and $\instance[Y \cup Z]$,
    respectively.
    If the following conditions hold, then $\instance$ is satisfiable:
    \begin{enumerate}
        \item \label{lem:independence:cond1}
          For every pair of variables $x \in X$ and $z \in Z$, the
          constraint $\sigma_{\constraints}(x,z)$ is empty or it implies
          $z - x > k$.
        \item \label{lem:independence:cond2}
          Assignments $\varphi_{1}$ and $\varphi_{2}$ satisfy the
          same unit constraints over every pair of variables in $Y$.
        \item \label{lem:independence:cond3}
          There is $T_1 \in \rationals$ such that
          $\varphi_1(x) < T_1 \leq \varphi_1(y) < T_1 + k$ 
          for all $x \in X$ and $y \in Y$.
        \item \label{lem:independence:cond4}
          There is $T_2 \in \rationals$ such that
          $T_2 \leq \varphi_2(y) < T_2 + k \leq \varphi_2(z)$
          for all $y \in Y$ and $z \in Z$.
    \end{enumerate}
\end{lemma}
\begin{proof}
First note that the ordering of the variables in $Y$ with respect to $\varphi_1$
and $\varphi_2$ is the same since, otherwise, there are two variables $y$
and $y'$ that do not satisfy the same simple constraints
under $\varphi_1$ and $\varphi_2$ and this contradicts condition \ref{lem:independence:cond2}. 
Rename the variables in $Y$ so that
$y_1,\dotsc,y_{|Y|}$ is in
non-decreasing order with respect to $\varphi_1$ and $\varphi_2$.
It is convenient to assume, without loss of generality, that
$\varphi_1(y_1)=\varphi_2(y_1)=0$ --
we can arrive at such assignments by subtracting
$\varphi_i(y_1)$ from $\varphi_i(y)$ for every $y \in Y$ and $i \in \{1,2\}$. We show next that
$\intP{\varphi_1(y)}=\intP{\varphi_2(y)}$ for every $y \in Y$. This clearly
holds for $y_1$ so we arbitrarily pick another variable $y \in Y$. Now, it holds that
$\varphi_i(y)-\varphi_i(y_1)=\varphi_i(y)$ for every $i \in \{1,2\}$. We
distinguish the following cases:
\begin{itemize}
\item $\varphi_1(y)-\varphi_1(y_1)=i$ for some $i \in \{0,\dotsc,k\}$,
\item $\varphi_1(y)-\varphi_1(y_1)\in (i,i+1)$ for some $i \in \{0,\dotsc,k-1\}$, or
\item $\varphi_1(y)-\varphi_1(y_1)\in (k,\infty)$.
\end{itemize}
Note that the third case cannot occur because of conditions 3 or 4.
Moreover, in the first case
we obtain from condition 2 that $\varphi_2(y)-\varphi_2(y_1)=i$, which implies that
$\varphi_1(y)=\varphi_2(y)=i$. Finally, in the second case, we obtain from condition 2 that
$\varphi_2(y)-\varphi_2(y_1)\in (i,i+1)$, which implies that
$\intP{\varphi_1(y)}=\intP{\varphi_2(y)}$.

Now we show that $\ratP{\varphi_1(y)} \odot \ratP{\varphi_1(y')}$ if and
only if $\ratP{\varphi_2(y)} \odot \ratP{\varphi_2(y')}$ for every $y,y' \in
Y$ and $\odot \in \{<,=,>\}$. We distinguish the following cases. If $\ratP{\varphi_i(y)} =
\ratP{\varphi_i(y')}$ for some $i \in \{1,2\}$, then $\varphi_i(y)-\varphi_i(y')$ is an integer and
conditions 3 and 4 imply that $\varphi_i(y)-\varphi_i(y') \in
\{-k,\dotsc,k\}$. Therefore, we obtain from condition 2 that
$\varphi_1(y)-\varphi_1(y')=\varphi_{2}(y)-\varphi_{2}(y')$ so
$\varphi_{2}(y)-\varphi_{2}(y')$ is an integer and 
$\ratP{\varphi_{2}(y)}=\ratP{\varphi_{2}(y')}$, as required. 
Otherwise,
suppose without loss of generality that $\ratP{\varphi_1(y)} < \ratP{\varphi_1(y')}$ and $\intP{\varphi_1(y)}\leq
\intP{\varphi_1(y')}$. Then, using condition 3, we obtain that
$\varphi_1(y')-\varphi_1(y) \in (i,i+1)$ for some $i \in
\{0,\dotsc,k-1\}$. By condition 2, it follows that
$\varphi_2(y')-\varphi_2(y) \in (i,i+1)$, too. Since the integer parts of
$y$ and $y'$ are the same for both $\varphi_1$ and $\varphi_2$, we see that $\ratP{\varphi_2(y)}<\ratP{\varphi_2(y')}$, as required.

We are now ready to define a common assignment $\varphi$ for $\inst$ as follows.
We keep the integer part of every variable in $X \cup
Y\cup Z$ the same as before, i.e.:
\begin{itemize}
\item for every $x \in X$, we set the
integer part of $\varphi(x)$ equal to the integer part of $\varphi_1(x)$,
\item for every $z \in Z$, we set the
  integer part of $\varphi(z)$ equal to the integer part of $\varphi_2(z)$, and
\item for every $y \in Y$, we set the
  integer part of $\varphi(y)$ equal to the integer part of $\varphi_1(y)$
  and $\varphi_2(y)$; using the property that the integer parts of
  $\varphi_1(y)$ and $\varphi_2(y)$ are equal.
\end{itemize}
To define the fractional part of $\varphi(v)$ for every $v \in X\cup
Y\cup Z$, let $\rho_1$ ($\rho_2$) be the 
ordered partition\footnote{An \textit{ordered partition} of a finite set $S$ is a sequence
of non-empty disjoint subsets $(S_1,\dots,S_\ell)$ such that
$\bigcup_{i=1}^{\ell} S_i = S$.}
of
the variables in $X\cup Y$ ($Y\cup Z$) ordered by non-decreasing
fractional part with respect to $\varphi_1$ ($\varphi_2$).
Then, $\rho_1$ and
$\rho_2$ are equal if restricted to the variables in $Y$ and therefore we can combine them into an ordered
partition $\rho$ of the variables in $X\cup Y\cup Z$. Moreover, from
$\rho$ we can obtain a function
$\sigma: X \cup Y \cup Z \rightarrow \range{0}{|X\cup
  Y\cup Z|-1}$ such that for every $\odot \in \{<,=,>\}$ it holds that:
\begin{itemize}
\item $\sigma(u) \odot \sigma(v)$ if and only if $\ratP{\varphi_1(u)}
  \odot \ratP{\varphi_1(v)}$ for every $u,v \in X\cup Y$ and
\item $\sigma(u) \odot \sigma(v)$ if and only if $\ratP{\varphi_2(u)}
  \odot \ratP{\varphi_2(v)}$ for every $u,v \in Y\cup Z$.
\end{itemize}
We now obtain the fractional part of $\varphi(v)$ by setting
${\varphi(v)}=\sigma(v)/n$, which implies that $\varphi$ satisfies:
\begin{itemize}
\item $\varphi(u) \odot \varphi(v)$ if and only if $\ratP{\varphi_1(u)}
  \odot \ratP{\varphi_1(v)}$ for every $u,v \in X\cup Y$ and
\item $\varphi(u) \odot \varphi(v)$ if and only if $\ratP{\varphi_2(u)}
  \odot \ratP{\varphi_2(v)}$ for every $u,v \in Y\cup Z$.
\end{itemize}
It now follows from Lemma~\ref{lem:ass-scon} (applied to $\varphi$ and
$\varphi_1$ as well as $\varphi$ and $\varphi_2$) that $\varphi$ satisfies
$\inst[X\cup Y]$ and $\inst[Y\cup Z]$. Moreover, because of
condition 3 and 4, it holds that $\varphi(z)-\varphi(x)>k$ for every $x
\in X$ and $z \in Z$, which together with condition 1 implies that
$\varphi$ also satisfies $\inst[X\cup Z]$. Therefore, $\varphi$ satisfies
$\inst$, as required.
\end{proof}
The above lemma shows that the instance can be split at any $V_i$ into two subinstances $\instance[L_i \cup V_i]$ and
$\instance[V_i\cup R_i]$ that are independent once we fix a certificate for the instance $\instance[V_i]$. This also means that a split at $V_i$ will only be useful  if the number of possible certificates of $\instance[V_i]$, or equivalently the number of variables in $V_i$, is not too large. That is, if $V_i$ contains at most $n /
\log{n}$ variables, then we say it is \textit{sparse} and, otherwise, we
say that it is \textit{dense}. 
The choice of $n / \log{n}$ as the threshold value is justified by the
following proposition:
\begin{proposition}
    \label{prop:n-div-logn}
  The number of certificates for any instance $\declareInstance$ of
  $\csp({\bf D})$ with at most $N\leq n/\log(n)$ variables can be
  computed in time $2^{O(n)}$.
\end{proposition}
\begin{proof}
  This follows immediately from Theorem~\ref{thm:enumassign} after
  observing that $2^{O(N(\log{N} + \log{k}))}=2^{O(n)}$, where $k =
  \num{\constraints}$ is constant.
\end{proof}

The next lemma is crucial for our algorithm since it allows us to
show that any solution naturally splits the instance into either two
or three almost independent subinstances; we will later see how the
subinstances are obtained from the partitions identified by the
lemma. Informally, case 1 in the lemma corresponds to a three-split and
gives rise to the two
instances $\instance[L_i\cup V_i]$ and $\instance[V_i\cup R_i]$ and
case 2 in the lemma corresponds to a five-split and gives rise to the three instances
$\instance[L_i\cup V_i]$, $\instance[V_i \cup V_{i+1} \cup \dotsb \cup V_{j-1} \cup V_j]$ (where $i \leq j$), and
$\instance[V_j\cup R_j]$.

\begin{lemma}
    \label{lem:splits}
    Let $\declareInstance$ be an instance of $\csp(\disjtemp{2,k})$, let $\varphi : \variables \rightarrow [0,kw)$ be an assignment of $\declareInstance$, and let $V_i$, $L_i$, and $R_i$ be defined as above with respect to $\varphi$.
    If $\abs{V} \geq 8$, then one of the following holds:
    \begin{enumerate}
        \item There is an index $0 \leq i \leq w+1$ such that $V_i$ is sparse, $\abs{L_i} \geq n/3$ and $\abs{R_i} \geq n/3$.\label{lem:splits:three}
        \item There are indices $0 \leq i < j \leq w+1$ such that $V_i$ and $V_j$ are sparse, $V_s$ is dense for all $i < s < j$, $\abs{L_i} < n/3$ and $\abs{R_j} < n/3$.\label{lem:splits:five}
    \end{enumerate}
\end{lemma}
\begin{proof}
    If $V_i$ is sparse, then $\abs{L_i} + \abs{R_i} = \abs{V} - \abs{V_i} \geq n - n / \log{n}$.
    Since $\log{n} \geq 3$, we have $\abs{L_i} + \abs{R_i} \geq {2n}/{3}$.
    Thus, for every sparse $V_i$ either $\abs{L_i} \geq n / 3$ or $\abs{R_i} \geq n / 3$.
    
    Let $i$ be the maximal index such that $V_i$ is sparse and $\abs{R_i} \geq n/3$.
    Similarly, let $j$ be the minimal index such that $V_j$ is sparse and $\abs{L_j} \geq n/3$.
    Such indices always exist since $V_0$ and $V_{w+1}$ are sparse.
    If $i \geq j$, then both $V_i$ and $V_j$ meet the conditions of Case~\ref{lem:splits:three}.
    Otherwise, all $V_s$ for $i < s < j$ are dense.
    If neither $V_i$ nor $V_j$ fulfills the conditions of Case~\ref{lem:splits:three}, 
    then $\abs{L_i} < n/3$ and $\abs{R_j} < n/3$, and we are in Case~\ref{lem:splits:five}.
\end{proof}

\algrenewcommand\algorithmicindent{1.0em}%
\begin{algorithm}[tbh!]
    \caption{}
    \label{alg:gammaone}
    \begin{algorithmic}[1]
        \Procedure{Solve}{$\declareInstance$}
            \If{$\varnothing \in \constraints$}
                \Reject
            \EndIf
            \If{$\abs{V} < 8$}\label{alg:gammaone:basecase}
                \State \Accept \textbf{ if }$\Call{ListCert}{\instance} \neq \varnothing$ \textbf{ else }\Reject
            \EndIf
            \If{$\Call{ThreeSplit}{\instance}$}
                \Accept
            \EndIf
            \If{$\Call{FiveSplit}{\instance}$}
                \Accept
            \EndIf
            \State \Reject
        \EndProcedure
        
        \Statex % empty line

        \Procedure{ThreeSplit}{$\declareInstance$}
            \For{\textbf{each} $3$-partition $(X,Y,Z)$ of $V$} 
                \If{$\abs{X},\abs{Z} \geq \frac{\abs{V}}{3}$ \textbf{and} $\abs{Y} \leq \frac{\abs{V}}{\log{\abs{V}}}$\textbf{ and}\label{alg:gammaone:1:partition1}
                \State $z - x > k \in \sigma_{\constraints}(x,z)$ \textbf{for} $x \in X, z \in Z$\label{alg:gammaone:1:partition2}
                }
                    \Statex \algtab\algtab\algtab\textit{// introduce
                      a fresh variable $y_{\min}$}
                    \State $Y' \leftarrow Y \cup \{y_{\min}\}$
                    \State $\instance_{Y'} \leftarrow (Y', \constraints_{Y} \cup \{ 0 \leq y - y_{\min} < k \}_{y \in Y})$ \label{alg:gammaone:1:instanceB} 
                    \For{$F_{Y'} \in \Call{ListCert}{\instance_{Y'}}$}\label{alg:gammaone:1:listassign}
                        \State $\instance_{1} \leftarrow (X\cup Y',\constraints_{X \cup Y} \cup F_{Y'} \cup \{ y_{\min} - x > 0 \}_{x \in X})$\label{alg:gammaone:1:instance1}
                        \State $\instance_{2} \leftarrow (Y'\cup Z,\constraints_{Y \cup Z} \cup F_{Y'} \cup \{ z - y_{\min} > k \}_{z \in Z})$\label{alg:gammaone:1:instance2}
                        \If{\Call{Solve}{$\instance_{1}$} \textbf{and} \Call{Solve}{$\instance_{2}$}}\label{alg:gammaone:1:recursive}
                            \State \Accept
                        \EndIf
                    \EndFor
                \EndIf
            \EndFor
            \State \Reject
        \EndProcedure

        \Statex % empty line
        \algrenewcommand\algorithmicindent{1.0em}%
        \Procedure{FiveSplit}{$\declareInstance$}
            \For{\textbf{each} $5$-partition $(S_1,S_2,S_3,S_4,S_5)$ of $\variables$}\label{alg:gammaone:2:partition}
                \State $X \leftarrow S_3 \cup S_4 \cup S_5$
                \If{$\abs{S_1}, \abs{S_5} < \frac{\abs{V}}{3}$ \textbf{and} $\abs{S_2}, \abs{S_4} \leq \frac{\abs{V}}{\log{\abs{V}}}$\textbf{ and}\label{alg:gammaone:2:partition1}
                \State $(x - s_1 > k) \in \sigma_{\constraints}(s_1,x)$ \textbf{for} $s_1 \in S_1$, $x \in X$\textbf{ and}\label{alg:gammaone:2:partition2}
                \State $(s_5 - s_3 > k) \in \sigma_{\constraints}(s_3,s_5)$ \textbf{for} $s_3 \in S_3$, $s_5 \in S_5$\label{alg:gammaone:2:partition3}
                }
                    \Statex \algtab\algtab\algtab\textit{// introduce fresh variables $s_2^{\min}$ and $s_4^{\min}$}
                    \State $S_2'=S_2\cup \{s_2^{\min}\}$
                    \State $S_4'=S_4\cup \{s_4^{\min}\}$
                    \State $\instance_{S_2'} \leftarrow (S_2',\constraints_{S_2} \cup \{ 0 \leq s_2 - s_2^{\min} < k \}_{s_2 \in S_2})$\label{alg:gammaone:2:instanceB}
                    \State $\instance_{S_4'} \leftarrow (S_4',\constraints_{S_4} \cup \{ 0 \leq s_4 - s_4^{\min} < k \}_{s_4 \in S_4})$\label{alg:gammaone:2:instanceD}
                    \For{$F_{S_2'} \in \Call{ListCert}{\instance_{S_2'}}$ \textbf{and}\label{alg:gammaone:2:listassign1}
                    \State \:\:$F_{S_4'} \in \Call{ListCert}{\instance_{S_4'}}$}\label{alg:gammaone:2:listassign2}
                        \State $\instance_1 \leftarrow (S_1\cup S_2',\constraints_{S_1 \cup S_2} \cup F_{S_2'} \cup \{ s_2^{\min} - s_1 > 0 \}_{s_1 \in S_1})$ \label{alg:gammaone:2:instance1}
                        \State $\instance_2 \leftarrow (S_2'\cup
                        S_3\cup S_4',\constraints_{S_2 \cup S_3 \cup S_4} \cup F_{S_2'} \cup \{ s_3 - s_2^{\min} > k \}_{s_3 \in S_3} \cup$\label{alg:gammaone:2:instance2}
                        \Statex \algtab\algtab\algtab\algtab\algtab\algtab\: $F_{S_4'} \cup \{ s_4^{\min} - s_3 > 0 \}_{s_3 \in S_3})$     
                        \State $\instance_3 \leftarrow (S_4'\cup S_5,\constraints_{S_4 \cup S_5} \cup F_{S_4'} \cup \{ s_5 - s_4^{\min} > k \}_{s_5 \in S_5})$\label{alg:gammaone:2:instance3}
                        \State $w \leftarrow k(\log{\abs{V}} + 2)$
                        \If{$\Call{Solve}{\instance_1}$ \textbf{and} $\Call{Solve}{\instance_3}$ \textbf{and}\label{alg:gammaone:2:recursive}
                        \State $\Call{SolveBounded}{w, \instance_2}$}\label{alg:gammaone:2:bounded}
                            \State \Accept
                        \EndIf
                    \EndFor
                \EndIf
            \EndFor
            \State \Reject
        \EndProcedure
        \algrenewcommand\algorithmicindent{1.0em}%
    \end{algorithmic}
\end{algorithm}
\algrenewcommand\algorithmicindent{2.0em}%

We are now ready to prove that the algorithm is correct and analyze its running time.
\begin{theorem}
    \label{thm:gammaonetime}
    Algorithm~\ref{alg:gammaone} solves instances of $\csp(\disjtemp{2,k})$
    in $2^{O(n \log{\log{n}})}$ time.
\end{theorem}
\begin{proof}
    Consider an arbitrary instance $\declareInstance$ of $\csp(\disjtemp{2,k})$.
    We prove the claim by induction based on $\abs{V}$.
    If the instance has fewer than 8 variables, the claim follows by Corollary~\ref{thm:generalupperbound}.
    Assume henceforth that $\abs{V} \geq 8$.

    First, we prove that if the algorithm accepts an instance, then it is satisfiable. 
    Suppose the procedure \textsc{ThreeSplit} accepts $\instance$.
    We show that instance $\instance$ is satisfiable since all four conditions of Lemma~\ref{lem:independence} are fulfilled.
    First, note that subinstances $\instance_1$ and $\instance_2$ in lines~\ref{alg:gammaone:1:instance1}~and~\ref{alg:gammaone:1:instance2} have at most ${2n}/{3} + 1 < n$ variables.
    Hence, they admit satisfying assignments by the inductive hypothesis.
    Condition~\ref{lem:independence:cond1} is ensured by the check on line~\ref{alg:gammaone:1:partition2}.
    Condition~\ref{lem:independence:cond2} is ensured since both subinstances $\instance_1$ and $\instance_2$ satisfy all constraints in  $F_{Y'}$.
    Conditions~\ref{lem:independence:cond3}~and~\ref{lem:independence:cond4}
    are ensured by the introduction of $y_{\min}$ and the constraints
    involving it.
    
    Suppose instead that the procedure \textsc{FiveSplit} accepts $\instance$.
    First, note that subinstances $\instance_1$ and $\instance_3$ in lines~\ref{alg:gammaone:2:instance1}~and~\ref{alg:gammaone:2:instance3} have at most ${n}/{3} + 1 < n$ variables.
    Hence, they admit satisfying assignments by the inductive hypothesis.
    Subinstance $\instance_2$ is satisfiable by Lemma~\ref{lem:w-disjtemp}.
    Observe that Lemma~\ref{lem:independence} applies to the subinstance induced by $S_3 \cup S_4 \cup S_5$. We see that
    Condition~\ref{lem:independence:cond1} is ensured by the check on line~\ref{alg:gammaone:2:partition3},
    Condition~\ref{lem:independence:cond2} is ensured since both
    subinstances $\instance_2$ and $\instance_3$ satisfy all
    constraints in $F_{S_4'}$, and
    Conditions~\ref{lem:independence:cond3}~and~\ref{lem:independence:cond4} are ensured by the introduction of $s_4^{\min}$ and the constraints involving it. 
    Let $X = S_3 \cup S_4 \cup S_5$ and consider the subinstance induced by $S_1 \cup S_2 \cup X$. 
    Lemma~\ref{lem:independence} is applicable and we see that
    Condition~\ref{lem:independence:cond1} is ensured by the check on line~\ref{alg:gammaone:2:partition2},
    Condition~\ref{lem:independence:cond2} is ensured since both
    subinstances $\instance_1$ and $\instance_2$ satisfy all
    constraints in $F_{S_2'}$, and
    Conditions~\ref{lem:independence:cond3}~and~\ref{lem:independence:cond4} are ensured by the introduction of $s_2^{\min}$ and the constraints involving it. 
    Hence, we have showed that $\instance$ is satisfiable. 

\smallskip

    We proceed by proving the other direction: if $\instance$ is
    satisfiable with satisfying assignment $\varphi$,
    then the algorithm accepts it.
    Lemma~\ref{lem:splits} implies that $\varphi$ splits the instance as in Case~\ref{lem:splits:three}~or~\ref{lem:splits:five}.

\medskip

    \noindent
    {\bf Case~\ref{lem:splits:three}.} We show that \textsc{ThreeSplit} accepts
    $\instance$. The procedure enumerates every 3-partition of the variables, 
    so at some step of the algorithm $X = L_i$, $Y = V_i$ and $Z =
    R_i$, where $(L_i,V_i,R_i)$ is the split under the assignment
    $\varphi$ according to Lemma~\ref{lem:splits}
    Case~\ref{lem:splits:three}.
    We set $\varphi(y_{\min}) = ki - k$ and observe that $\varphi$ satisfies instances in Lines~\ref{alg:gammaone:1:instanceB},~\ref{alg:gammaone:1:instance1}~and~\ref{alg:gammaone:1:instance2}.
    Note that the procedure \textsc{Solve} accepts $\instance_1$ and $\instance_2$ by the inductive hypothesis.

\medskip

\noindent
    {\bf Case~\ref{lem:splits:five}.} We show that \textsc{FiveSplit} accepts
    $\instance$. The procedure enumerates every 5-partition of the variables, 
    so at some step of the algorithm $S_1 = L_i$, $S_2 = V_i$, $S_3 = R_i \cap L_j$, $S_4 = V_j$ and $S_5 = R_j$, 
    where $(L_i,V_i,R_i \cap L_j,V_j,R_j)$ is the split under the
    assignment $\varphi$ according to Lemma~\ref{lem:splits}
    Case~\ref{lem:splits:five}.
    We set $\varphi(s_2^{\min}) = ki - k$ and $\varphi(s_4^{\min}) =
    kj - k$ and observe that $\varphi$ satisfies the instances in Lines~\ref{alg:gammaone:2:instanceB},~\ref{alg:gammaone:2:instanceD},~\ref{alg:gammaone:2:instance1},~\ref{alg:gammaone:2:instance2}~and~\ref{alg:gammaone:2:instance3}.
    Note that the procedure \textsc{Solve} accepts $\instance_1$ and $\instance_3$ by the inductive hypothesis. 
    By Lemma~\ref{lem:splits}, all subsets in $R_i \cap L_j$ are dense. 
    Since each of them contains at least ${n}/{\log{n}}$ variables, there may be at most $\log{n}$ such subsets.
    Taking $V_i$ and $V_j$ into account, we conclude that $\varphi$ satisfies instance $\instance_2$ with span at most $k(\log{n} + 2)$.
    This completes the correctness proof.

    \paragraph{Time complexity.}
    Let $T(n)$ be the running time of Algorithm~\ref{alg:gammaone} on an instance of $\csp(\disjtemp{2,k})$ with $n$ variables.
    We claim that $T(n) \leq c^n (\log{n} + 2)^n = 2^{O(n \log \log n)}$
    for some constant $c$.
    If $n < 8$, then $T(n)$ is constant.
    Otherwise, $T(n) = T_1(n) + T_2(n) + \textnormal{poly}(n)$, 
    where $T_1$ and $T_2$ are the running times of the procedures \textsc{ThreeSplit} and \textsc{FiveSplit}, respectively.
    Note that $T_1(n) < T_2(n)$ for all $n$, so we can focus our attention on the running time of \textsc{FiveSplit}.

    The running time $T_2(n)$ is bounded from above by
    $5^n \cdot 2^{2n} \cdot (2T(\tfrac{n}{3} + 1) + (k(\log{n} + 2))^n) \cdot \textnormal{poly}(n)$,    
    where
    $5^n$ is an upper bound on the number of $5$-partitions of $V$,
    ${2}^{2n}$ comes from the upper bound on the running time of the calls to \textsc{ListCert} in lines~\ref{alg:gammaone:2:listassign1}~and~\ref{alg:gammaone:2:listassign2},
    $2T(\frac{n}{3} + 1)$ is an upper bound on the running time of the recursive calls in line~\ref{alg:gammaone:2:recursive}, and
    $(k(\log{n} + 2))^n$ comes from the running time of the bounded-span algorithm {\sc SolveBounded} in line~\ref{alg:gammaone:2:bounded} (see Lemma~\ref{lem:w-disjtemp}).
    Observe that for sufficiently large values of $n$ we have
    $(\log{n} + 2)^n > \alpha \log{(\varepsilon n)}^{\varepsilon n}$
    for arbitrary $\alpha \geq 0$ and $0 \leq \varepsilon < 1$.
    Hence, $(k(\log{n} + 2))^n$ asymptotically dominates $2T(\tfrac{n}{3} + 1)$ by our initial hypothesis.
    Finally, observe that
    \[ T(n) < 2 \cdot (40k)^n \cdot (\log{n} + 2)^n \cdot \textnormal{poly}(n). \]
    Setting $c = 81k$ completes the proof.
\end{proof}

\section{Lower Bounds on Time Complexity}
\label{sec:lower-bounds-time}

This section contains our time complexity lower bound results
for various subclasses of $\csp(\D)$.
The full picture of upper and
lower bounds on time complexity is summarized in Table~\ref{tb:time-summary}.
The majority of our lower-bound results are based on reductions from the
$n \times n$ \probfont{Independent Set} problem.

\pbDef{$n \times n$ \probfont{Independent Set}}
{A graph $G = (V,E)$ with vertex set~$V = \{(i,j) : 1 \leq i \leq n, 1 \leq j \leq n\}$.}
{Is there an independent set in $G$ with one vertex from each row, i.e., subset of vertices $\{ (1,j_1), \dots, (n,j_n) \}$ where no pair of vertices is connected by an edge?} 

Lokshtanov et al.~\cite{lokshtanov2018slightly} have proved the following result.

\begin{theorem}
    \label{thm:nxnindset}
    $\nxnindset$ cannot be solved in $2^{o(n \log{n})}$ time unless the ETH fails.    
\end{theorem}

We begin by proving lower bounds for $\csp(\D_{a,0})$ when $a \geq 4$.
Note that $\csp(\D_{a,0}^{\leq})$
is trivially in P for every $a \geq 0$ so this result cannot be
extended to $\csp(\D^{\leq}_{a,0})$.

\begin{theorem}
    \label{thm:omegakresult}
    $\csp(\D_{4,0})$
    cannot be solved in $2^{o(n \log{n})}$ time unless the ETH fails.
\end{theorem}
\begin{proof}
    We reduce from $\nxnindset$.
    Given an instance $G$ of this problem, do the following.

\begin{enumerate}
\item
    Introduce $n$ column variables $c_{1},\dots,c_{n}$ and the constraints
    $c_{i} < c_{i-1}$ for all $i \in \range{2}{n}$.

\item
    Introduce $n$ row variables $r_1,\dots,r_n$.
    To ensure that each $r_i$ is equal to one of the column variables,
    add the following constraints:
    $c_1 \leq r_i$, $r_i \leq c_n$ and 
    $(r_i \leq c_{j-1}) \lor (r_i \geq c_j)$
    for all $j \in \range{2}{n}$.

\item
    No pair of vertices $(i,j)$ and $(k,\ell)$ adjacent in $G$ can be simultaneously included in the independent set.
    To ensure this property, we add the following constraint:
    \[ (r_i < c_{j}) \lor (r_i > c_{j}) \lor (r_k < c_{\ell}) \lor (r_k > c_{\ell}). \]
\end{enumerate}

The resulting set of constraints only use relations in $\csp(\disjtemp{4,0})$.
The correctness of the reduction is easy to verify: If $G$ has an independent set $I$, than setting $r_i = c_j$ for all $(i,j) \in I$ satisfies all constraints of the instance above, and vice versa. 
The reduction requires polynomial time and introduces $2n$ variables.
    Thus, if $\csp(\D_{4,0})$ admits a $2^{o(n \log{n})}$ algorithm, then so does $\nxnindset$ and this contradicts the ETH by Theorem~\ref{thm:nxnindset}.
\end{proof}

A slightly weaker bound can be inferred from
Theorem~11 in \cite{Jonsson:Lagerkvist:mfcs2018}:
if
the randomized ETH holds, then there is no 
randomized algorithm for
$\csp(\D_{4,0})$ that runs in $O(c^n)$
time for any $c \geq 0$.
We continue by studying $\csp(\disjtemp{3,1}^{\leq})$.
We use Sidon sets in the proof so the reader
may want to skip back to Section~\ref{sec:sidon} 
for a reminder.

\begin{theorem} \label{thm:D3,1-hardness}
$\csp(\disjtemp{3,1}^{\leq})$ is not solvable in $2^{o(n\log{n} )}$ time if the
    ETH holds.
\end{theorem}
\begin{proof}
The proof is by a reduction from $n \times n$ \textsc{Independent Set}.
    Assume $G=(V,E)$ to be an arbitrary instance of this problem
    with $r$ rows. We use the results from Section~\ref{sec:sidon}
 and construct (in polynomial time) a Sidon set $S=\{a_0,\dots,a_r\}$ where $0=a_0 < a_2 < \dots < a_r$  and $a_r \leq 8r^2$. We present the rest of the reduction in three steps.

\begin{enumerate}
\item
    Introduce $a_r$ fresh variables $y_0,\dots,y_{a_r}$
    and use the relation $y=x+1$ for enforcing $y_i=i$, $0 \leq i \leq a_r$.
    %This construction can be generated in polynomial time since $a_r \leq 8r^2 \leq 8||G||^2$.

\item
    Introduce one variable $x_r$
    for each row $r$. 
    We first ensure that the range of each variable is in $\{0,\dots,y_{a_r}\}$
    by adding the constraint
    \[ x_r \leq a_i-1 \vee x_r = a_i \vee x_r \geq a_i+1\]
    for each $1 \leq i \leq a_r$ together with the constraints
    $x \geq 0$ and $x \leq y_{a_r}$.
    Next, we restrict the range to be in $S$ by
    adding the constraint
   \[x \leq y_s-1 \vee x \geq y_s+1\]
for each $s \in \{0,\dots,a_r\} \setminus S$.

\item
Let $R(x,y,z)$ denote the relation $x-y \neq z$, i.e. $R(x,y,z) \equiv (x-y < z) \vee (y-x < z)$.
    For every edge $((c,c'),(d,d')) \in E$, compute $e=a_{c'} - a_{d'}$ and add
    the constraint $R(x_c,x_d,y_e)$.

\end{enumerate}

\noindent
The resulting set of constraints only use relations in $\csp(\disjtemp{3,1}^{\leq})$
This reduction can be performed in polynomial time. Steps 1. and 2. take
polynomial time since $a_r \leq 8r^2 \leq 8||G||^2$ and step 3.
can obviously performed in polynomial time in the size of $E$.
Let $(V',C')$ denote the resulting instance of $\csp(\disjtemp{3,1})$.
We see that $|V'| \leq 8r^2+r \leq 8|V|+ \lceil \sqrt{|V|} \rceil \leq 9|V|$.
Thus, if $(V',C')$ is satisfiable if and only if $G$ is a yes-instance,
then CSP$(\disjtemp{3,1})$ is not solvable in
$2^{o(n \log{n} )}$
time by Theorem~\ref{thm:nxnindset}.
We conclude the proof by proving this equivalence.

\medskip

\noindent
{\bf Forward  direction.} Assume $f$ is a solution to
$(V',C')$. We claim that for every $1 \leq c \neq d \leq r$,
it holds that $((c,f(c)),(d,f(d)))$ is not in $E$ and $G$ is a yes-instance.
Assume to the contrary that $c,d$ can be chosen such that
$((c,f(c)),(d,f(d))) \in E$.
This edge implies that there is a constraint $R(x_c,x_d,y_e) \in C'$
where $e=a_{f(c)}-a_{f(d)}$. Hence,

$$f(x_c)-f(x_d) \neq f(y_e) \Rightarrow$$
$$a_{c'}-a_{d'} \neq a_{f(c)}-a_{f(d)} \Rightarrow$$
$$a_{c'}-a_{d'}  \neq a_{c'}-a_{d'}$$

and we have reached a contradiction.

\medskip

\noindent
{\bf Backward direction.}
Assume $X=\{v_1,\dots,v_r\} \subseteq V$ is an independent set and $v_i$ occurs in
position $(i,i')$, $1 \leq i,i' \leq r$. 
Define the assignment $f$ such that $f(x_i)=a_{i'}$, $1 \leq i \leq r$,
and $f(y_i)=i$, $0 \leq i \leq a_r$. 
This assignment satisfies all constraints introduced in step 2.

Arbitrarily choose an edge $((c,c^*),(d,d^*)) \in E$. It gives rise to
the constraint $R(x_c,x_d,y_e)$ where $e=(a_{c^*}-a_{d^*})$.
We know that $((c,c'),(d,d'))$ is a non-edge in $E$ since $(c,c')=v_c$,
$(d,d')=v_d$, and $v_c,v_d$ are in the independent set $X$.
The constraint $R(x_c,x_d,y_e)$ is equivalent to $x_c - x_d \neq y_e$.
We apply the assignment $f$ to it:

$$f(x_c)-f(x_d) \neq f(y_{a_{c^*}-a_{d^*}}) \Rightarrow$$
$$a_{c'}-a_{d'} \neq a_{c^*}-a_{d^*}$$

We know that $((c,c'),(d,d'))$ is a non-edge in $E$ while
$((c,c^*),(d,d^*))$ is an edge in $E$. 
Hence, $a_{c'} \neq a_{c^*}$ or $a_{d'} \neq a_{d^*}$. 
Recall that
for all $w_1,\dots,w_4 \in S$ such that $w_1 \neq w_2$ and $w_3 \neq w_4$, $w_1-w_2=w_3-w_4$
if and only if $w_1=w_3$ and $w_2=w_4$.
Thus, if $a_{c'} \neq a_{d'}$ and $a_{c^*} \neq a_{d^*}$, then the disequality
holds.
If $a_{c'}=a_{d'}$, then $a_{c^*} \neq a_{d^*}$ (since
 $((c,c'),(d,d')) \not\in E$ and
$((c,c^*),(d,d^*)) \in E$) and we conclude
that the disequality holds.
The case when $a_{c^*}=a_{d^*}$ is symmetric.
\end{proof}

Our final lower bound based on a reduction from
$n \times n$ \textsc{Independent Set} concerns the problem
$\csp(\disjtemp{2}^{\leqslant})$. Since there is no
upper bound on the coefficients in this case, the lower bound
is parameterized both by the number of variables and the maximum
over the coefficients.

\begin{theorem}
    \label{thm:binary-hardness}
    $\csp(\disjtemp{2}^{\leqslant})$ is not solvable in $2^{o(n (\log{n}+\log{k}) )}$ time if the
    ETH holds.
\end{theorem}
\begin{proof}
    The proof is by reduction from $n \times n$ \textsc{Independent Set}.
    Given an instance $G$ of this problem, we introduce a zero variable 
    $z$ and two variables $x_r,x'_r$ for each row $r$.
    We add the following constraints:
    \begin{enumerate}
      \item $\bigvee_{i=1}^{n} x_r - z = i$ for every row $r$,
      \item $\bigvee_{i=1}^{n} x'_r - x_r = ni - i$ and $\bigvee_{i=1}^{n} x'_r - z = ni$ for every row $r$,
      \item $x'_a - x_b \in (-\infty, ni - j - 1] \lor x'_a - x_b \in [ni - j + 1, \infty)$
      for every edge $\{(a,i), (b,j)\}$ in $G$. 
    \end{enumerate}
    Constraints of the first type restrict the domain of $x_r$ to $\{1,\dots,n\}$.
    Constraints of the second type ensure that $x'_r = nx_r$.
    Constraints of the third type ensure are equivalent to $x'_a - x_b \neq ni - j$,
    which forbids setting $x_a = i$ and $x_b = j$.
    Furthermore, function $f : \range{1}{n} \times \range{1}{n} \rightarrow \range{0}{n^2-1}$ defined as $f(i,j) = ni-j$ is bijective,
    so $x'_a - x_b \neq ni - j$ is satisfied by every other choice of values, i.e. whenever $x_a \neq i$ or $x_b \neq j$.
    Thus, the resulting instance of $\csp(\disjtemp{2}^{\leqslant})$ has a solution if and only if 
    $G$ has an independent set with one variable per row.
    The total number of variables in the resulting instance is $2n+1$
    and the absolute values of the integers appearing in the constraints do not exceed $n^2$.
    Thus, an algorithm solving $\csp(\disjtemp{2}^{\leqslant})$ in $2^{o(n(\log{n} + \log{k}))}$ time can be used for solving $\nxnindset$ in
    $$2^{o(n(\log{(2n+1)} + \log{n^2}))} = 2^{o(n\log{n})}$$
    time and this contradicts the ETH by Theorem~\ref{thm:nxnindset}.
\end{proof}

Our final lower bound for $\csp(D_{2,k})$ relies on a modification of a
result by Traxler~\cite{traxler2008time}.
Let $\dcsp$ be the constraint satisfaction problem with domain $D = \range{1}{d}$ and 
binary relations \mbox{$R_{a,b} = (x \neq a \lor y \neq b)$} for all $a,b \in D$.
Note that this relation is akin to the binary relation in
our reductions from $(n \times n)$ {\sc Independent Set}
that is used to forbid choosing an edge.
We cannot use the same reduction here because the numerical bound $k$ is fixed.
Indeed, a $2^{o(n \log n)}$ lower bound for this problem
stands in contradiction with Theorem~\ref{thm:gammaonetime} or the ETH.
For our reductions, we consider
a modified version of $\dcsp$ denoted by $\dcspchi$.
It allows binary relations
\begin{equation*}
    \textstyle
    R_{a,b}^{\chi} = (x \neq a \lor y \neq b) \land \bigwedge_{c \in D} (x \neq c \lor y \neq c)
\end{equation*}
for all $a,b \in D$.
For convenience, we consider the constraints that rule out $(c,c)$ tuples (i.e. the
right-hand side constraints in the definition above) separately, and assume the following rule: 
if an instance of $\dcspchi$ includes a constraint $(x \neq a \lor y \neq b)$, then it
implicitly includes the required constraints $(x \neq c \lor y \neq c)$ for all $c \in D$.
Additionally, we allow unary relations $x \neq a$ for all $a \in D$.

\begin{lemma}
    \label{lem:traxler}
    For any $r \in \naturals$ and any instance of $\dcspchi$ with $n$ variables, there exists an equivalent instance of $\dcspchi[d^r]$ with $\ceil{n / r}$ variables.
\end{lemma}

\begin{proof}
    Let $\mathcal{I} = (V,\mathcal{C})$ be an instance of $\dcspchi$ with $\abs{V} = n$.
    Augment $V$ with at most $r-1$ extra variables so that its new size $n'$ becomes a multiple of $r$.
    Partition $V$ into $\ell = n' / r = \ceil{n / r}$ disjoint subsets $V_1, \dots, V_\ell$ of equal size and index the elements of each subset arbitrarily.
    
    The set of tuples $D^r$ represents all assignments to the variables in a subset $V_i$.
    For convenience, we use the tuples directly as the domain of $\dcspchi[d^r]$.
    Define an instance $\mathcal{I'} = (\variables', \constraints')$ of $\dcsp[d^r]$ as follows.
    For each subset $V_i$ in the partition introduce a variable $z_i$ to $V'$.
    Note that $\abs{V'} = \ell$.
    
    First, consider a unary constraint $x \neq a$ in $\constraints$.
    Assume $x \in V_i$ and $i_x$ is the index of $x$ in $V_i$.
    Add constraints $z_i \neq t$ to $\constraints'$ for all $t \in D^r$ such that $t_{i_x} = a$.
    Now consider a binary constraint $(x \neq a) \lor (y \neq b)$ in $\constraints$.
    Assume $x \in V_i$, $y \in V_j$ and $i_x, j_y$ are the indices of $x, y$ in the respective subsets.
    If $i \neq j$, then add constraints $(z_i \neq s \lor z_j \neq t)$ for all tuples $s,t \in D^r$ such that $s_{i_x} = a$ and $t_{j_y} = b$.
    If $i = j$ (index of $y$ is $i_y$), add unary constraints $z_i \neq t$ for all $t$ such that $t_{i_x} = a$ and $t_{i_y} = b$. 
    
    Observe that the required constraints $(x \neq c \lor y \neq c)$ for all $c \in D$ are converted into $(z_i \neq s \lor z_j \neq t)$ for every pair of tuples $s,t \in D^r$ that coincide in at least one position.
    Clearly, this includes the case when $s$ and $t$ are equal.
    Hence, $\mathcal{I}'$ is an instance of $\dcspchi[d^r]$.
    Proving the equivalence of $\mathcal{I}$ and $\mathcal{I}'$ is analogous to the proof of Lemma~1~in~\cite{traxler2008time}.
\end{proof}

%Next, we introduce Golomb rulers.\todo{PJ: This is now done in the toolbox.}
%A \textit{Golomb ruler}~\cite{golomb1972number} (also known as a \textit{Sidon %set}~\cite{sidon1932satz}) is a set of integers such that the difference between any %pair of its elements is unique.
%The order of a Golomb ruler is the number of elements in it and the length of a ruler %is the difference between its maximal and minimal elements.
%For example, $\{0,1,4,6\}$ is a Golomb ruler of order $4$ with length $6$.
%There is a way to construct Golomb rulers with length quadratic in their order.
%
%\begin{proposition}[\cite{erdos1941problem}]
%    \label{prop:golombconstruction}
%    Let $p \geq n$ be an odd prime. Then $G_n = \{ pa + (a^2 \bmod p) : a \in \range{0}%{n-1} \}$ is a Golomb ruler.
%\end{proposition}

We are now ready to prove the lower bound.

\begin{theorem}
    \label{thm:d2klowerbound}
    If we assume that the ETH holds, then for arbitrary $m > 0$ there is an integer $k$ such that any algorithm solving $\csp(\disjtemp{2,k}^{\leq})$ requires at least $O^*({m^n})$ time.
\end{theorem}
\begin{proof}
    We start by proving that the time complexity of
    $\dcspchi[d]$
    increases with $d$ assuming the ETH.
    Let 
    $$c_{d} = \inf \{c \in \reals\ : \mbox{there is a $2^{cn}$-time algorithm solving $\dcspchi$} \}.$$
    We show that $\lim_{d \rightarrow \infty} c_d = \infty$.
    The \mbox{\textsc{$3$-Colourability}} problem cannot be solved in subexponential time
    assuming the ETH~\cite{impagliazzo2001problems}.
    $\dcspchi[3]$ is a generalization of \mbox{\textsc{$3$-Colourability}} so $c_3 > 0$.
    By Lemma~\ref{lem:traxler}, for any $r \in \naturals$ we have $c_{3^r} \geq c_{3} \cdot r$.
    Observe that $\lim_{r \rightarrow \infty} c_{3} \cdot r = \infty$, so $\lim_{d \rightarrow \infty} c_d = \infty$.

    Next, we show that for any instance of $\dcspchi$ with $n$ variables 
    there is an equivalent instance of $\csp(\disjtemp{2,k}^{\leq})$ with $n+1$ variables and $k \in O(d^2)$.
    Let $\instance$ be an instance of $\dcspchi$ with $n$ variables.
    Construct an instance $\instance'$ of $\csp(\disjtemp{2,k})$ as follows.
    Choose $k$ so that $\range{-k}{k}$ contains a Sidon set of order $d$ as a subset.
    One can choose such a $k$ to be in $O(d^2)$ by Section~\ref{sec:sidon}.
    Denote the corresponding Sidon set by $G_d$.
    Associate each integer in $\range{1}{d}$ with a unique element of $G_d$ 
    via the bijection $\rho : \range{1}{d} \rightarrow G_d$.
    Introduce zero variable $z$ to express unary relations. 
    For each variable $x$ in $\instance$ introduce a new variable $v_x$.
    Define $E_x = \{\rho(c) : (x \neq c) \in \constraints \}$.
    Restrict the domain of $v_x$ to $D_x = G_d \setminus E_x$ by the constraint $\bigvee_{i \in D_x} v_x - z \in \{i\}$.
    Any constraint $(x \neq a \lor y \neq b)$ can be expressed by enforcing $v_x - v_y \neq \delta$ 
    where $\delta = \rho(a) - \rho(b)$.
    This is done by adding the constraint 
    $$v_x - v_y \in (-\infty, \delta - 1] \lor v_x - v_y \in [\delta + 1, \infty).$$
    By the properties of Sidon sets, this constraint is satisfied if and only if $x \neq \rho(a)$ or $y \neq \rho(b)$.

    The reduction introduces only one extra variable and the lower bound on $\dcspchi$ carries over to $\csp(\disjtemp{2,k}^{\leq})$.
\end{proof}

\section{Parameterized Upper Bounds}
\label{sec:ub}

The results in Section~\ref{sec:lower-bounds-time} imply
that most subproblems of $\csp(\D)$
cannot be solved in subexponential time under
the ETH.
This is a strong motivation for analyzing the
problems from a parameterized perspective.
A highly successful approach for 
identifying tractable fragments of CSPs 
is to restrict variable-constraint interactions
via the underlying 
primal and incidence graphs.
The {\em primal graph} has the variables as its vertices 
with any two joined by an edge if they occur together 
in a constraint. 
The {\em incidence graph} is the bipartite graph with 
two disjoint sets of vertices corresponding to the variables and the constraints, respectively, and
a constraint vertex and a variable vertex are joined by an edge if the variable occurs in the scope of the constraint.
The {\em treewidth} of such graphs has been used extensively
and 
it has been
successfully employed for many application areas.
In particular, it has been used for problems such as
SAT, CSP, and ILP~\cite{DBLP:journals/jair/BliemMMW20,Bodirsky:Dalmau:jcss2013,DBLP:journals/ai/GanianO18,gpw06,Gottlob:etal:ai2002,Huang:etal:ai2013,Samer:Szeider:jcss2010}.
Formal definitions and some auxiliary results are collected
in Section~\ref{sec:treewidth}.

The results in this section (together with the lower bound results that
will be proved in Section~\ref{sec:lb}) are
summarized in Table~\ref{tb:summary}.
The main result is 
 an \XP algorithm for
  $\csp(\D_{\infty,k})$, $k < \infty$, where the parameter
  is the treewidth of the incidence graph. This algorithm
  is presented in Section~\ref{sec:xp-algorithm}.
The treewidth of the incidence graph cannot be larger than
the treewidth of the primal graph plus one so
tractability
results are more general if they hold for incidence treewidth and
hardness results are more general if they hold for primal treewidth.
%There are at least two \XP\ algorithms
%for infinite-domain CSPs appearing in the literature
%(one by Bodirsky \& Dalmau~\cite{Bodirsky:Dalmau:jcss2013} 
%and one by Huang et al.~\cite{Huang:etal:ai2013}) 
%and they could potentially be applicable to 
%$\csp(\D_{\infty,k})$, $k < \infty$.
%We devote Section~\ref{sec:related-work} to these two
%algorithms and we conclude that
%they cannot be used in our setting.

\subsection{Treewidth}
\label{sec:treewidth}

Treewidth
is based on {\em tree decompositions} \cite{Bertele:Brioschi:NDP72,Robertson:Seymour:jctb84}:
a tree decomposition $(T, \chi)$ of an undirected graph $G = (V, E)$ consists 
of a rooted tree $T$ and a mapping $\chi$ from nodes $V(T)$ 
of the tree to subsets of $V$.
The subsets $\chi(t)$ are called {\em bags}.
$T_t$ denotes the sub-tree rooted at $t$,
while $\chi(T_t)$ denotes the set of all vertices occurring in the bags of $T_t$, i.e.
$\chi(T_t) = \bigcup_{s \in V(T_t)} \chi(s)$.
A tree decomposition has the following properties:
\begin{enumerate}
    \item for every $\{u,v\} \in E$, there is a node $t \in
      V(T)$ such that $u,v \in \chi(t)$, and
    \item for every $v \in V$, the set of bags of $T$ containing
      $v$ forms a non-empty sub-tree of $T$.
\end{enumerate}

An example is given in Figure~\ref{fig:tdec}.
The width of a tree decomposition $T$ is $\max \{
\abs{\chi(t)}-1 : t \in T \}$. The {\em treewidth} of a graph $G$,
denoted by $\tw{G}$, is the minimum width of a tree decomposition of
$G$.
It is \NP-complete to determine if a graph has treewidth at most $w$ \cite{Arnborg:etal:sijmaa87} 
but when $w$ is fixed, the graphs with treewidth $w$ can be recognized 
and corresponding tree decompositions can be constructed in linear time~\cite{Bodlaender:sicomp96}.

\begin{figure}
  \centering
  \begin{tikzpicture}[node distance=1cm]
    \tikzstyle{every node}=[draw,fill,black,circle,inner sep=1mm]
    \tikzstyle{every edge}=[draw,line width=0.5mm]
    
    \draw
    node[label=below:$e$] (e) {}
    node[right of=e,label=below:$f$] (f) {}
    node[right of=f,label=below:$g$] (g) {}
    node[right of=g,label=below:$h$] (h) {}
    
    node[above of=e,label=left:$b$] (b) {}
    node[right of=b,label=above:$c$] (c) {}
    node[right of=c,label=above:$d$] (d) {}
    
    node[above of=b,label=above:$a$] (a) {}
    ;
    
    \draw
    (a) edge[] (b)
    (a) edge (c)
    
    (b) edge (c)
    (c) edge[] (d)
    (b) edge (e)
    (c) edge (f)
    (d) edge (f)
    (d) edge (g)
    
    (e) edge (f)
    (f) edge (g)
    (g) edge (h)
    ;

    \begin{scope}[level/.style={sibling distance = 3cm}, level distance = 1cm]
      \tikzstyle{every node}=[draw,ellipse,fill=lightgray,inner sep=0.5mm]
      \tikzstyle{every edge}=[draw,line width=0.5mm]
          
      \draw
      (h) +(5cm,2.1cm) node[] (r) {$c,d,f$}
      child {
        node[] (t11) {$b,c,f$}
        child {
          node[] (t21) {$a,b,c$}
        }
        child {
          node[] (t22) {$b,e,f$}
        }
      }
      child {
        node[] (t12) {$d,f,g$}
        child {
          node (t23) {$g,h$}
        }
      }
      ;
    \end{scope}          
  \end{tikzpicture}
  \caption{A graph (left) and an optimal tree decomposition of the
    graph (right).}
  \label{fig:tdec}
\end{figure}
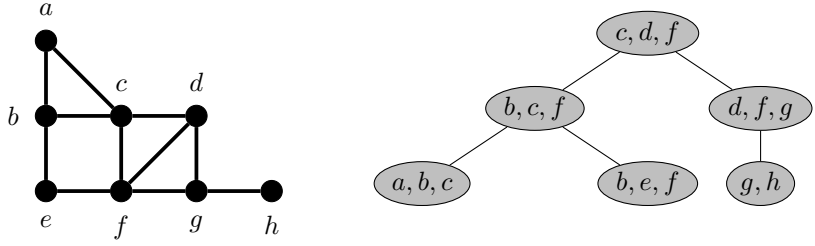

We simplify the presentation by using restricted tree decompositions. A tree decomposition is {\em nice} if
(1)
    $\chi(r) = \varnothing$ for the root $r$ and $|\chi(l)|=1$
     and for all leaf nodes $l$ in $T$, and
    (2)
    every non-leaf node in $T$ is of one of the following types:
    \begin{itemize}
        \item An {\em introduce node}: a node $t$ with exactly one child $t_0$
        such that $\chi(t) = \chi(t_0) \cup \{v\}$ for some $v \in V$.
        \item A {\em forget node}: a node $t$ with exactly one child $t_0$
        such that $\chi(t) = \chi(t_0) \setminus \{v\}$ for some $v \in V$.
        \item A {\em join node}: a node $t$ with exactly two children $t_1$ and $t_2$
        such that $\chi(t) = \chi(t_1) = \chi(t_2)$.
    \end{itemize}

Nice tree decompositions are illustrated in Figure~\ref{fig:ntdec}.
Note that nice tree
decompositions are merely a structured type of tree
decomposition and their only purpose is to simplify the presentation
of dynamic programming algorithms.
It is \NP-complete to determine if a graph has treewidth at most~$w$~\cite{Arnborg:etal:sijmaa87},
but when $w$ is fixed, then graphs with treewidth~$w$ can be recognized 
and corresponding tree decompositions can be constructed in linear time.

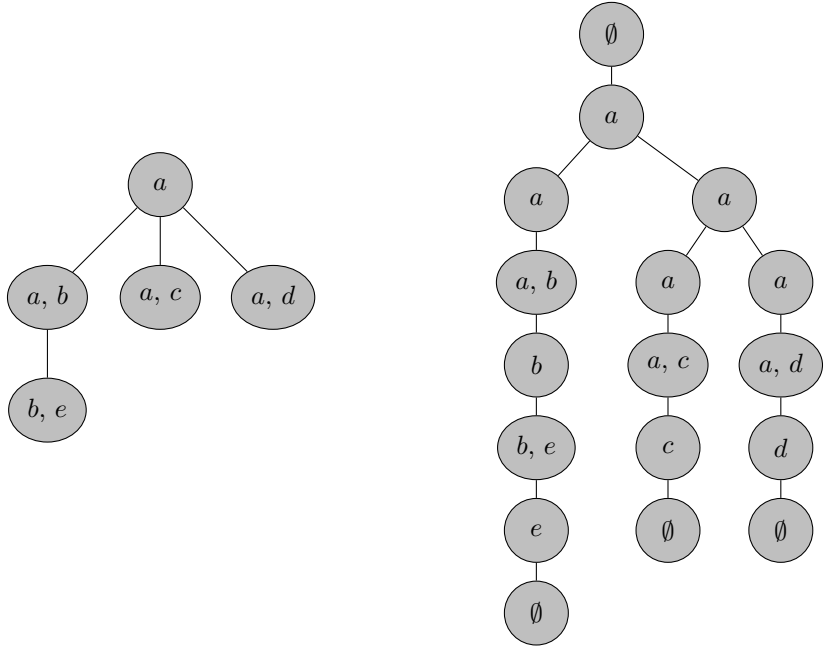
\begin{figure}
  \begin{center}
    \begin{tikzpicture}[level/.style={sibling distance = 1.5cm}, level distance = 1.5cm]
      \tikzstyle{every node}=[]
      \tikzstyle{every edge}=[draw,line width=1mm]
      \tikzstyle{tn}=[draw, ellipse, fill=lightgray,inner sep=1mm, minimum height=0.85cm, minimum width=0.85cm]
      \tikzstyle{tnn}=[tn, line width=1pt]

      \draw
      (0,0) node[tn] (r) {$a$}
      
      child {
        node[tn] (r1) {$a$, $b$}
        child {
          node[tn] (r11) {$b$, $e$}
        }
      }
      child {
        node[tn] (r2) {$a$, $c$}
      }
      child {
        node[tn] (r3) {$a$, $d$}
      }
      ;

      \begin{scope}[level/.style={sibling distance = 1.5cm}, level distance = 1.1cm]
        \tikzstyle{every node}=[]
        \tikzstyle{every edge}=[draw,line width=1mm]
        \tikzstyle{tn}=[draw, ellipse, , fill=lightgray,inner sep=1mm, , minimum height=0.85cm, minimum width=0.85cm]
        \tikzstyle{tnn}=[tn, line width=2pt]
        
        \draw
        (r) +(6cm,2cm) node[tn] (r) {$\emptyset$}
        child {
          node[tn] (rc) {$a$}
          child[sibling distance=2cm] {
            node[tn] (r1jc) {$a$}
            child {
              node[tn] (r1) {$a$, $b$}
              child {
                node [tn] (r1c) {$b$} 
                child {
                  node[tn] (r11) {$b$, $e$}
                  child {
                    node[tn] (r11c) {$e$}
                    child {
                      node[tn] (r11cl) {$\emptyset$}
                    }
                  }
                }
              }
            }
          }
          child[sibling distance=3cm] {
            node[tn] (rj) {$a$}
            child {
              node[tn] (r2jc) {$a$}
              child {
                node[tn] (r2) {$a$, $c$}
                child {
                  node[tn] (r2c) {$c$}
                  child {
                    node[tn] (r2cl) {$\emptyset$}
                  }
                }
              }
            }
            child {
              node[tn] (r3jc) {$a$}
              child {
                node[tn] (r3) {$a$, $d$}
                child {
                  node[tn] (r3c) {$d$}
                  child {
                    node[tn] (r3cl) {$\emptyset$}
                  }
                }
              }
            }
          }
        }
        ;

      \end{scope}
    \end{tikzpicture}%
  \end{center}
  \caption{A tree decomposition (left) and a corresponding nice tree
    decomposition (right). }
  \label{fig:ntdec}
\end{figure}

\begin{proposition}[{Bodlaender \& Kloks~\cite{bodlaender1996efficient}}; Kloks~\cite{Kl94}]\label{pro:comp-ntd}
  Let $G=(V,E)$ be a graph. For fixed $w$, if~$G$ has treewidth at most~$w$, then 
 a nice tree decomposition
  of width at most $w$ with $\bigoh(|V|)$ nodes can be
  computed in linear time.
\end{proposition}
% Given a tree decomposition $T$ of a graph $G = (V,E)$ with $O(\abs{V})$ nodes,
% one can construct a nice tree decomposition of the same width 
% and $O(\abs{V})$ nodes in linear time~\cite{bodlaender1996efficient}. 

Let $\inst=(\variables,\constraints)$ be an instance of
$\csp({\bf A})$. The \emph{primal graph} of $\inst$, denote by $\pg(\inst)$,
is the graph with vertices $\variables$ having an edge
between two variables if both appear together in the scope of one
constraint. The incidence graph of $\inst$, denote by $\ig(\inst)$, is 
the bipartite graph with $\variables$ on one side and $\constraints$
on the other side having an edge between a variable and a constraint
if the variable appears in the scope of the constraint.
It is well known that the
treewidth of the incidence graph is at most equal to the treewidth of
the primal graph plus one~\cite{Kolaitis:Vardi:jcss2000} and that the
incidence treewidth can be arbitrary smaller than the primal
treewidth.
This means that tractability
results are more general if they hold for incidence treewidth and
hardness results are more general if they hold for primal treewidth.

We continue with a few observations concerning treewidth.
We discussed in Section~\ref{sec:csp} that zero variables can be used for simulating
unary constraints, and that they do not
affect the time complexity with more than a multiplicative factor.
We have a similar situation in the parameterized setting: adding a zero
variable can increase the
treewidth of the incidence graph by at most $1$ so
algorithmic results (such as the forthcoming Theorem~\ref{the:xp-alg}) are still valid for
this extended formalism.
We will use a related observation when we argue about the treewidth of the graphs
obtained in our hardness result (Theorem~\ref{thm:w1-hard}).

\begin{proposition}[{Bodlaender~\cite{Bodlaender:sicomp96}}]\label{pro:tw-del}
  Let $G$ be a graph and $B \subseteq V(G)$. Then the treewidth of $G$
  is at most $|B|+\tw{G-B}$.
\end{proposition}

\subsection{\XP\ Algorithm for CSP$(\D_{\infty,k})$}
\label{sec:xp-algorithm}

We are now ready to present our dynamic programming algorithm for 
$\csp(\D_{\infty,k})$.
\begin{theorem} \label{the:xp-alg}
  $\csp(\D_{\infty,k})$ can be solved in $(nk)^{\bigoh(w)}$ time, where $w$ is
  the treewidth of the incidence graph and $n$ is the number of
  variables.
  % There is an algorithm solving instances of $\csp{\D_{\infty,k}}$
  % with $n$ variables and incidence tree decompositions of width $w$
  % in $(nk)^{O(w)}$ time.
\end{theorem}

Note that the bound implies that $\csp(\D)$ is in \XP\ whenever
the numeric values are bounded by a polynomial in the input size.
Proposition~\ref{pro:comp-ntd} implies that the computation of a nice tree
decomposition of the incidence graph does not incur an additional
run-time overhead. We may thus assume that a nice tree
decomposition is provided in the input,
and it is hence
sufficient to show the following.
\begin{theorem}
  \label{the:xp-alg-td}
  $\csp(\D_{\infty,k})$ can be solved in time $(nk)^{\bigoh(w)}$ provided that
  a nice tree decomposition of the incidence graph of width at most
  $w$ is given as part of the input.
\end{theorem}

Let $\inst = (V,C)$ be an instance of $\csp(\D_{\infty,k})$ 
with $n$ variables and assume $(T, \chi)$ is a nice tree decomposition
of the incidence graph of $\inst$ of width $w$.
Bags of this decomposition contain vertices corresponding to both
variables and constraints. 
To distinguish between them, we use $\var{t}$ 
to denote all variables in the bag $\chi(t)$ and 
$\con{t}$ to denote all constraints in $\chi(t)$. 
These definitions naturally extend to the subsets of $V(T)$.
Note that by Theorem~\ref{lem:compact-assign},
we may assume that every solution for $\inst$
maps the variables into the set $\cdf=\cd{n,k}$.

Intuitively, the algorithm behind Theorem~\ref{the:xp-alg} works as
follows. It uses bottom-up dynamic programming on
the nodes of $T$ starting from the leaves and finishing at the root.
It computes a compact representation, 
represented by a set of valid records, of all solutions to
$\inst$ restricted to the variables and constraints in $\chi(T_t)$ for every
node $t \in V(T)$.

A record for $t \in V(T)$ is a pair $(\rva,\rca)$ where
\begin{itemize}
\item $\rva: \var{t} \rightarrow \cdf$ is an assignment of values in $\cdf$
to the variables in $\var{t}$, and 
\item $\rca : \con{t} \rightarrow \rcd$, where $\rcd=\{S,U\} \cup \SB (v,d) \SM v \in V
  \textnormal{ and } d \in \cdf\SE$ such that for every constraint $c \in
  \con{t}$ either:
  \begin{itemize}
  \item $\rca(c)=S$ signaling that the constraint $c$ is
    already satisfied, 
  \item $\rca(c)=U$ signaling that the
    constraint $c$ is not yet satisfied,
  \item $\rca(c)=(v,d)$, where $v \in \scope{c}\cap
    (\var{T_t}\setminus \var{t})$ and $d \in \cdf$,
    signaling that $c$ is not yet satisfied, but satisfying $c$ can
    use the assumption that $v$ is set to $d$. This also means that
    $c$ will be satisfied by satisfying a simple constraint on $v$
    and some variable in $V \setminus \var{T_t}$.
  \end{itemize}
\end{itemize}
Note that there are at most $|\cdf|$ possible choices for every
variable in $\var{t}$ and at most $|V||\cdf|+2$ possible choices for
every constraint in $\con{t}$. Therefore, the total number of valid
records for $t$ is at most $(|V||\cdf|+2)^{w+1}$.

% For a function $f$, let $f^{-1}$ be its inverse. For
% the special function $\rca : \con{t} \rightarrow \rcd$, we let
% $\rca^{-1}(F)$ be the set of all constraints $c$ in $\con{t}$ such that
% $\rca(c) \notin \{S,U\}$, i.e. $\rca(c)=(v,d)$ for some $v \in V$ and
% $d \in \cdf$.
For $X \in \{S,U\}$, define the inverse 
$\rca^{-1}(X)$ as $\{ c \in \con{t} \mid \rca(c) = X \}$
and let 
$$\rca^{-1}(F) = \con{t} \setminus (\rca^{-1}(S) \cup \rca^{-1}(U)),$$
i.e. $\rca(c)=(v,d)$ for some $v \in V$ and $d \in \cdf$ for all $c \in \rca^{-1}(F)$.

The semantic of a record is defined as follows. We say that a record
$(\rva,\rca)$ is \emph{valid} for $t$ if there is an 
assignment $\ria : \var{T_t} \rightarrow \cdf$ such that:

\begin{enumerate}[(R1)]
% (R1)  $\tau$ satisfies all constraints in $(\con{T_t} \setminus (\con{t} \setminus
%   \rca^{-1}(S))$ and does not satisfy any constraint in $\con{t}\setminus \rca^{-1}(S)$,
\item  $\ria$ does not satisfy any constraint in $Y = \con{t} \setminus \rca^{-1}(S)$
and satisfies all constraints in $\con{T_t} \setminus Y$,
  
\item $\ria(v)=\rva(v)$ for every $v \in \var{t}$, and

\item $\ria(v)=d$ holds for every constraint $c \in \con{t}$ with $\rca(c)=(v,d)$.
\end{enumerate}

Let $\RRR(t)$ be the set of all valid records for $t$.
Note that~$\inst$ has a solution if and only if $\RRR(r)\neq
\emptyset$ for the root $r$ of $T$ since the records in $\RRR(r)$ 
represent solutions for the whole instance.
Moreover, once we have computed the set of records for all nodes, a
straightforward application of standard
techniques~\cite{DowneyFellows13} can be used to obtain a solution
for $\inst$ using a second top-to-bottom run
through the tree-decomposition.

Next, we will show that $\RRR(t)$ can be computed via a dynamic
programming algorithm on $(T,\chi)$ in a bottom-up manner. The
algorithm starts by computing the set of all valid records for the
leaves of $T$ and then proceeds by computing the set of all valid records
for the other three types of nodes of a nice tree-decomposition
(always selecting nodes all of whose children have already been
processed). The following lemmas show how this is achieved for the
different types of nodes of $(T,\chi)$.

\begin{lemma}[\textbf{variable leaf node}]\label{lem:DP-leaf-v}
  Let $t \in V(T)$ be a leaf node with $\chi(t)=\{v\}$ for some
  variable $v\in V$. Then, $\RRR(t)$ can be computed in
  $\bigoh(|\cdf|)$ time.
\end{lemma}
\begin{proof}
  $\RRR(t)$ consists of all records $(\rva, \emptyass)$ for every
  assignment $\rva : \{v\}\rightarrow \cdf$, so $\RRR(t)$ can
  be computed by enumerating all assignments $\rva : \{v\}\rightarrow
  \cdf$ for $v$ in $\bigoh(|\cdf|)$ time. Correctness follows
  immediately from the definition of valid records.
\end{proof}

\begin{lemma}[\textbf{constraint leaf node}]\label{lem:DP-leaf-c}
  Let $t \in V(T)$ be a leaf node with $\chi(t)=\{c\}$ for some
  constraint $c\in C$. Then, $\RRR(t)$ can be computed in
  $\bigoh(1)$ time.
\end{lemma}
\begin{proof}
  $\RRR(t)$ consists of the record $(\emptyass, \rca)$, where $\rca :
  \{c\} \rightarrow \rcd$ is defined by setting $\rca(c)=U$.
  Thus, $\RRR(t)$ can
  be computed in constant time and correctness follows
  immediately from the definition of valid records.
\end{proof}

\begin{lemma}[\textbf{variable introduce node}]\label{lem:DP-intro-v}
  Let $t \in V(T)$ be an introduce node with child $t_0$ such that
  $\chi(t)\setminus \chi(t_0)=\{v\}$ for some
  variable $v\in V$. Then, $\RRR(t)$ can be computed in
  $\bigoh(|\RRR(t_0)| \cdot |\cdf| \cdot |\inst|)$ time.
\end{lemma}
\begin{proof}
  Informally, the set $\RRR(t)$ is obtained from $\RRR(t_0)$ by
  extending every record $R_0=(\rva_0,\rca_0)$ in $\RRR(t_0)$ with an assignment $\rva_v :
  \{v\} \rightarrow \cdf$ for the variable $v$ and then updating the
  record (i.e. updating $\rca_0$) if $\rva_v$ causes additional constraints to be satisfied.
  More formally, for every $(\rva_0,\rca_0) \in \RRR(t_0)$ and every assignment
  $\rva_v : \{v\} \rightarrow \cdf$, the set $\RRR(t)$ contains the
  record $(\rva,\rca)$, where:
  \begin{itemize}
  \item $\rva(u)=\rva_0(u)$ for all $u \in \chi(t_0)$ and $\rva(v)=\rva_v(v)$,
  \item $\rca(c)=S$ for every constraint $c \in \rca_0^{-1}(S) \cup U'
    \cup F'$, where:
    \begin{itemize}
    \item $U'$ is the set of all constraints $c \in \rca_0^{-1}(U)$ that
      are satisfied by the (partial) assignment $\rva$ and
    \item $F'$ is the set of all constraints $c \in \rca_0^{-1}(F)$
      that are satisfied by setting $v$ to $\rva_v(v)$ and $u$ to $d$,
      where $(u,d)=\rca_0(c)$.
    \end{itemize}
  \item $\rca(c)=\rca_0(c)$ for every other constraint $c$, i.e. every
    constraint $c \in \con{t}\setminus (\rca_0^{-1}(S) \cup U'\cup F')$.
  \end{itemize}

  Towards showing correctness of the definition for $\RRR(t)$, we first show
  that every valid record $R=(\rva,\rca)$ for $t$ is added to
  $\RRR(t)$. Because $R$ is valid, there is an 
  assignment $\ria : \var{T_t} \rightarrow \cdf$ satisfying
  (R1)--(R3). Let $\rva_0$ be the restriction of $\rva$ to $\var{t_0}$
  and let $\ria_0$ be the restriction of $\ria$ to $\var{T_{t_0}}$. Let
  $Z$ be the set of all constraints in $\con{t}=\con{t_0}$ that are
  satisfied by $\ria$ but not satisfied by $\ria_0$. Moreover, 
  let $X \subseteq Z$ contain the constraints that are
  satisfied by $\rva$ and set $Y=Z\setminus X$. Then, for every
  constraint $c \in Y$, there is (at least one) variable, denoted by~$y(c)$, in
  $\var{T_t}\setminus \var{t}$ such
  that the partial assignment setting $y(c)$ to $\ria(y(c))$ and setting $v$
  to $\rva(v)$ satisfies~$c$. This implies that the record
  $R_0=(\rva_0,\rca_0)$ defined by setting $\rca_0(c)=\rca(c)$ for every $c
  \in \con{t_0}\setminus (X\cup Y)$, $\rca_0(c)=U$ for every $c \in
  X$, and $\rca_0(c)=(y(c),\ria(y(c)))$ for every $c \in Y$ is
  contained in $\RRR(t_0)$. Finally, $U'=X$ and $F'=Y$ holds
  for the record $R_0$, so $R$ is added to $\RRR(t)$.

  It remains to show that if a record $R=(\rva,\rca)$ is added to
  $\RRR(t)$, then $R$ is valid for $t$. Suppose that $R$ is obtained
  from the record $R_0=(\rva_0,\rca_0) \in \RRR(t_0)$.
  There is an assignment $\ria_0 :
  \var{T_{t_0}} \rightarrow \cdf$ satisfying (R1)--(R3) since $R_0$ is valid for $t_0$. Now it is
  straightforward to verify that the extension $\ria$ of $\ria_0$
  obtained by setting $\ria(v)=\rva(v)$ witnesses that $R$ is a valid
  record.

  Finally, the run-time of the procedure follows because there are
  $|\RRR(t_0)| \cdot |\cdf|$ pairs of records in $\RRR(t_0)$ and assignments
  $\rva_v$ for $v$. Computing the record for one
  such combination requires evaluating the constraints in $\con{t}$ for partial
  assignments and thus takes $\bigoh(|\inst|)$ time.
\end{proof}

\begin{lemma}[\textbf{constraint introduce node}]\label{lem:DP-intro-c}
  Let $t \in V(T)$ be an introduce node with child $t_0$ such that
  $\chi(t)\setminus \chi(t_0)=\{c\}$ for some
  constraint $c\in C$. Then, $\RRR(t)$ can be computed in
  $\bigoh(|\RRR(t_0)||\inst|)$ time.
\end{lemma}

\begin{proof}
  Informally, the set $\RRR(t)$ is obtained from $\RRR(t_0)$ by
  checking, for every record $(\rva_0,\rca_0) \in \RRR(t_0)$, whether $\rva_0$
  satisfies the constraint $c$ and if so, extending $\rca_0$ by setting
  $c$ to being satisfied, and if not, extending $\rca_0$ by setting $c$
  to being unsatisfied.
  More formally, for every record $(\rva_0,\rca_0) \in \RRR(t_0)$:
  \begin{itemize}
  \item if the constraint $c$ is satisfied by the partial assignment~$\rva_0$, then $\RRR(t)$ contains the record $(\rva_0,\rca)$, where
    $\rca$ is the extension of $\rca_0$ that sets $c$ to $S$.
  \item otherwise, i.e. if $\rva_0$ does not satisfy $c$, then $\RRR(t)$
    contains the record $(\rva_0,\rca)$, where $\rca$ is the
    extension of $\rca_0$ that sets $c$ to $U$.
  \end{itemize}
  % Towards showing correctness of the definition for $\RRR(t)$, we
  % first note that a record $R=(\rva,\rca)$ is added to $\RRR(t)$ if
  % and only if there is a record $R_0=(\rva,\rca_o) \in \RRR(t_0)$ and
  % $\rca(c)=S$ if $\rva$ satisfies $c$ and $\rca(c)=U$, otherwise. Now,
  % $R_0 \in \RRR(t_0)$ if and only if there is an 
  % assignment $\ria_0 : \var{T_{t_0}} \rightarrow \cdf$ satisfying
  % (R1)--(R3). Because $(T,\chi)$ is a tree decomposition, it follows
  % that the scope of $c$ does not contain any variable from
  % $\var{T_t}\setminus \var{t}$. Therefore, $\ria_0$ satisfies $c$ if
  % and only if already $\rva$ satisfies $c$. This shows that $\ria_0$
  % witnesses the validity of $R$.
  Towards showing correctness of the definition for $\RRR(t)$, we first show
  that every valid record $R=(\rva,\rca)$ for $t$ is added to
  $\RRR(t)$. Because $R$ is valid, there is an 
  assignment $\ria : \var{T_t} \rightarrow \cdf$ satisfying
  (R1)--(R3). Because $(T,\chi)$ is a tree decomposition, it follows
  that the scope of $c$ does not contain any variable from
  $\var{T_t}\setminus \var{t}$; otherwise the edge between $c$ and the variable in
  $\var{T_t}\setminus \var{t}$ in the incidence graph
  is not contained in any bag of $T$. Therefore, $\rca(c) \in \{S,U\}$ so
  if $\rca(c)=S$, then $c$ is already satisfied by the
  partial assignment $\rva$. It follows that $\ria$ witnesses that the
  record $(\rva,\rca_0)$, where $\rca_0$ is the restriction of $\rca$
  to $\con{t'}$, is in $\RRR(t_0)$ and $R$ is added to $\RRR(t)$.

  It remains to show that if a record $R=(\rva,\rca)$ is added to
  $\RRR(t)$, then $R$ is valid for $t$. Assume that $R$ is obtained
  from the record $R_0=(\rva,\rca_0) \in \RRR(t_0)$. We know that $R_0$ is valid for $t_0$ so there is an assignment $\ria_0 :
  \var{T_{t_0}} \rightarrow \cdf$ satisfying (R1)--(R3). Because
  $(T,\chi)$ is a tree decomposition, it follows
  that the scope of $c$ does not contain any variable from
  $\var{T_t}\setminus \var{t}$. Hence, $\rca(c) \in \{S,U\}$ so
  if $\rca(c)=S$, then $c$ is already satisfied by the
  partial assignment $\rva$. Therefore, $\ria$ witnesses that $R$ is
  valid.
  
  % The correctness of the definition of $\RRR(t)$ follows from the
  % definition of the semantics of the records together with the fact
  % that all neighbours of $a$ in $\chi(T_t)$ are inside of $\chi(t')$.
  Finally, the run-time follows because we have to consider
  every record $(\rva_0,\rca_0)$ in $\RRR(t_0)$ and we can check in
  $\bigoh(|\inst|)$ time whether $\rva$ satisfies $c$ or not.
\end{proof}

\begin{lemma}[\textbf{variable forget node}]\label{lem:DP-forget-v}
  Let $t \in V(T)$ be a forget node with child $t_0$ such that
  $\chi(t_0)\setminus \chi(t)=\{v\}$ for some
  variable $v\in V$. Then, $\RRR(t)$ can be computed in
  $\bigoh(|\RRR(t_0)|2^ww)$ time.
\end{lemma}
\begin{proof}
  Informally, $\RRR(t)$ is obtained from $\RRR(t_0)$ by
  restricting $\rva_0$ of every record $(\rva_0,\rca_0) \in \RRR(t_0)$
  to $\var{t}$, but allowing the assignment that sets $v$ to
  $\rva_0(v)$ to satisfy any set of yet unsatisfied constraints in
  $\rca_0^{-1}(U)$ that have $v$ in their scope.
  More formally, for every record $(\rva_0,\rca_0) \in \RRR(t_0)$ and every subset
  $U'$ of $\rca_0^{-1}(U)\cap \SB c \in C \SM v \in \scope{c}\SE$, the set
  $\RRR(t)$ contains the record $(\rva, \rca)$, where $\rva$ is the
  restriction of $\rva_0$ to $\var{t}$ and $\rca$ is defined by setting
  $\rca(c)=\rca_0(c)$ for every $c \in \con{t}\setminus U'$ and
  $\rca(c)=(v,\rva_0(v))$ for every $c \in U'$.

  Towards showing the correctness of the definition for $\RRR(t)$, we first show
  that every valid record $R=(\rva,\rca)$ for $t$ is added to
  $\RRR(t)$. Because $R$ is valid, there is an 
  assignment $\ria : \var{T_t} \rightarrow \cdf$ satisfying
  (R1)--(R3). Let $X$ be the set of all constraints $c$ in $\con{t}$
  such that $\rca(c)=(v,d)$. We know that $\ria$ satisfies (R3) so $d=\ria(v)$ for all constraints in $X$. 
  % Then,
  % $\ria$ witnesses that the constraint $R_0=(\rva_0,\rca_0)$, where
  % $\rva_0$ is the extension of $\rva$ setting $v$ to $\ria(v)$ and
  % $\rca_0$ is obtained from $\rca$ after setting $\rca_0(c)=U$ for
  % every $c \in X$. But then the record $R'$ together with the set
  % $U'=X$ shows that $R$ is added to $\RRR(t)$.
  Then, $\ria$ witnesses validity of the record $R_0=(\rva_0,\rca_0)$,
  where $\rva_0$ is the extension of $\rva$ setting $v$ to $\ria(v)$ and
  $\rca_0$ is obtained from $\rca$ by setting $\rca_0(c)=U$ for
  every $c \in X$. Now, the record $R_0$ together with the set
  $U'=X$ shows that $R$ is added to $\RRR(t)$.

  It remains to show that if a record $R=(\rva,\rca)$ is added to
  $\RRR(t)$, then $R$ is valid for $t$. Assume that $R$ is obtained
  from the record $R_0=(\rva_0,\rca_0) \in \RRR(t_0)$. Then,
  because $R_0$ is valid for $t_0$, there is an assignment $\ria_0 :
  \var{T_{t_0}} \rightarrow \cdf$ satisfying (R1)--(R3). Moreover,
  the assignment $\ria_0$ witnesses the validity of $R$.
  % The correctness of the definition of $\RRR(t)$ follows from the
  % definition of the semantics of the records together with the fact
  % that all neighbours of $a$ in $\chi(T_t)$ are inside of $\chi(t')$.
  Finally, the run-time follows because there are
  at most $|\RRR(t_0)|2^w$ pairs of a record in $\RRR(t_0)$ and a subset
  $U'$ and the time required to compute a record for such a pair is at
  most $\bigoh(w)$. 
\end{proof}

\begin{lemma}[\textbf{constraint forget node}]\label{lem:DP-forget-c}
  Let $t \in V(T)$ be a forget node with child $t_0$ such that
  $\chi(t_0)\setminus \chi(t)=\{c\}$ for some
  constraint $c\in C$. Then, $\RRR(t)$ can be computed in
  $\bigoh(|\RRR(t_0)| |\chi(t)|)$ time.
\end{lemma}
\begin{proof}

Informally, $\RRR(t)$ is obtained from $\RRR(t_0)$ by
  taking all records $(\rva_0,\rca_0)$ in $\RRR(t_0)$ that satisfy $c$ and
  restricting $\rca_0$ to $\con{t}$.
  More formally, for every record $(\rva_0,\rca_0) \in \RRR(t_0)$ such that
  $\rca_0(c)=S$, $\RRR(t)$ contains the record $(\rva_0,\rca)$, where
  $\rca$ is the restriction of $\rca_0$ to $\con{t}$.

  Towards showing the correctness of the definition for $\RRR(t)$, we first show
  that every valid record $R=(\rva,\rca)$ for $t$ is added to
  $\RRR(t)$. Because $R$ is valid, there is an 
  assignment $\ria : \var{T_t} \rightarrow \cdf$ satisfying
  (R1)--(R3). Because $\ria$ satisfies (R1), it also satisfies
  the constraint $c$. Therefore the record
  $R_0=(\rva,\rca_0)$, where $\rca_0$ is the extension of $\rca$ to
  $c$ by setting $\rca_0(c)=S$, is valid and hence
  $R_0\in\RRR(t_0)$. Therefore, $R$ is added to $\RRR(t)$. 

  It remains to show that if a record $R=(\rva,\rca)$ is added to
  $\RRR(t)$, then $R$ is valid for $t$. Assume that $R$ is obtained
  from the record $R_0=(\rva_0,\rca_0) \in \RRR(t_0)$. Then,
  because $R_0$ is valid for $t_0$, there is an assignment $\ria_0 :
  \var{T_{t_0}} \rightarrow \cdf$ satisfying (R1)--(R3). Moreover,
  from the definition of $\RRR(t)$, we obtain that $\ria_0$ satisfies
  $c$ so $\ria_0$ witnesses that $R$ is valid.

  % The correctness of the definition of $\RRR(t)$ follows from the
  % definition of the semantics of the records together with the fact
  % that all neighbours of $a$ in $\chi(T_t)$ are inside of $\chi(t')$.
  Finally, the run-time estimate is correct because it takes $\bigoh(|\chi(t)|)$ time to
  check whether $\rca_0(c)=S$ and to
  compute the restriction of $\rca$ to $\con{t}$ for a
  record $(\rva_0,\rca_0)$ in $\RRR(t_0)$.
\end{proof}

\begin{lemma}[\textbf{join node}]\label{lem:DP-join}
  Let $t \in V(T)$ be a join node with children $t_1$ and $t_2$, where
  $\chi(t)=\chi(t_1)=\chi(t_2)$. Then, $\RRR(t)$ can be computed in 
  $\bigoh(|\RRR(t_1)||\RRR(t_2)||\inst|)$ time.
\end{lemma}
\begin{proof}

Informally, $\RRR(t)$ is obtained from $\RRR(t_1)$ and $\RRR(t_2)$ by
  combining all pairs of records $(\rva_i,\rca_i)$ in $\RRR(t_i)$ that
  agree on the assignments $\rva_i$ to a new record and updating the
  set of satisfied constraints.
  More formally, we say that two records $(\rva_1,\rca_1) \in
  \RRR(t_1)$ and $(\rva_2,\rca_2) \in \RRR(t_2)$ are \emph{compatible}
  if $\rva_1=\rva_2$ and for every constraint $c\in \con{t}$ such that
  for $i \in \{1,2\}$, $\rca_i(c)=(v_i,d_i)$, and 
  the partial assignment setting $v_i$ to $d_i$
  satisfies $c$. Then, for every pair of compatible records $(\rva_1,\rca_1) \in
  \RRR(t_1)$ and $(\rva_2,\rca_2) \in \RRR(t_2)$, the set $\RRR(t)$
  contains the record $(\rva,\rca)$, where:
  \begin{itemize}
  \item $\rva=\rva_1=\rva_2$ and
  \item $\rca(c)=S$ if either:
    \begin{itemize}
    \item $\rca_1(c)=S$ or $\rca_2(c)=S$ or
    \item $\rca_1(c)=(v_1,d_1)$ and $\rca_2(c)=(v_2,d_2)$ and the
      (partial) assignment setting $v_1$ to $d_1$ 
      and $v_2$ to $d_2$ satisfies $c$.
    \end{itemize}
  \item $\rca(c)=U$ if $\rca_1(c)=U$ and $\rca_2(c)=U$,
  \item $\rca(c)=(v,d)$ if either:
    \begin{itemize}
    \item $\rca_1(c)=(v,d)$ and $\rca_2(c)=U$ or
    \item $\rca_1(c)=U$ and $\rca_2(c)=(v,d)$
    \end{itemize}
  \end{itemize}

  We will now show the correctness of the definition of $\RRR(t)$. We first show
  that every valid record $R=(\rva,\rca)$ for $t$ is added to
  $\RRR(t)$. Because $R$ is valid, there is an 
  assignment $\ria : \var{T_t} \rightarrow \cdf$ satisfying
  (R1)--(R3). Let $\ria_i$ be the restriction of $\ria$ to
  $\var{T_{t_i}}$. Note first that every constraint $c \in
  \con{T_t}\setminus \con{t}$, is either satisfied by $\ria_1$ or by
  $\ria_2$. This is because $(T,\chi)$ is a tree decomposition so $\scope{c} \subseteq \var{T_{t_i}}$ for some $i \in
  \{1,2\}$. Moreover, every constraint $c \in \con{t}$ that is
  satisfied by $\ria$ is either (1) already satisfied by $\ria_i$ (for
  some $i\in \{1,2\}$) or (2) it is satisfied by a simple constraint
  involving two variables $v_i \in \var{T_i}\setminus \var{t}$
  assigned according to $\ria$. We set $\rca_i(c)=S$ if $c$ is satisfied
  by $\ria_i$. Otherwise, we set $\rca_i(c)=U$ if $\rca(c)=U$ or $\rca(c)=S$, and $\ria_{3-i}$ satisfies $c$. Finally, we set $\rca_i(c)=(v_i,\ria(v_i))$ if
  either $\rca(c)=(v_i,\ria(v_i))$ and $v_i \in \var{T_{t_i}}\setminus
  \var{t}$, or $\rca(c)=S$ but neither $\ria_1$ nor $\ria_2$ satisfy
  $c$ and setting $v_i$ to $d_i$ satisfies $c$. Then, the records
  $R_i=(\rva,\rca_i)$ are valid for $t_i$ as witnessed by $\ria_i$, so 
  $(\rva,\rca_i)\in \RRR(t_i)$. Moreover, $R_1$ and $R_2$
  are compatible and $R$ is the result of combining $R_1$ and $R_2$,
  showing that $R$ is added to $\RRR(t)$.
  
  It remains to show that if a record $R=(\rva,\rca)$ is added to
  $\RRR(t)$, then $R$ is valid for $t$. Assume that $R$ is obtained
  from two compatible records $R_i=(\rva,\rca_i) \in \RRR(t_i)$. Then,
  because $R_i$ is valid for $t_i$, there is an assignment $\ria_i :
  \var{T_{t_i}} \rightarrow \cdf$ satisfying (R1)--(R3). It is now
  straightforward to verify that the assignment $\ria$ obtained by
  combining $\ria_1$ and $\ria_2$ witnesses that $R$ is valid for $t$.

  % The correctness of the definition of $\RRR(t)$ follows from the
  % definition of the semantics of the records together with the fact
  % that all neighbours of $a$ in $\chi(T_t)$ are inside of $\chi(t')$.
  Finally, the run-time follows because there are at most
  $|\RRR(t_1)||\RRR(t_2)|$ compatible pairs of records and for every
  such pair it takes time at most $\bigoh(|\inst|)$ to compute the
  combined record for $\RRR(t)$.
\end{proof}

We can now conclude the results in this section.

\begin{proof}[Proof of Theorem~\ref{the:xp-alg}]
  The algorithm computes the
  set of all valid records $\RRR(t)$ for every node $t$ of $T$ using a
  bottom-up dynamic programming algorithm starting in the leaves of
  $T$. It then solves $\inst$ by checking whether $\RRR(r)\neq \emptyset$.
  The
  correctness of the algorithm follows from 
   Lemmas~\ref{lem:DP-leaf-v}--\ref{lem:DP-join}. 
  The run-time 
  of the algorithm is at most the number of nodes of $T$, which can be
  assumed to be bounded from above by $\bigoh(|\inst|)$ (Proposition~\ref{pro:comp-ntd}),
  times the maximum time required to compute $\RRR(t)$ for any of the
  node types of a nice tree-decomposition, which is obtained for join
  nodes with a run-time of $\bigoh(|\RRR(t_1)||\RRR(t_2)||\inst|)$. 
  It follows that
  $$\bigoh((|V||\cdf|+2)^{2(w+1)}(|\inst|)^2)\in (nk)^{\bigoh(w)}$$ is the total run-time because $|\RRR(t)|\leq (|V||\cdf|+2)^{w+1}$.
  %which because of
  % lemmas~\ref{lem:tw-dp-leaf}, \ref{lem:tw-dp-intro}, \ref{lem:tw-dp-forget},
  % and~\ref{lem:tw-dp-join} is at most $\bigoh(|\VR(t)|^2|\RT|\tw(G))$.
  % Because $|\VR(t)|\leq |R|^{|\RT|(2\tw(G)+1)}$, we obtain
  % $\bigoh(|R|^{2|\RT|(2\tw(G)+1)}(\tw(G))^2|\RT||A|)$ as the total running-time of the algorithm.
\end{proof}

% \begin{theorem} \label{thm:xp}
%     $\csp{\D}$ restricted to instances $\inst$ such that
%     $\num{\inst} \in \norm{\inst}^{O(1)}$ is in \XP. 
% \end{theorem}
% \begin{proof}
% Instances in the definition of the theorem can be viewed
% as instances of the problem $\csp{\D_{\infty, k}}$,
% where $k = \num{\inst}$.
% Consider the running time of the algorithm in Lemma~\ref{lem:xp-alg}.
% Denoting the number of variables in $\inst$ by $n$,
% we observe that $kn \in \norm{\inst}^{O(1)}$ and, 
% hence, the theorem holds. 
% \end{proof}

\section{Parameterized Lower Bounds} 
\label{sec:lb}

This section contains two main results:
we show that $\csp(\D_{2,\infty})$ is \paraNP-hard (Section~\ref{sec:pnp-hardness})
and that $\csp(\D_{2,1})$ is 
$\W{1}$\hy hard (Section~\ref{sec:w1-hardness}) when parameterized
by primal treewidth.
These results indicate that there is no fpt algorithm
for $\csp(\D_{2,k})$, $k \in \{1,2,\dots\} \cup \{\infty\}$,
under standard complexity-theoretic assumptions.
The \paraNP-hardness result is proved by a direct
reduction from \probfont{Subset Sum}
while the $\W{1}$\hy hardness result
is based on a reduction from a variant of the
\probfont{Subset Sum} problem
that we call
\probfont{Multi-Dimensional Partitioned Subset Sum}.
$\W{1}$\hy hardness for the latter problem is proved in
Section~\ref{sec:mdpss}.
The full picture of parameterized upper and
lower bounds is summarized in Table~\ref{tb:summary}.

To present our results, we need some
additional technical machinery.
A parameterized problem is, formally speaking, a subset of $\Sigma^* \times {\mathbb N}$
where $\Sigma$ is the input alphabet. Reductions between parameterized problems need to take
the parameter into account. To this end, we will use {\em parameterized reductions} (or fpt-reductions).
Let $L_1$ and $L_2$ denote parameterized problems with $L_1 \subseteq \Sigma_1^* \times {\mathbb N}$
and $L_2 \subseteq \Sigma_2^* \times {\mathbb N}$. 
A parameterized reduction from $L_1$ to $L_2$ is a
mapping $P: \Sigma_1^* \times {\mathbb N} \rightarrow \Sigma_2^* \times {\mathbb N}$
such that
\begin{enumerate}[(1)]
  \item 
  $(x, k) \in  L_1$ if and only if $P((x, k)) \in L_2$, 
  \item 
  the mapping can be computed by an fpt-algorithm 
  with respect to the parameter $k$, and 
  \item there is a computable function $g : {\mathbb N} \rightarrow {\mathbb N}$ 
such that for all $(x,k) \in L_1$ if $(x', k') = P((x, k))$, then $k' \leq g(k)$.
\end{enumerate}
The class $\Weft[1]$ contains all problems that are fpt-reducible to \textsc{Independent Set} parameterized
by the size of the solution, i.e. the number of vertices in the maximum independent set.
Showing $\Weft[1]$-hardness (by an fpt-reduction) for a problem rules out the existence of a fixed-parameter
algorithm under the standard assumption $\FPT \neq \Weft[1]$.
The class \paraNP contains all parameterized problems that can be solved by a nondeterministic algorithm in time 
$f(k)\cdot ||x||^{O(1)}$ for some computable function $f$. 
It is known that 
$\FPT=\paraNP$
if and only if 
$\PP=\NP$.
A problem is \paraNP-hard if it is 
\NP-hard for a constant value of the parameter.

\subsection{\paraNP-hardness for $\csp(\D_{2,\infty})$}
\label{sec:pnp-hardness}

We show that if there is no upper bound on the
size of the numbers used in the constraints, then
$\csp(\D_{2,\infty})$ is \NP-hard, even for
instances whose primal graph has constant treewidth. 
In other words, we prove that $\csp(\D_{2,\infty})$ is
\paraNP -hard.
This result is based on the \NP-hard problem \probfont{Subset Sum}~\cite{gj79}.

\pbDef{\probfont{Subset Sum}}
{A set of integers $S$ and an integer $N$.}
{Is there a set $S' \subseteq S$ such that $N=\sum_{s \in S'}s$?
}

\begin{theorem} \label{thm:subsum-d2}
    $\csp(\D_{2,\infty})$ is \NP-hard, even for instances whose
    primal graph has treewidth at most $2$.
\end{theorem}
\begin{proof}
  We present a polynomial-time reduction from the
  \SubsetSum problem to $\csp(\D_{2,\infty})$. Let
  $(S,N)$ be an instance of \SubsetSum with $S=\{s_1,\dotsc,s_n\}$.
  We construct an
  equivalent instance $\inst$ of $\csp(\D_{2,\infty})$ as follows.
  Introduce $n+1$ variables $x_0,\dots,x_n$.
  For every $i$ with $1 \leq i \leq n$, introduce the constraint 
    $x_{i} - x_{i-1} = 0 \lor x_{i} - x_{i-1} = s_i$.
  Finally, add the constraint $x_n - x_0 = N$.
  Note that the primal graph of $\inst$ is a cycle, so its
  treewidth is at most $2$. 
  %The equivalence of the
  %instances is obvious, since choosing an integer $s_i$
  %corresponds to setting $x_i-x_{i-1}$ to $s_i$.
   Given a solution to $\inst$,
   selecting those $s_i$ for which $x_{i} - x_{i-1} = s_{i}$ 
   yields a subset of $S$ that sums up to $N$.
   In the opposite direction, a solution to $\inst$
   can be constructed from the subset $S' \subseteq S$
   that sums up to $N$ by setting $x_{i} - x_{i-1} = s_i$ if $s_i \in S'$
   and $x_{i} - x_{i-1} = 0$ otherwise.
  % Note that the primal graph of $\inst$ is a cycle,
  % so its treewidth is $2$.
\end{proof}

\subsection{$\W{1}$\hy hardness for \probfont{Multi-Dimensional Partitioned Subset Sum}}
\label{sec:mdpss}

Our parameterized hardness result for $\csp(\D_{2,k})$ is based on a
variant of \SubsetSum; we note here that similar but slightly
different variants of \SubsetSum have been considered before~\cite{DBLP:journals/algorithmica/GanianKO21,DBLP:journals/algorithmica/GanianOR23}.
Let $k$ denote a natural number and let $\mvec{v}$ denote a vector of
dimension $K=\binom{k}{2}$. We sometimes refer to the coordinates of
$\mvec{v}$ by a pair $(a,b)$ of natural numbers with $1 \leq a < b
\leq k$; here, we implicitly use an arbitrary bijection between the
$K$ pairs $(a,b)$ satisfying the inequality
and the $K$ coordinates of the vector $\mvec{v}$.
We say that $\mvec{v}$ is {\em uniform} if every non-zero
coordinate of $\mvec{v}$ has the same value $s(\mvec{v})$.
Finally, for an integer $N$, we let $\mvec{N}$ denote the
$K$-dimensional vector that is equal to $N$ at every coordinate.

\pbDefP{\probfont{Multi-dimensional Partitioned Subset Sum}
  (\MPSS{})}{Integers $k$ and $N$, and sets $V_1,\dotsc,V_k$ and
  $E_{1},\dotsc,E_{K}$
  of uniform $K$-dimensional vectors over the
  natural numbers such that:
  \begin{itemize}
  \item 
    % Every vector $\mvec{v} \in V_i$ is non-zero only at the 
    % coordinates $(a,b)$ such that $a=i$ or $b=i$.
    Every vector $\mvec{v} \in V_i$ is non-zero at all coordinates
    $(a,b)$ such that $a=i$ or $b=i$ and zero elsewhere.
%    Whenever convenient,
%    we refer to the $K$-coordinates of the vectors in $S=\bigcup_{1\leq i \leq k}V_i\cup
%    \bigcup_{1\leq r \leq K}E_{r}$ by a pair $(a,b)$ of
%    natural numbers with $1 \leq a < b \leq k$.
  \item Every vector $\mvec{v} \in E_{r}$ is non-zero only at the
    coordinate $r$.
  % \item Every vector $\mvec{v}$ appearing in $\textup{VE}=(\bigcup_{i=1}^kV_i) \cup (\bigcup_{r=1}^KE_i)$ has
  %   the same value $s(\mvec{v})$ at all non-zero coordinates.
  \end{itemize}
}{$k$}{
%Can we choose exactly one vector from every of the sets $V_i$ and
%  $E_{r}$, whose sum adds up to $N$ in every coordinate? 
  Are
  there $\mvec{v}^1,\dotsc, \mvec{v}^k$ and
  $\mvec{e}^1,\dotsc,\mvec{e}^{K}$ with $\mvec{v}^i \in V_i$
  and $\mvec{e}^r \in E_r$ such that
  $(\sum_{i=1}^k\mvec{v}^i)+(\sum_{r=1}^{K}\mvec{e}^r)=\mvec{N}$?
}

\begin{theorem} \label{pro:mpss-hard}
  \MPSS{} is strongly \W{1}\hy hard (i.e. it is \W{1}\hy hard
  even if
  all numbers are encoded in unary).
\end{theorem}
\begin{proof}
  We prove the lemma by a  parameterized reduction from
  \probfont{Multicoloured Clique}, which is well known to be 
  \W{1}\hy complete~\cite{Pietrzak03}.
  Given an integer $k$ and a
  $k$-partite graph $G$ with partition $U_1,\dotsc,U_k$, 
  the \probfont{Multicoloured Clique} problem
  asks whether $G$ contains a $k$-clique
  (note that since the sets $U_i$ are independent, any $k$-clique must contain exactly one vertex from each set~$U_i$).
  We let 
  $W_{(i,j)}$ denote the set of all edges in $G$ with one endpoint in $U_i$ and
  the other endpoint in $U_j$, for every $i$ and $j$ with $1\leq i < j
  \leq k$.
  To show the lemma, we construct an instance
  $$\inst=(k,N,(V_i)_{1\leq i \leq k},(E_{r})_{1 \leq r \leq K}),$$
  of \MPSS{} in polynomial time where all integers in
  $\inst$ are bounded by a polynomial in $|V(G)|$
  and $K=\binom{k}{2}$. Our construction
  yields an instance $\inst$ such that $G$ contains a
  $k$-clique if and only if $\inst$ has a solution.
  
  We will employ Sidon sets from Section~\ref{sec:sidon} in the reduction. 
  Namely, we need a Sidon set containing $|V(G)|$ natural numbers, i.e.
 one number for each vertex of $G$. Since the numbers in
 the Sidon set will be used as numbers in $\inst$, we need
 to ensure that the largest of these numbers is bounded by a
 polynomial in $|V(G)|$.
  We know from Section~\ref{sec:sidon} that such a set (where the bound on the largest element is $8|V|^2$) 
  can be computed in polynomial time.
  In the following, we will
  assume that we are given such a Sidon sequence $\SSS$
  and we let~$\SSS(i)$
  denote the~$i$\hy th element of $\SSS$ for $1 \leq i \leq |V(G)|$. 
  Let $\max(\SSS)$ and $\max_2(\SSS)$ denote
  the largest element of~$\SSS$ and the maximum sum of any
  two distinct elements in~$\SSS$, respectively. We will furthermore assume that the vertices
  of $G$ are identified with the numbers from $1$ to $|V(G)|$ and
  therefore $\SSS(v)$ is properly defined for every $v \in V(G)$.

  We are now ready to construct the instance $\inst$.
  We set $N=\max_2(\SSS)+1$ and
  proceed to the construction of the sets $V_1,\dotsc,V_k$ and
  the sets $E_1,\dotsc,E_K$. For every $i$ with $1 \leq i \leq k$ and every $u \in U_i$, the set $V_i$ contains the vector
  $\mvec{u}$ with $s(\mvec{u})=\SSS(u)$ being non-zero at all
  coordinates $(a,b)$ such that either $a=i$ or $b=i$. Moreover, for
  every $1 \leq r \leq K$ and every $e=(u,v) \in W_{r}$, the set
  $E_r$ contains the vector $\mvec{e}$ with $s(\mvec{e})=(\max_2(\SSS)+1)-(\SSS(u)+\SSS(v))$,
  and the non-zero value appearing only at coordinate $r$.
  
  This completes the construction of $\inst$. It is clear that $\inst$
  can be constructed in polynomial time and that every integer in
  $\inst$ is at most $\max_2(\SSS)+1$ so $\inst$ is polynomially bounded in
  $|V(G)|$. Intuitively, the construction relies on the fact that
  since the sum of each pair of vertices is unique, we can uniquely
  associate each pair with an edge between these vertices, whose value
  will then be the global upper bound of
  $\max_2(\SSS)+1$ minus the unique sum.

  It remains to show that $G$ contains a $k$-clique if and only if $\inst$ has
  a solution. 

\smallskip

  \noindent
  {\bf Forward direction.} Let $C$ be a
  $k$-clique in $G$ with vertices $u_1,\dotsc,u_k$ such that $u_i \in
  U_i$ for every $i$ with $1\leq i \leq k$. Choose the
  vector $\mvec{u}_i$ from $V_i$, $1 \leq i \leq k$,
  and the vector $\mvec{e}_{(i,j)}$ from $E_{(i,j)}$, where $\mvec{e}_{(i,j)}$ is
  the edge with endpoints $u_i$ and $u_j$ for every $i$ and $j$ with
  $1 \leq i < j \leq k$.
  We claim that this choice
  is a solution for $\inst$.
  Let $\mvec{t}$ be the vector
  $(\sum_{i=1}^k\mvec{u}_i)+(\sum_{i=1}^K\mvec{e}_i)$.
  For every coordinate $(i,j)$ with $1 \leq i < j \leq k$, the vectors
  $\mvec{u}_i$, $\mvec{u}_j$, and
  $\mvec{e}_{(i,j)}$ are the only vectors that are non-zero at
  the coordinate $(i,j)$. Therefore,
  $\mvec{t}[(i,j)]=s(\mvec{u}_i)+s(\mvec{u}_j)+s(\mvec{e}_{(i,j)})$. Moreover, using the identities
  $s(\mvec{u}_i)=\SSS(u_i)$,
  $s(\mvec{u}_j)=\SSS(u_j)$, and
  $s(\mvec{e}_{(i,j)})=(\max_2(\SSS)+1)-(\SSS(u_i)+\SSS(u_j))$, we obtain that
  $$\mvec{t}[(i,j)] = \SSS(u_i)+\SSS(u_j)+({\rm max}_2(\SSS)+1)-(\SSS(u_i)+\SSS(u_j))={\rm max}_2(\SSS)+1=N,$$
  as required.

  \smallskip
  
  \noindent
  {\bf Backward direction.}
  Assume $\mvec{u}_i \in V_i$, $1 \leq i \leq k$ together with
  $\mvec{e}_j \in E_j$, $1 \leq j \leq K$ is a solution for $\inst$.
  We claim that $\{u_1,\dotsc,u_k\}$ forms a $k$-clique in
  $G$, i.e. $e_{(i,j)}=\{u_i,u_j\}$ is an edge of $G$ 
  for every $i$ and $j$ with $1 \leq i < j \leq k$.
  Note that the only
  vectors in the solution for $\inst$ that have a
  non-zero contribution towards the $(i,j)$-th coordinate of the sum vector
  are the vectors
  $\mvec{u}_i$, $\mvec{u}_j$, and $\mvec{e}_{i,j}$. Since
  $s(\mvec{u}_i)=\SSS(u_i)$, $\mvec{u}_j=\SSS(u_j)$,
  and $N=\max_2(\SSS)+1$, we see that
  $s(\mvec{e}_{(i,j)})=(\max_2(\SSS)+1)-(\SSS(u_i)+\SSS(u_j))$. Moreover,
  $\SSS$ is a Sidon sequence so the sum $(\SSS(u_i)+\SSS(u_j))$ is
  unique. It follows that $e_{(i,j)}=\{u_i,u_j\}$ as required.
  \end{proof}

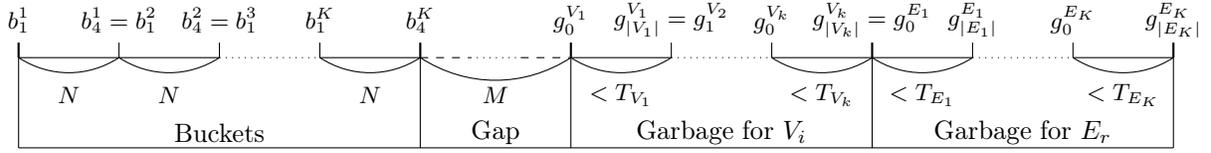
\begin{figure*}
    \centering
%KD: Scale added for camera-ready version to fix overfull hbox
    \begin{tikzpicture}[xscale=0.89, yscale=1]
        \def \w {1.5};
        \def \h {0.2};
        \def \e {0.5};
        \def \g {3*\w/2};
        \def \l {1.2};
        \def \mathlabelsize{\small};
        
        %% Buckets
        % line
        \draw (0,0) -- (2*\w,0);
        \draw (2*\w,0) -- (3*\w,0) [dotted];
        \draw (3*\w,0) -- (4*\w,0);
        % ticks
        \draw (0,0) -- (0,\h) [thick];
        \draw (\w,0) -- (\w,\h);
        \draw (2*\w,0) -- (2*\w,\h);
        \draw (3*\w,0) -- (3*\w,\h);
        \draw (4*\w,0) -- (4*\w,\h) [thick];
        % labels on ticks
        \node at (0, \e) {\mathlabelsize $b^{1}_{1}$};
        \node at (\w,\e) {\mathlabelsize $b^{1}_{4} = b^{2}_{1}$};
        \node at (2*\w,\e) {\mathlabelsize $b^{2}_{4} = b^{3}_{1}$};
        \node at (3*\w,\e) {\mathlabelsize $b^{K}_{1}$};
        \node at (4*\w,\e) {\mathlabelsize $b^{K}_{4}$};
        % arcs below with label N
        \draw (0,0) arc (240:300:\w);
        \node at (\w/2, -\e) {\small $N$};
        \draw (\w,0) arc (240:300:\w);
        \node at (3*\w/2, -\e) {\small $N$};
        \draw (3*\w,0) arc (240:300:\w);
        \node at (7*\w/2, -\e) {\small $N$};
        
        %% Gap of width M
        \draw[dashed] (4*\w,0) -- (4*\w + \g/3,0);
        \draw[dotted] (4*\w + \g/3,0) -- (4*\w + 2*\g/3,0);
        \draw[dashed] (4*\w + 2*\g/3,0) -- (4*\w + \g,0);
        \draw (4*\w,0) arc (240:300:\g);
        \node at (4*\w+\g / 2, -\e) {\small $M$};
        
        %% Garbage for Vi and Er
        % line
        \draw (4*\w+\g,0) -- (5*\w+\g,0);
        \draw (5*\w+\g,0) -- (6*\w+\g,0) [dotted];
        \draw (6*\w+\g,0) -- (8*\w+\g,0);
        \draw (8*\w+\g,0) -- (9*\w+\g,0) [dotted];
        \draw (9*\w+\g,0) -- (10*\w+\g,0);
        % ticks 
        \draw (4*\w+\g,0) -- (4*\w+\g,\h) [thick];
        \draw (5*\w+\g,0) -- (5*\w+\g,\h);
        \draw (6*\w+\g,0) -- (6*\w+\g,\h);
        \draw (7*\w+\g,0) -- (7*\w+\g,\h) [thick];
        \draw (8*\w+\g,0) -- (8*\w+\g,\h);
        \draw (9*\w+\g,0) -- (9*\w+\g,\h);
        \draw (10*\w+\g,0) -- (10*\w+\g,\h) [thick];
        % labels on ticks
        \node at (4*\w+\g, \e) {\mathlabelsize $g^{V_1}_{0}$}; 
        \node at (5*\w+\g, \e) {\mathlabelsize $g^{V_1}_{\abs{V_1}} = g^{V_2}_{1}$};
        \node at (6*\w+\g, \e) {\mathlabelsize $g^{V_k}_{0}$};
        \node at (7*\w+\g, \e) {\mathlabelsize $g^{V_k}_{\abs{V_k}} = g^{E_1}_{0}$};
        \node at (8*\w+\g, \e) {\mathlabelsize $g^{E_1}_{\abs{E_1}}$};  
        \node at (9*\w+\g, \e) {\mathlabelsize $g^{E_{K}}_{0}$};
        \node at (10*\w+\g, \e) {\mathlabelsize $g^{E_{K}}_{\abs{ E_{K}}}$};
        % arcs below with label <T
        \draw (4*\w+\g,0) arc (240:300:\w);
        \node at (9*\w/2+\g,-\e) {\small $< T_{V_1}$}; 
        \draw (6*\w+\g,0) arc (240:300:\w);
        \node at (13*\w/2+\g,-\e) {\small $<T_{V_k}$}; 
        \draw (7*\w+\g,0) arc (240:300:\w);
        \node at (15*\w/2+\g,-\e) {\small $< T_{E_1}$}; 
        \draw (9*\w+\g,0) arc (240:300:\w);
        \node at (19*\w/2+\g,-\e) {\small $< T_{E_K}$};       
        
        %% Delimiters and annotations for blocks
        % Projections
        \draw (0,0) -- (0,-\l);
        \draw (4*\w,0) -- (4*\w,-\l);
        \draw (4*\w+\g,0) -- (4*\w+\g,-\l);        
        \draw (7*\w+\g,0) -- (7*\w+\g,-\l);
        \draw (10*\w+\g,0) -- (10*\w+\g,-\l);
        % Arrows
        \draw (0,-\l) -- (10*\w+\g,-\l);
        % Labels
        \node at (2*\w, -\l+0.2) {Buckets};
        \node at (4*\w + \g/2, -\l+0.2) {Gap};
        \node at (5.5*\w+\g, -\l+0.2) {Garbage for $V_i$};
        \node at (8.5*\w+\g, -\l+0.2) {Garbage for $E_r$};
    \end{tikzpicture}
    \caption{The board consisting of the bucket part and the garbage
      part defined in the proof of Theorem~\ref{thm:w1-hard}. Here,
      $T_{A} = \sum_{\mvec{a} \in A}s(\mvec{a})$ for $A \in \{V_1,\dotsc,V_k,E_1,\dotsc,E_K\}$.}
    \label{fig:board}
\end{figure*}

\begin{figure}[b]
    \centering
%KD: Scale added for camera-ready version to fix overfull hbox
    \begin{tikzpicture}[scale=0.975]
        \def \w {2};
        \def \h {-2};
        \def \b {-3};
        \def \g {3};
        \def \t {2pt};
        % coordinates
        \coordinate (x) at (0,0);
        \coordinate (y) at (\w,0);
        \coordinate (b1) at (\b,\h);
        \coordinate (b2) at (\b+\w,\h);
        \coordinate (b3) at (\b+\w+1,\h);
        \coordinate (b4) at (\b+\w+2,\h);
        \coordinate (g1) at (\g,\h);
        \coordinate (g2) at (\g+\w,\h);
        \coordinate (t) at (0,\t);
        % lines
        \draw (x) -- (y)[thick];
        \draw (b1) -- (b2)[thick];
        \draw (b2) -- (b3) -- (b4);
        \draw (b4) -- (g1)[loosely dotted];
        \draw (g1) -- (g2)[thick];
        \draw (x) -- (b1)[dashed];
        \draw (x) -- (g1)[dashed];
        \draw (y) -- (b2)[dashed];
        \draw (y) -- (g2)[dashed];
        % ticks
        \draw ($(b1) - (t)$) -- ($(b1) + (t)$);
        \draw ($(b2) - (t)$) -- ($(b2) + (t)$);
        \draw ($(b3) - (t)$) -- ($(b3) + (t)$);
        \draw ($(b4) - (t)$) -- ($(b4) + (t)$);
        \draw ($(g1) - (t)$) -- ($(g1) + (t)$);
        \draw ($(g2) - (t)$) -- ($(g2) + (t)$);
        % labels
        \filldraw (x) circle (2pt) node[above] at (x) {$x^{V_i}_{\ell,c}$};
        \filldraw (y) circle (2pt) node[above] at (y) {$y^{V_i}_{\ell,c}$};
        \node[below] at (b1) {$b^c_1$};
        \node[below] at (b2) {$b^c_2$};
        \node[below] at (b3) {$b^c_3$};
        \node[below] at (b4) {$b^c_4$};
        \node[below] at (g1) {$g^{V_i}_{\ell-1}$};
        \node[below] at (g2) {$g^{V_i}_{\ell}$};
        % arcs
        \draw[densely dotted] (x) arc (120:60:\w); 
        \node[above] at ($(x)!0.5!(y) + (0,\w/10)$) {\small $s(\mvec{v}_\ell)$};
        \draw[densely dotted] (b1) arc (240:300:\w+2);
        \node[below] at ($(b1)!0.5!(b4) + (0,-\w/4)$) {\small $N$};
    \end{tikzpicture}
    \caption{
        Possible ways to place the variables $x^{V_i}_{\ell,c}$ and $y^{V_i}_{\ell,c}$
        corresponding to the non-zero coordinate $c = (i,j)$ of the
        $\ell$-th vector $\mvec{v}_\ell$ in $V_i$.
    }
    \label{fig:placing-variable}
\end{figure}
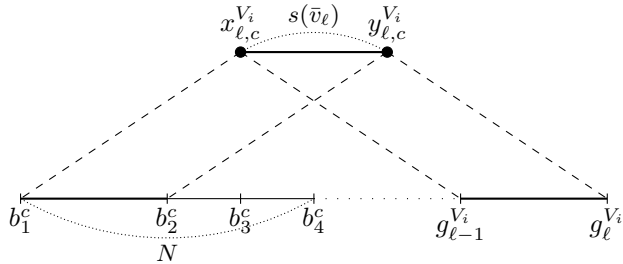

\subsection{\W{1}-hardness for $\csp(\D_{2,1})$}
\label{sec:w1-hardness}

We provide a parameterized reduction from \MPSS{}, which together
with Theorem~\ref{pro:mpss-hard} establishes the result. To
simplify the reduction, we provide it in two stages. First we
show how to construct an equivalent instance $\inst''$ of
$\csp(\D_{2})$ and then we show how to obtain the desired instance
$\inst'$ of $\csp(\D_{2,1})$ from $\inst''$. 

Before giving the formal proof in Theorem~\ref{thm:w1-hard}, let us provide an
informal overview of the main ideas behind the proof.
Let
$\inst=(k,N,(V_i)_{1\leq i \leq k},(E_{r})_{1 \leq r \leq K})$
be an instance of \MPSS{}.
We let our playing board be the real line; please refer to Figure~\ref{fig:board} for an illustration.
First, for every vector $\mvec{v}$ in
$\textup{VE}=(\bigcup_{i=1}^kV_i)\cup (\bigcup_{i=1}^KE_i)$ and
every non-zero coordinate $c$ of $\mvec{v}$, we introduce a segment on the board
represented by two variables $x$ and $y$ at distance exactly
$s(\mvec{v})$ from each other. We divide our board into
two main
parts: the {\em bucket} part and the {\em garbage} part.  While the bucket part provides
placeholders for the segments of the vectors chosen to
be in a solution for $\inst$, the garbage part provides placeholders for all
other segments. Crucial for the idea is a gadget that ensures that
a segment can only be in one of two places, i.e. either its place
inside the bucket part or its place inside the garbage part. 
To illustrate the idea behind this gadget, 
suppose one wants to ensure that a variable $x$ is
either equal to a variable $a$ or equal to a variable $b$. This can
be achieved by the ternary constraint $x=a\lor x=b$.
However, since we are
only allowed to use binary constraints, it becomes more complicated.
The idea is that we additionally ensure that the distance between
$a$ and $b$ is between $M$ and $2M-1$ for some number $M$. Then we can
ensure that $x$ is either equal to $a$ or equal to $b$ by using the
constraints $x=a \lor x-a\geq M$ and $x=b \lor b-x\geq M$.

With this in mind, let us provide some details on the bucket part
and the garbage part. The main idea behind the bucket part is that
it provides placeholders for the segments representing the non-zero
coordinates of all vectors that are in the solution for $\inst$.
More specifically, consider a solution for $\inst$ choosing exactly one
vector $\mvec{v}^i$ from each $V_i$ and exactly one vector
$\mvec{e}^r$ from each $E_r$. Then for every
coordinate $r=(i,j)$, the solution contains exactly three vectors
that are non-zero at coordinate $r$, i.e. the vector $\mvec{v}^i$,
the vector $\mvec{v}^j$, and the vector $\mvec{e}^r$. Thus, the
bucket part will provide three placeholders. This is achieved by
introducing four variables $b^r_1,\dotsc,b^r_4$ for every coordinate
$r$ with the idea that, the
place between $b^r_1$ and $b^r_2$ is a
placeholder for the $r$-th coordinate of $\mvec{v}^i$, the place between
$b^r_2$ and $b^r_3$ is a placeholder the $r$-th coordinate of $\mvec{v}_j$, and the place
between $b^r_3$ and $b^r_4$ is a placeholder for the $r$-th
coordinate of $\mvec{e}_r$. Finally, to verify that the sum of all
vectors in the solution is equal to $N$ at each coordinate $r$, we
introduce the constraint $b_4^r-b_1^r=N$.

The main function of the garbage part is to ensure two
things: (1) if a segment representing a non-zero coordinate of some
vector $\mvec{v}$ in $\textup{VE}$ is chosen to be in the bucket part, then all
segments representing non-zero coordinates of $\mvec{v}$ are chosen to be in the bucket
part and (2) the segments of at least one vector from every set $V_i$ and every set
$E_r$ are chosen to be in the bucket part. To achieve this, the
garbage part consists of $k+K$ parts, i.e. one part for every set
$V_i$ and one part for every set $E_r$. Moreover, the part for a set
$A \in \{V_1,\dotsc,V_k,E_1,\dotsc,E_K\}$, has one placeholder for
every vector $\mvec{a} \in A$, which can hold all segments
representing non-zero coordinates of the vector $\mvec{a}$. This is
achieved by introducing $|A|+1$ variables $g^A_0,\dotsc,g^A_{|A|}$
such that the place between $g^A_{i-1}$ and $g^A_{i}$ is reserved to
hold all segments of the $i$-th vector in $A$. Here, it
is important to recall that every non-zero coordinate of every
vector $\mvec{v}$ in $\textup{VE}$ has the same value $s(\mvec{v})$.
Finally, we ensure (2) by 
% ensuring that not all vectors of $A$ can fit
% into the garbage part by adding the constraint
% $g^A_{|A|}-g_0^A<\sum_{\mvec{a}\in A}s(\mvec{a})=T_A$.
adding the constraint 
$g^A_{|A|}-g_0^A<\sum_{\mvec{a}\in A}s(\mvec{a})=T_A$,
which ensures that not all vectors of $A$ can fit
into the garbage part. We are now ready to provide the formal proof.

\begin{theorem} \label{thm:w1-hard}
  $\csp(\D_{2,1})$ is strongly \W{1}\hy hard parameterized by primal
  treewidth.
\end{theorem}
\begin{proof}
  Let $\inst=(k,N,(V_i)_{1\leq i \leq k},(E_{r})_{1 \leq r \leq K})$
  be an arbitrary instance of \MPSS{}.
  We start by
  introducing the board consisting of the bucket part and
  the garbage part; see Figure~\ref{fig:board} for an illustration.

  \paragraph{The bucket part.}
  For every coordinate $r$ (of the vectors in $\inst$), where $1
  \leq r \leq K$, we introduce $4$ \emph{bucket variables}
  $b^{r}_1,\dotsc, b^{r}_4$ and ensure that those appear
  consecutively in any solution by using the constraints
  $b^{r}_{l+1}-b^{r}_{l}\geq 0$ for every $l$ with $1 \leq l
  <4$. We introduce a constraint that ensures that the distance
  between $b^r_1$ and $b^r_4$ is exactly $N$, i.e. the constraint
  $b^r_4-b^r_1=N$.
  Finally, we arrange the bucket variables for different
  coordinates in the natural order by introducing the constraints
  $b^{r}_4=b_1^{r+1}$ for every $r$ with $1 \leq r < K$.

  % The idea behind the buckets is that they provide placeholders for
  % the non-zero coordinates of the vectors that are in a solution for
  % $\inst$. To see this, consider a solution for $\inst$ choosing exactly one
  % vector $\mvec{v}_i$ from each $V_i$ and exactly one vector
  % $\mvec{e}_j$ from each $E_j$. Then for every
  % coordinate $r=(i,j)$, the solution contains exactly three vectors
  % that are non-zero at coordinate $r$, i.e. the vector $\mvec{v}_i$,
  % the vector $\mvec{v}_j$, and the vector $\mvec{e}_r$. Then, the
  % place between $b^r_1$ and $b^r_2$ is a
  % placeholder for $\mvec{v}_i$, the place between
  % $b^r_2$ and $b^r_3$ is a placeholder for $\mvec{v}_j$, and the place
  % between $b^r_3$ and $b^r_4$ is a placeholder for $\mvec{e}_r$.

  % While the buckets provide the placeholders for the vectors chosen to
  % be in a solution for $\inst$, the garbage variables (introduced in
  % the following) provide placeholders for every other vector. That is
  % for every set $V_i$ and $E_i$ and every vector inside these sets the
  % garbage variables will provide a placeholder for such
  % vector. Moreover, the garbage gadget will also be crucial to ensure
  % that at least one vector of every set is chosen to be in a solution
  % by restricting the total space between all garbage variables of a
  % particular set to be less than the sum of all vectors in the set.

  \paragraph{The garbage part (for the sets $V_i$).}
  For every set $V_i=\{\mvec{v}_1,\dotsc,\mvec{v}_{|V_i|}\}$, we
  introduce $|V_i|+1$ \emph{garbage variables} $g^{V_i}_0,\dotsc
  g^{V_i}_{|V_i|}$ and ensure that those appear
  consecutively in any solution by using the constraints
  $g^{V_i}_{l+1}-g^{V_i}_{l}\geq 0$ for every $l$ with $0 \leq l
  <|V_i|$. We introduce a constraint that ensures that the distance
  between $g^{V_i}_0$ and $g^{V_i}_{|V_i|}$ is smaller than $T=\sum_{\mvec{v}
    \in V_i}s(\mvec{v})$, i.e. the constraint $g^{V_i}_{|V_i|}-g^{V_i}_0<T$.
  Additionally, we arrange the garbage variables for different
  sets $V_i$ in the natural order by introducing the constraints
  $g^{V_{i+1}}_0=g_{|V_i|}^{V_i}$ for every $i$ with $1 \leq i < k$.

  \paragraph{The garbage part (for the sets $E_{r}$).}
  For every set $E_{i}=\{\mvec{v}_1,\dotsc,\mvec{v}_{|E_{i}|}\}$, we
  introduce $|E_{i}|+1$ \emph{garbage variables} $g^{E_{i}}_0,\dotsc
  g^{E_{i}}_{|E_{i}|}$ and ensure that those appear
  consecutively in any solution by using the constraints
  $g^{E_{i}}_{l+1}-g^{E_{i}}_{l}\geq 0$ for every $l$ with $0 \leq l
  <|E_i|$. We introduce a constraint that ensures that the distance
  between $g^{E_i}_0$ and $g^{E_i}_{|E_i|}$ is smaller than $T=\sum_{\mvec{v}
    \in E_i}s(\mvec{v})$, i.e. the constraint $g^{E_i}_{|E_i|}-g^{E_i}_0<T$.
  Additionally, we arrange the garbage variables for different
  sets $E_i$ in the natural order by introducing the constraints
  $g^{E_{i+1}}_0=g_{|E_i|}^{E_i}$ for every $i$ with $1 \leq i <
  K$.

  We ensure that the garbage variables of the sets $E_i$ are placed
  after the garbage variables of the sets $V_i$ by adding the
  constraint $g^{V_k}_{|V_k|}=g^{E_1}_0$.  
  Finally, to make the later arguments simpler, we make sure that
  the last bucket variable $b^{K}_4$ has sufficient
  distance to the first garbage variable $g^{|V_i|}_0$. We let $M = K \cdot N+\sum_{\mvec{v} \in \textup{VE}}s(\mvec{v})$ and
  add the constraint $g^{|V_1|}_0-b^{K}_4=M$.

  \paragraph{The vector variables for the sets $V_i$.}
  For every set $V_i=\{\mvec{v}_1,\dotsc,\mvec{v}_{|V_i|}\}$, every
  $\ell$ with $1 \leq \ell \leq |V_i|$, and every non-zero coordinate
  $c$ of $\mvec{v}_\ell$, we introduce two variables $x^{V_i}_{\ell,c}$
  and $y^{V_i}_{\ell,c}$ and the constraint
  $y^{V_i}_{\ell,c}-x^{V_i}_{\ell,c}=s(\mvec{v}_\ell)$ ensuring that
  the distance between $y^{V_i}_{\ell,c}$ and $x^{V_i}_{\ell,c}$ is
  exactly $s(\mvec{v}_\ell)$. We associate one bucket
  variable, denoted by $B(x^{V_i}_{\ell,c})$ and $B(y^{V_i}_{\ell,c})$, respectively, and one garbage variable,
  denoted by $G(x^{V_i}_{\ell,c})$ and $G(x^{V_i}_{\ell,c})$, respectively,
  with $x^{V_i}_{\ell,c}$ and $y^{V_i}_{\ell,c}$ as follows.
  If $c=(i,j)$ (for some $j>i$), we set $B(x^{V_i}_{\ell,c})=b^c_1$
  and $B(y^{V_i}_{\ell,c})=b^c_2$. Otherwise, i.e. if $c=(j,i)$ (for some
  $j<i$), we set $B(x^{V_i}_{\ell,c})=b^c_2$ and
  $B(y^{V_i}_{\ell,c})=b^c_3$. Moreover, we set $G(x^{V_i}_{\ell,c})=g^c_{\ell-1}$
  and $G(y^{V_i}_{\ell,c})=g^c_{\ell}$. We add constraints that
  ensure that $x^{V_i}_{\ell,c}$ is either equal to
  $B(x^{V_i}_{\ell,c})$ or $G(x^{V_i}_{\ell,c})$ (see Figure~\ref{fig:placing-variable}). As we will show
  later this can be guaranteed by the constraints:
  \begin{itemize}
  \item $x^{V_i}_{\ell,c}=B(x^{V_i}_{\ell,c}) \lor
    x^{V_i}_{\ell,c}-B(x^{V_i}_{\ell,c})\geq M$ and
  \item $x^{V_i}_{\ell,c}=G(x^{V_i}_{\ell,c}) \lor G(x^{V_i}_{\ell,c})-x^{V_i}_{\ell,c} \geq
    M$.
  \end{itemize}
  Similarly, we add constraints so that
  $y^{V_i}_{\ell,c}$ is either equal to
  $B(y^{V_i}_{\ell,c})$ or $G(y^{V_i}_{\ell,c})$, i.e. we add the
  constraints:
  \begin{itemize}
  \item $y^{V_i}_{\ell,c}=B(y^{V_i}_{\ell,c}) \lor
    y^{V_i}_{\ell,c}-B(y^{V_i}_{\ell,c})\geq M$ and
  \item $y^{V_i}_{\ell,c}=G(y^{V_i}_{\ell,c}) \lor G(y^{V_i}_{\ell,c})-y^{V_i}_{\ell,c} \geq
    M$.
  \end{itemize}
  Let $\textup{XY}_{V}$ denote the set of all vector variables for the sets $V_i$, i.e.
  the set $$\SB x^{V_i}_{\ell,c}, y^{V_i}_{\ell,c} \SM 1\leq i \leq k \land 1 \leq \ell
  \leq |V_i| \land c \textup{ is a non-zero coordinate of }\mvec{v}^i_\ell\SE.$$
  
  \paragraph{The Vector Variables for the sets $E_i$.}
  For every set $E_i=\{\mvec{e}_1,\dotsc,\mvec{e}_{|E_i|}\}$ and every
  $\ell$ with $1 \leq \ell \leq |E_i|$, we introduce two variables $x^{E_i}_{\ell}$
  and $y^{E_i}_{\ell}$ and the constraint
  $y^{E_i}_{\ell}-x^{E_i}_{\ell}=s(\mvec{e}_\ell)$ ensuring that
  the distance between $y^{E_i}_{\ell}$ and $x^{E_i}_{\ell}$ is
  exactly $s(\mvec{e}_\ell)$. Similarly, to the vector variables for
  the sets $V_i$, we associate a bucket variable and a garbage
  variable with $x^{E_i}_{\ell}$ and $y^{E_i}_{\ell}$, defined by
  setting: $B(x^{E_i}_{\ell})=b^i_3$,
  $G(x^{E_i}_{\ell})=g^{E_i}_{\ell-1}$, $B(y^{E_i}_{\ell})=b^i_4$, and
  $G(y^{E_i}_{\ell})=g^{E_i}_{\ell}$.
  We add constraints that
  ensure that $x^{E_i}_{\ell}$ is either equal to
  $B(x^{E_i}_{\ell})$ or $G(x^{E_i}_{\ell})$. As we will show
  later, this is guaranteed by the constraints:
  \begin{itemize}
  \item $x^{E_i}_{\ell}=B(x^{E_i}_{\ell}) \lor
    x^{E_i}_{\ell}-B(x^{E_i}_{\ell})\geq M$ and
  \item $x^{E_i}_{\ell}=G(x^{E_i}_{\ell}) \lor G(x^{E_i}_{\ell})-x^{E_i}_{\ell} \geq
    M$.
  \end{itemize}
  Finally, we add constraints that
  imply that $y^{E_i}_{\ell}$ is either equal to
  $B(y^{E_i}_{\ell})$ or $G(y^{E_i}_{\ell})$:
  \begin{itemize}
  \item $y^{E_i}_{\ell}=B(y^{E_i}_{\ell}) \lor
    y^{E_i}_{\ell}-B(y^{E_i}_{\ell})\geq M$ and
  \item $y^{E_i}_{\ell}=G(y^{E_i}_{\ell}) \lor G(y^{E_i}_{\ell})-y^{E_i}_{\ell} \geq
    M$.
  \end{itemize}
  We will verify that these constraints have the required properties later on.
  Let $\textup{XY}_{E}$ denote the set of all vector variables for the sets $E_i$, i.e.
  the set $$\SB x^{E_i}_{\ell}, y^{E_i}_{\ell} \SM 1\leq i \leq K \land 1 \leq \ell
  \leq |V_i|\SE$$ and let $\textup{XY} = \textup{XY}_V\cup \textup{XY}_E$.

\medskip
    
  This completes the construction of the instance $\inst''$ of
  $\csp(\D_{2})$. We first show that the primal treewidth of $\inst''$ is
  at most $4K+3$ and consequently bounded by
  a function of the parameter $k$ only. Let $B=\{
  b^i_l \colon 1 \leq i \leq K \land 1 \leq l \leq 4\}$ be
  the set of all $4K$ bucket variables and let $G$ be the
  primal graph of $\inst''$ after removing the variables in $B$.
  It is 
  straightforward to verify that $G$ has treewidth at most
  $3$ and we obtain, from Proposition~\ref{pro:tw-del}, that the primal
  graph of $\inst''$ has treewidth at most $|B|+3=4K+3$.
  % \begin{itemize}
  % \item for every $V_i$, the graph $G$ induced on the variables
  %   $g^{V_i}_0, \dotsc, g^{V_i}_{|V_i|}$ and the variables $\{x^{V_i}_{1,
  %   one cycle on $g^{V_i}_0,x^{V_i}_{\ell
  % \end{itemize}
  
  We now show the equivalence of the instances $\inst$ and $\inst''$.

\medskip

\noindent
{\bf Forward direction.}
  Let $\mvec{v}^1_{i_1},\dotsc,
  \mvec{v}^k_{i_k}$ and $\mvec{e}^1_{j_1},\dotsc,\mvec{e}^{K}_{j_{K}}$ with
  $\mvec{v}^\ell_{i_\ell} \in V_\ell$ and $\mvec{e}^\ell_{j_\ell} \in E_\ell$ be a solution for
  $\inst$. Informally, the main idea to obtain a solution for $\inst''$
  is to set the variables $x^{V_\ell}_{i_\ell,c}$ and $y^{V_\ell}_{i_\ell,c}$
  equal to their respective bucket variables, i.e. the variables
  $B(x^{V_\ell}_{i_\ell,c})$ and $B(y^{V_\ell}_{i_\ell,c})$, and similarly for the
  variables $x^{E_\ell}_{j_\ell}$ and $y^{E_\ell}_{j_\ell}$. All other variables
  in $\textup{XY}$ are then set to be equal to
  their respective garbage variables. Since $\mvec{v}^1_{i_1},\dotsc,
  \mvec{v}^k_{i_k}$ and $\mvec{e}^1_{j_1},\dotsc,\mvec{e}^{K}_{j_{K}}$ is a
  solution for $\inst$ (and
  $\sum_{\ell=1}^k\mvec{v}^\ell_{i_\ell}+\sum_{\ell=1}^{K}\mvec{e}^\ell_{j_\ell}=\mvec{N}$),
  this ensures that the distance between $b^r_1$ and $b^r_4$ is
  exactly $N$ for every coordinate/bucket $r$ and the distance between
  $g^{A}_0$ and $g^{A}_{|A|}$ is less than the sum of all vectors in
  the set $A \in \{V_1,\dotsc,V_k,E_1,\dotsc,E_{K}\}$.

  More formally, we set $x^{V_\ell}_{i_\ell,c}$ and $y^{V_\ell}_{i_\ell,c}$ equal
  to $B(x^{V_\ell}_{i_\ell,c})$ and $B(y^{V_\ell}_{i_\ell,c})$, respectively, for every
  $\ell$ with $1 \leq \ell \leq k$ and every non-zero coordinate $c$ of
  $\mvec{v}^\ell_{i_\ell}$. Similarly, we set $x^{E_\ell}_{j_\ell}$ and $y^{E_\ell}_{j_\ell}$ equal
  to $B(x^{E_\ell}_{j_\ell})$ and $B(y^{E_\ell}_{j_\ell})$, respectively, for every
  $\ell$ with $1 \leq \ell \leq K$. For every other
  variable $v$ in $\textup{XY}$, we set $v$ equal to $G(v)$. Finally,
  we set $g^{V_\ell}_{i_\ell-1}$ equal to $g^{V_\ell}_{i_\ell}$ for every $\ell$ with $1 \leq \ell
  \leq k$ and $g^{E_\ell}_{j_\ell-1}$ equal to $g^{E_\ell}_{j_\ell}$ for every $\ell$ with $1 \leq \ell
  \leq k$. Note that because of the distances between the variables
  %of the form $x^{?}_?$ and the variables of the form $y^?_?$
  in $\textup{XY}$, this already fixes the position (value)
  of each variable (up to an additive constant). Note also that all constraints are
  satisfied. In particular, the constraints $b^c_4-b^c_1=N$ are
  satisfied because
  $\sum_{\ell=1}^k\mvec{v}^\ell_{i_\ell}+\sum_{\ell=1}^{K}\mvec{e}^\ell_{j_\ell}=\mvec{N}$.
  Similarly, for every set $A \in \{V_1,\dotsc,
  V_k,E_1,\dotsc,E_{K}\}$ the constraints $g^{A}_{|A|}-g^{A}_0<\sum_{\mvec{v} \in
    A}s(\mvec{v})$ are satisfied since $g^{A}_{|A|}-g^{A}_0=(\sum_{\mvec{v} \in
    A\setminus C}s(\mvec{v}))$ and $A\cap C \neq \emptyset$, where
  $C=\{\mvec{v}^1_{i_1},\dotsc,\mvec{v}^k_{i_k},\mvec{e}^1_{j_1},\dotsc,\mvec{e}^{K}_{j_{K}}\}$.

\medskip

\noindent
{\bf Backward direction.}
Let $\alpha$ be an arbitrary solution
  to $\inst''$. We start by showing the following claim.
  \begin{claim}\label{clm:xy-or}
    For every $v \in \textup{XY}$ either $\alpha(v)=\alpha(B(v))$ or
    $\alpha(v)=\alpha(G(v))$.
  \end{claim}
  \begin{claimproof}
  Suppose to the contrary that this is not the case.
  Because $v \in \textup{XY}$, $v$ appears in the two constraints:
  \begin{itemize}
  \item $v=B(v) \lor v-B(v)\geq M$ and
  \item $v=G(v) \lor G(v)-v\geq M$.  
  \end{itemize}
  Thus, $\alpha(v)-\alpha(B(v))\geq M$ and
  $\alpha(G(v))-\alpha(v)\geq M$. However, this is only possible if
  $\alpha(G(v))-\alpha(B(v))\geq 2M$, which, as we will show now, is not
  the case. It follows from the relation between
  the bucket and garbage variables and the definition of $B(v)$ and
  $G(v)$ that $\alpha(B(v)) \geq b^1_1$ and $\alpha(G(v)) \leq
  g^{E_{K}}_{|E_{K}|}$. Hence,
  $\alpha(G(v))-\alpha(B(v))\leq
  \alpha(g^{E_{K}}_{|E_{K}|})-\alpha(b^1_1)$.
  Because of the constraints on the bucket variables and the garbage
  variables, i.e. the constraints:
  \begin{itemize}
  \item $b^i_4-b^i_1=N$ for every $i$ with $1 \leq i\leq
    K$,
  \item $b^{i+1}_1=b^i_4$ for every $i$ with $1 \leq i<
    K$,
  \item $g^A_{|A|}-g^A_0<\sum_{\mvec{v} \in A}s(\mvec{v})$ for every $A
    \in \{V_1,\dotsc,V_k,E_1,\dotsc,E_{K}\}$,
  \item $g^{V_{i+1}}_{0}=g^{V_i}_{|V_i|}$ for every $i$ with $1 \leq i\leq
    k$,
  \item $g^{E_{i+1}}_{0}=g^{E_i}_{|E_i|}$ for every $i$ with $1 \leq i\leq
    K$,
  \item $g^{V_k}_{|V_k|}=g^{E_1}_0$,
  \item $g^{V_1}_0-b^{K}_4=M$.
  \end{itemize}
  we obtain that:
  \begin{eqnarray*}
    \alpha(G(v))-\alpha(B(v)) & \leq & \alpha(g^{E_{K}}_{|E_{K}|})-\alpha(b^1_1)\\
                               & < & NK+M+\sum_{\mvec{v}\in S}s(\mvec{v}) \\
                               & \leq & NK+\sum_{\mvec{v}\in S}s(\mvec{v})+M \\
                               & \leq & 2M \\
  \end{eqnarray*}
  We see that $\alpha(G(v))-\alpha(B(v))<2M$. This contradiction concludes the proof
  of Claim~\ref{clm:xy-or}.
\end{claimproof}

  We say that a vector $\mvec{v}^i_\ell \in V_i$ is in the bucket if
  $\alpha(x^{V_i}_{\ell,c})=\alpha(B(x^{V_i}_{\ell,c}))$ and
  $\alpha(y^{V_i}_{\ell,c})=\alpha(B(y^{V_i}_{\ell,c}))$ for every non-zero
  coordinate $c$ of $\mvec{v}^i_\ell$. Moreover, we say that
  $\mvec{v}^i_\ell \in V_i$ is in the garbage if $\alpha(x^{V_i}_{\ell,c})=\alpha(G(x^{V_i}_{\ell,c}))$ and
  $\alpha(y^{V_i}_{\ell,c})=\alpha(G(y^{V_i}_{\ell,c}))$ for every non-zero
  coordinate $c$ of $\mvec{v}^i_\ell$. Similarly, we say that a vector
  $\mvec{e}^i_{\ell} \in E_i$ is in the bucket if
  $\alpha(x^{E_i}_{\ell})=\alpha(B(x^{E_i}_{\ell}))$ and
  $\alpha(y^{E_i}_{\ell})=\alpha(B(y^{E_i}_{\ell}))$ and we say that
  $\mvec{e}^i_{\ell} \in E_i$ is in the garbage if $\alpha(x^{E_i}_{\ell})=\alpha(G(x^{E_i}_{\ell}))$ and
  $\alpha(y^{E_i}_{\ell})=\alpha(G(y^{E_i}_{\ell}))$.

  Based on Claim~\ref{clm:xy-or}, we will now show that for every set $V_i$ and
  every set $E_i$ exactly one vector is in the bucket and all other
  vectors (of the set) are in the garbage. We start by showing the claim for the
  sets $V_i$.
  \begin{claim}\label{clm:sets-eo-v}
    For every $V_i$ there is exactly one vector $\mvec{v}^i_\ell \in
    V_i$ such that $\mvec{v}^i_{\ell}$ is in the bucket, and all
    other vectors in $V_i$ are in the garbage.
  \end{claim}
  \begin{claimproof}
  We first show that at least one vector $\mvec{v}^i_\ell$ is in the bucket. Suppose
  to the contrary that this is not the case.
  Claim~\ref{clm:xy-or} implies that
  for every $\ell$ there is a non-zero coordinate $c$ of
  $\mvec{v}^i_\ell$ such that
  $\alpha(x^{V_i}_{\ell,c})=\alpha(G(x^{V_i}_{\ell,c}))$ and, consequently, $\alpha(y^{V_i}_{\ell,c})=\alpha(G(y^{V_i}_{\ell,c}))$. Thus, 
  the distance between $g^{V_i}_{\ell-1}$ and
  $g^{V_i}_{\ell}$ is exactly $s(\mvec{v}^i_\ell)$ so the
  distance between $g^{V_i}_0$ and $g^{V_i}_{|V_i|}$ equals
  $\sum_{\mvec{v}\in V_i}s(\mvec{v})$. This violates the constraint
  $g^{V_i}_{|V_i|}-g^{V_i}_0<\sum_{\mvec{v}\in V_i}s(\mvec{v})$ 
  and
  consequently contradicts our assumption that $\alpha$ is a solution for
  $\inst''$. We conclude that there is at least one vector
  $\mvec{v}^i_\ell$ in the bucket.
  
  It remains to show that all other vectors $\mvec{v}^i_{\ell'}\in V_i$
  for $\ell'\neq \ell$ are in the garbage, i.e.
  $\alpha(x^{V_i}_{\ell',c})=\alpha(G(x^{V_i}_{\ell',c}))$ and
  $\alpha(y^{V_i}_{\ell',c})=\alpha(G(y^{V_i}_{\ell',c}))$ for every $\ell'\neq \ell$ and every non-zero
  component $c$ of $\mvec{v}^i_{\ell'}$. Suppose this is not the case and assume that
  the claim is violated for $\ell'$ and $c$. Then, by
  Claim~\ref{clm:xy-or}, it follows that
  $\alpha(x^{V_i}_{\ell',c})=\alpha(B(x^{V_i}_{\ell',c}))$
  (and therefore also
  $\alpha(y^{V_i}_{\ell',c})=\alpha(B(y^{V_i}_{\ell',c}))$). This
  implies that the distance between $B(x^{V_i}_{\ell',c})$ and
  $B(y^{V_i}_{\ell',c})$ is equal to
  $s(\mvec{v}^i_{\ell'})$. However, since $\alpha(x^{V_i}_{\ell,c})=\alpha(B(x^{V_i}_{\ell,c}))$
  and $\alpha(y^{V_i}_{\ell,c})=\alpha(B(y^{V_i}_{\ell,c}))$, we
  obtain that the distance between
  $B(x^{V_i}_{\ell',c})=B(x^{V_i}_{\ell,c})$ and
  $B(y^{V_i}_{\ell',c})=B(y^{V_i}_{\ell,c})$ is equal to
  $s(\mvec{v}^i_\ell)$. However, this is not possible because no two
  vectors in $V_i$ agree on $s(\mvec{v})$ and hence $s(\mvec{v}^i_{\ell'})\neq s(\mvec{v}^i_{\ell})$.
\end{claimproof}

An analogous proof
shows the statement of Claim~\ref{clm:sets-eo-v} for the sets $E_i$
  (instead of the sets $V_i$).

  \begin{claim}\label{clm:sets-eo-e}
    For every $E_i$ there is exactly one vector $\mvec{e}^i_\ell \in
    E_i$ such that $\mvec{e}^i_{\ell}$ is in the bucket, and all
    other vectors in $E_i$ are in the garbage.
  \end{claim}

  We continue by showing that the vectors that are in the bucket form a
  solution for $\inst$. Let $\mvec{v}^i \in V_i$ and
  $\mvec{e}^j\in E_j$ with $1 \leq i \leq k$ and $1
  \leq j \leq K$ be the vectors that are in  the bucket; these vectors
  exist due to Claims~\ref{clm:sets-eo-v} and \ref{clm:sets-eo-e}.
  The constraints $b^l_4-b^l_1=N$ imply that
  \begin{eqnarray*}
    \left(\sum_{i=1}^k\mvec{v}^i\right)+\left(\sum_{i=1}^{K}\mvec{e}^i\right)=\mvec{N}.
  \end{eqnarray*}
   for every $1 \leq l \leq K$.
  Hence, the vectors $\mvec{v}^1,\dotsc,\mvec{v}^k$ and
  $\mvec{e}^1,\dotsc,\mvec{e}^{K}$ indeed form a solution for $\inst$.

\medskip

  \newcommand{\op}{\odot}
  % continue here

The proof so far shows that $\inst$ has a solution
if and only if $\inst''$ has a solution.
  In the final part of the proof, we 
  show how to transform the instance $\inst''$ into the equivalent
  instance $\inst'$ of $\csp(\D_{2,1})$. To achieve this we first
  replace every constraint of the form $a-b \: \op \: n$ (for variables
  $a$ and $b$, natural number $n$ with $n>1$, and $\op \in
  \{\leq,<,=,\geq,>\}$) by a `path' on $n$ auxiliary variables. More
  formally, to replace the constraint $C=a-b \: \op \: n$, we add $n$
  auxiliary variables $h_1^C,\dotsc,h^C_n$ and the constraints $h_1-b \:
  \op \:  1$, $h^C_{i+1}-h^C_i \: \op \: 1$ for every $i$ with $1\leq i \leq n$,
  and $h_n^C=b$. This allows us to replace the following constraints
  of $\inst''$:
  \begin{itemize}
  \item the constraints $b^i_4-b^i_1=N$ for every $i$ with $1 \leq i
    \leq K$,
  \item the constraints $g^{A}_{|A|}-g^{A}_0 < \sum_{\mvec{v} \in
      A}s(\mvec{v})$ for every $A \in \{V_1,\dotsc,V_k,E_1,\dotsc,
    E_{K}\}$,
  \item the constraint $g^{V_1}_0-b^{K}_4=M$,
  \item the constraints
    $y^{E_i}_{\ell}-x^{E_i}_{\ell}=s(\mvec{e}^i_\ell)$, for every $i$
    and $\ell$ with $1 \leq i \leq K$ and $1 \leq \ell \leq |E_i|$, and
  \item the constraints
    $y^{V_i}_{\ell,c}-x^{V_i}_{\ell,c}=s(\mvec{v}^i_\ell)$, for every $i$
    and $\ell$ with $1 \leq i \leq k$, $1 \leq \ell \leq |E_i|$, and
    every non-zero component $c$ of $\mvec{v}^i_\ell$.
  \end{itemize}
  Note that this reduction is polynomial since the numbers in the
  instance $\inst$ can be assumed to be polynomially bounded in the input size
  because \MPSS{} is strongly \W{1}\hy hard. Also note that this
  replacement does not increase the treewidth of the primal graph by
  more than $1$ since the new primal graph can be obtained by
  subdividing edges of the original primal graph and it is well known
  that subdividing edges can only increase the treewidth of a graph by
  at most $1$. Let $\inst'''$ be the instance obtained from $\inst''$
  after replacing all of the constraints as described above.

  It now only remains to replace the remaining constraints for the
  variables in $\textup{XY}$. Recall that these constraints are of
  the form $z=B(z) \: \lor \: z-B(z)\geq M$ and $z=G(z) \: \lor \: G(z)-z\geq M$
  for some $z \in \textup{XY}$. As we saw in
  Claim~\ref{clm:xy-or}, the effect of these two constraints is that
  for every variable $z \in \textup{XY}$ and every solution $\alpha$
  for $\inst''$, either $\alpha(z)=\alpha(B(z))$ or 
  $\alpha(z)=\alpha(G(z))$. To replace these constraints, we first
  introduce a gadget $U(a,b,Z)$, where $a$ and $b$ are variables and
  $Z$ is a natural number, which ensures that either $a=b$ or $b-a\geq
  Z$. The gadget $U(a,b,Z)$ has $2Z+1$ auxiliary variables $h_0,\dotsc,h_{2Z}$ and the
  following constraints:
  \begin{enumerate}[(C1)]
  \item $h_0=a$, $h_{i+1}-h_i \in [0,1]$ for every $i$ with
    $0\leq i < 2Z$, and $h_{2Z}=b$,
  \item $h_{i+2}-h_i=0 \lor h_{i+2}-h_i>1$ for every $i$ with
    $0 \leq i < 2Z-1$.
  \end{enumerate}
  \begin{claim}\label{clm:UabZ}
    Consider the instance $U(a,b,Z)$ for variables $a$ and $b$ and
    natural number $Z$. Then:
    \begin{itemize}
    \item every solution $\beta$ for $U(a,b,Z)$ satisfies either
      $\beta(a)=\beta(b)$ or $\beta(b)-\beta(a) \in (Z,2Z]$, and
    \item for every number $Z'\in \{0\} \cup (Z,2Z]$, there is a solution $\beta$
      for $U(a,b,Z)$ such that $\beta(b)-\beta(a)=Z'$.
    \end{itemize}
  \end{claim}
  \begin{claimproof}
  Let $\beta$ be a solution for $U(a,b,Z)$. The constraints
  in (C1) imply that $\beta(b)-\beta(a) \in [0,2Z]$. 
  If $\beta(a)=\beta(b)$, then there is nothing to
  show. Hence, assume that $\beta(a)<\beta(b)$. Now,
  there is an $i$ with $0 \leq i <2Z-1$ such that $h_{i+2}-h_{i}>0$.
  We first show that $h_{i+2}-h_{i}>0$ for every $i$
  with $0 \leq i <2Z-1$. Suppose that this is not the case and let $i$
  be an index such that $h_{i+2}-h_{i}>0$ but either
  $h_{i+3}-h_{i+1}=0$ or $h_{i+1}-h_{i-1}=0$. In both cases it follows
  from the constraints in (C2) that $h_{i+2}-h_{i}>1$. Consequently,
  the constraints in (C1) imply that $h_{i+1}-h_i>0$ and
  $h_{i+2}-h_{i+1}>0$. However, this implies that $h_{i+3}-h_{i+1}>0$
  (since $i<2Z-2$) and $h_{i+1}-h_{i-1}>0$ (since $i>0$) so we obtain a contradiction. Hence,
  $h_{i+2}-h_{i}>0$ for every $i$ with $0 \leq i <2Z-1$, which
  together with the constraints in (C2) implies that $h_{i+2}-h_{i}>1$ for
  every $i$ with $0 \leq i <2Z-1$. 
  Since 
  \begin{eqnarray*}
    \beta(b)-\beta(a) & \geq &
                               \sum_{i=1}^Z\beta(h_{2i})-\beta(h_{2i-2})\\
                      & > &  \sum_{i=1}^Z1\\
    & = & Z
  \end{eqnarray*}
  it follows that $\beta(b)-\beta(a)>Z$. This completes the proof of
  the first statement of the claim.

  We continue with the second
  statement of the claim. Arbitrarily choose $Z'\in \{0\} \cup (Z,2Z]$. If $Z'=0$, then we
  set
  $$\beta(a)=\beta(h_0)=\beta(h_1)=\dotsb=\beta(h_{2Z-1})=\beta(h_{2Z})=\beta(b)$$
  and this assignment clearly satisfies all constraints in (C1) and (C2). If $Z' \in (Z,2Z)$, then we set $\beta(a)=\beta(h_0)$,
  $\beta(h_{i+1})-\beta(h_i)=0.5+\epsilon$ for every $i$ with $0 \leq
  i < 2Z$, and $\beta(h_{2Z})=\beta(b)$, where $\epsilon=(Z'-Z)/2Z$. It is straightforward to verify that this assignment satisfies the constraints in (C1) and (C2).
  This completes the proof of Claim~\ref{clm:UabZ}.
\end{claimproof}

  We are now ready to show how to replace the constraints
  $z=B(z) \lor z-B(z)\geq M$ and $z=G(z) \lor G(z)-z\geq M$
  for every variable $z \in \textup{XY}$. That is, for 
  $z \in \textup{XY}$, we replace the constraint
  $z=B(z) \lor z-B(z)\geq M$ with the gadget $U(z,B(z),M)$ and we
  replace the constraint $z=G(z) \lor G(z)-z\geq M$ with the gadget
  $U(G(z),z,M)$.

  Then, $\inst'$ is obtained from
  $\inst'''$ after replacing all the remaining constraints of the
  variables in $\textup{XY}$ as described above. Clearly, $\inst'$ is
  an instance of $\csp(\D_{2,1})$. Furthermore, the treewidth of the
  primal graph of $\inst'$ is at most the treewidth of the primal
  graph of $\inst''$ plus $2$. This is because the treewidth of the
  primal graph of $\inst'''$ is at most the treewidth of $\inst''$
  plus $1$ (as we already argued above). Furthermore, the primal graph
  for $\inst'$ is obtained from the primal graph of $\inst'''$ by
  replacing the edges between $x$ and $B(x)$ as well as between $x$ and $G(x)$
  with the primal graph of the gadget $U(a,b,Z)$ for every $x \in
  \textup{XY}$. The result now follows because the treewidth of the
  primal graph of $U(a,b,Z)$ is at most $2$. Because the treewidth of
  the primal graph of $\inst''$ is at most $4K+3$, we
  obtain that the treewidth of the primal graph of $\inst'$ is at most
  $4K+5$.
\end{proof}
Since all variables in the proof of Theorem~\ref{thm:w1-hard} are only assigned integers,
we can replace every constraint $L<R$ that uses $<$ in the construction, i.e., the constraints of the form $g^{V_i}_{|V_i|}-g^{V_i}_0<T$ and $g^{E_i}_{|E_i|}-g^{E_i}_0<T$, by $L\leq R-1$. Therefore, we obtain the following corollary from Theorem~\ref{thm:w1-hard}.
\begin{corollary}\label{cor:w1-hard}
  $\csp(\D^{\leq}_{2,1})$ is strongly \W{1}\hy hard parameterized by primal
  treewidth.
\end{corollary}

\section{Generalizations}
\label{sec:extensions}

The results that we have proved in Sections~\ref{sec:upper-bounds-time}--\ref{sec:lb} are restricted
in two ways: (1) formulas are assumed to be in
conjunctive normal form and (2) the variable domains are assumed
to be the set of rationals.
We consider the satisfiability problem for DL without these
restrictions in the following two sections.
Thus, we discuss DL with general formulas (i.e. the \probfont{DL-Sat} problem from Section~\ref{sec:difflogic})   in Section~\ref{sec:generalformulas}
and we discuss DL with integer variable domains in
Section~\ref{sec:integerdomains}.

\subsection{General Formulas}
\label{sec:generalformulas}

Every DL formula can be converted into a logically equivalent formula that is in CNF by using well-known laws of logic.
We use this fact for proving the following result.

\begin{theorem} \label{thm:dl-sat-bounds}
\probfont{DL-Sat} is solvable in
$2^{O(n (\log n + \log k))}$ time where $n$ is the number of
variables in the given formula $\phi$ and $k=\num{\phi}$.
\probfont{DL-Sat} is not solvable in
$2^{o(n (\log n + \log k))}$ if the ETH holds.
\end{theorem}
\begin{proof}
The lower bound is an immediate consequence of Theorem~\ref{thm:binary-hardness} since
every instance of $\csp(\disjtemp{2}^{\leqslant})$ can be viewed as a DL formula (as was discussed in Section~\ref{sec:csp}).
To show the upper bound, we let
$\phi$ denote an arbitrary instance of
\probfont{DL-Sat}. Assume $\phi$ contains $n$ variables and
that $k=\num{\inst}$.
Every existential sentence $\phi$ admits a logically
equivalent existential sentence $\phi'$
such that $\phi'$ is in CNF, $\phi$ and $\phi'$ contains the same number of
variables, and $k=\num{\phi'}$.
The formula $\phi'$ may be viewed as an instance $\inst=(V,C)$ of $\csp(\D)$ where $|V|=n$ and $\num{\inst}=k$.
Theorem~\ref{lem:compact-assign}
 implies that $\phi'$ is satisfiable if and only if
it has a solution
 $f : V \rightarrow \cd{n,k}$.
Since $\phi$ and $\phi'$ are logically equivalent formulas, the
same holds for $\phi$. The upper bound
follows immediately 
since $\cd{n,k}$ contains $2^{O(n (\log n + \log k))}$
elements
(as was proved in Corollary~\ref{thm:generalupperbound}).
\end{proof}

The conversion of a DL formula into CNF can obviously lead to an exponential larger formula and the conversion process may
thus take exponential time. Note, however, that we do not need
to compute the CNF formula explicitly in the proof
of Theorem~\ref{thm:dl-sat-bounds}.

Theorem~\ref{thm:dl-sat-bounds} is closely connected to
{\em Satisfiability Modulo Theories} (SMT), i.e.
the decision problem for logical sentences with respect to a
given background theory, where logical formulas are expressed in classical first-order logic with
equality.
%An accessible introduction to SMT can be found in the survey %by Barrett et al.~\cite{Barrett:etal:SMT}.
Let SMT$({\cal T})$ be the problem of determining whether a
first-order sentence 
is true with respect to a background theory ${\cal T}$,
and let SMT$_{\exists}({\cal T})$ be the
subproblem where universal quantifiers are not allowed.
If we let ${\cal T}_{\rm diff}$ denote the background theory for difference
constraints, then
\probfont{DL-Sat} and SMT$_{\exists}({\cal T}_{\rm diff})$ are the same computational problems. 
Jonsson and Lagerkvist~\cite[Theorem 9]{Jonsson:Lagerkvist:algorithmica2023} prove bounds
valid for any background theory:
SMT$_{\exists}(\emptyset)$ is solvable in $2^{O(|V| \log |V|)}$
time but it
cannot be solved in $2^{o(|V| \log |V|)}$ time unless
the ETH is false.
Theorem~\ref{thm:dl-sat-bounds} thus implies that SMT$_{\exists}({\cal T}_{\rm diff})$
is only marginally harder than
SMT$_{\exists}(\emptyset)$.

%SMT$_{\not \forall}(\emptyset)$ is often referred to as {\em %equality logic}
%and this problem is important in, for instance, hardware %verification~\cite{Burch:Dill:cav94}. 
%We see that \probfont{DL-Sat}

Applying the restricted time complexity results and the parameterized
results to \probfont{DL-Sat} directly is, unfortunately, not possible.
The clause arity parameter is obviously not well-defined for an arbitrary existential formula $\phi$ since it is not required
to be in CNF. Similarly, the primal and incidence
graphs are not well-defined in this case. 
Converting the formula into CNF is typically not a viable option
since this process may take exponential time and it may produce
a formula that is exponentially larger than the original formula.
A simple (but sometimes sufficiently powerful) workaround is based on
generalizing the results to more complex subformulas
than clauses. We present one possible way of doing this.
Recall from Section~\ref{sec:csp} that we can always view a CSP instance
as a primitive positive sentence over some structure.
We have used this perspective throughout the article: we view
an existential formula in CNF
as a CSP instance where the structure contains the relations
that describe the allowed clauses.
Clearly, we can instead consider relations
that describe other subformulas than
clauses. This must be done with care, though.
These relations cannot use auxiliary variables
in their definitions since this introduces a
time complexity dependency on the number of
subformulas and not only on the number of variables
and the magnitude of the coefficients.
Furthermore, the definitions of the subformulas
must (in a certain sense) be  easy to compute.
We circumvent this problem by restricting ourselves
to a finite number of subformula types; this restriction can often be lifted but it needs a careful analysis based on the chosen relations and the representation
of them.
We arrive at the following result.

\begin{proposition} \label{prop:qffo-relations}
If ${\bf A}$ is a finite structure that is quantifier-free definable in $\stpfull$, then
the following hold.

\begin{enumerate}
\item
CSP$({\bf A})$ is solvable in $2^{O(n \log n)}$ time.

\item
If the relations in ${\bf A}$ have arity at most 2, then CSP$({\bf A})$ is solvable in $2^{O(n \log \log n)}$ time.

\item
If the relations in ${\bf A}$ have arity at most 3 and $\num{{\bf A}}=0$, then CSP$({\bf A})$ is solvable in $2^{O(n)}$ time.

\item
CSP$({\bf A})$ is in \XP when parameterized by the treewidth of the incidence graph.

\end{enumerate}
\end{proposition}
\begin{proof}
The structure ${\bf A}$ is finite and
every relation in ${\bf A}$ has a quantifier-free definition in $\stpfull$. We may without loss of generality assume that the defining formulas are in CNF. This implies that
every relation can be viewed as a conjunction of relations in $\D_{a,k}$ where $a,k < \infty$.
Since ${\bf A}$ is finite, we may assume that we have access to a table containing
the defining CNF formulas for each relation in ${\bf A}$.

Let $\inst=(V,C)$ denote an instance of CSP$({\bf A})$ and arbitrarily choose a
constraint $R(x_1,\dots,x_n)$ in $C$. The relation $R$ has a definition
$$R(x_1,\dots,x_n) \equiv \bigwedge_{i=1}^p R_i(x_{i,1},\dots,x_{i,{\rm ar}(R_i)})$$
where $R_1,\dots,R_p \in \D_{a,k}$ and
$\{x_{i,j} : 1 \leq i \leq p, 1 \leq j \leq \max\{{\rm ar}(R_m) : 1 \leq m \leq p\}\} \subseteq \{x_1,\dots,x_n\}$.
This constraint in $\inst$ can be replaced by $p$ constraints $$R_1(x_{1,1},\dots,x_{1,{\rm ar}(R_1)}),\dots,R_p(x_{p,1},\dots,x_{p,{\rm ar}(R_p)})$$ 
and this transformation does not
affect the solvability of the instance since it preserves the set of solutions.
Let $\inst'=(V',C')$ denote the instance that results from applying this
transformation to each constraint in $C$. We note that $\inst'$ can be computed
in polynomial time, $V'=V$, $\num{C}=\num{C'}$, and $\inst'$ has a solution if and only if $\inst$ has
a solution.
Now, item 1 follows from Corollary~\ref{thm:generalupperbound}, 
item 2 follows from Theorem~\ref{thm:gammaonetime}, and
item 3 follows from Theorem~\ref{thm:D30}.
Finally, we claim that item
4 follows from Theorem~\ref{the:xp-alg} together with the observation that
$\tw{\ig(\inst')} \leq q \cdot \tw{\ig(\inst)}$, where $q$ is the smallest integer such that any
constraint in $\inst$ is replaced by at most $q$ constraints in $\inst'$; note that $q$ can be considered constant because ${\bf A}$ is finite. The last observation follows because any tree decomposition of $\ig(\inst)$ can be transformed into a tree decomposition of $\ig(\inst')$ by replacing any vertex corresponding to a constraint $c$ in $\inst$ by the at most $q$ vertices corresponding to the constraints that replace $c$ in
$\inst'$.
\end{proof}

\subsection{Integer Domains}
\label{sec:integerdomains}

We show that the complexity results presented in this
article also hold if we restrict DL to integer variable domains. Henceforth,
we let $\D_{a,k}^{\mathbb Z}$ denote the set of relations in $\D_{a,k}$
restricted to the integers.
Let ${\bf A}$ and ${\bf B}$ be two structures.
We write $\csp({\bf A}) \leq_0 \csp({\bf B})$
if there exists a polynomial-time reduction $F$ from
$\csp({\bf A})$ to $\csp({\bf B})$
that introduces no additional variables, i.e.
if $(V,C)$ is an instance of 
$\csp({\bf A})$, then $F((V,C))=(V,C')$.
The existence of such a reduction implies the following.

\begin{enumerate}
\item
If $\csp({\bf A})$ is not solvable within a time bound $f(|V|)$,
then $\csp({\bf B})$ is not solvable within $f(|V|)$, either.

\item
If $\csp({\bf B})$ is solvable within a time bound $f(|V|)$,
then $\csp({\bf A})$ is solvable within $f(|V|)$, too.

\end{enumerate}

%\begin{lemma} \label{lem:integer-solution}
%If an instance $I$ of \csp$(\D^{\leq})$ has a solution, then it has
%a solution over $\integers$, too.
%\end{lemma}
%\begin{proof}
%Let $\inst=(V,C)$ with $C=\{c_1,\dots,c_p\}$ denote an arbitrary %satisfiable instance of
%CSP$(\D^{\leq})$.
%Since $\inst$ is satisfiable, we can pick one literal $l_i$
%out of the definition of every constraint $c_1,\dots,c_p$ such that %$\{l_1,\dots,l_p\}$
%is satisfiable. This set of constraints
%admits an integer solution since the literals
%are of the form $x-y \leq c$ and the bound $c$ is an integer. %(***Add reference.***)
%Consequently, $\inst$ admits an integer solution.
%\end{proof}

We continue by presenting a number of reductions.

\begin{lemma} \label{lem:integer-ak}
$\csp(\D^{\leq}_{a,k}) \leq_0 \csp(\D^{\mathbb Z}_{a,k}) \leq_0 \csp(\D^{\leq}_{a,k+1})$.
\end{lemma}
\begin{proof}
\noindent
We first verify that if an instance $\inst$ of \csp$(\D^{\leq})$ has a solution, then it has
a solution over $\integers$, too.
Let $\inst=(V,C)$ with $C=\{c_1,\dots,c_p\}$ denote an arbitrary satisfiable instance of
CSP$(\D^{\leq})$.
Since $\inst$ is satisfiable, we can pick one literal $l_i$
out of the definition of every constraint $c_1,\dots,c_p$ such that $\{l_1,\dots,l_p\}$
is satisfiable. This set of constraints
admits an integer solution since the literals
are of the form $x-y \leq c$ and the bound $c$ is an integer: this follows from the original algorithm for solving STPs
by Dechter et al.~\cite[Section 3]{Dechter:etal:ai91}
but it is also a consequence of the theory of total unimodularity~\cite[Section 13.2]{Papadimitriou:Steiglitz:CO}.
Consequently, $\inst$ admits an integer solution.

We begin with the reduction $\csp(\D^{\leq}_{a,k}) \leq_0 \csp(\D^{\mathbb Z}_{a,k})$. 
Let $R$ denote
a relation in $\D^{\leq}_{a,k}$
and let $R_{\mathbb Z}$ have the same definition as $R$
but with domain ${\mathbb Z}$ instead of ${\mathbb Q}$. Let $\inst=(V,C)$ denote an arbitrary instance of
$\csp(\D^{\leq}_{a,k})$ and let $\inst_{\mathbb Z}=(V,C_{\mathbb Z})$
where $C_{\mathbb Z}=\{R_{\mathbb Z}(x_1,\dots,x_k) : R(x_1,\dots,x_k) \in C\}$.
If $\inst$ is not satisfiable, then $\inst_{\mathbb Z}$ is not
satisfiable since ${\mathbb Z} \subseteq {\mathbb R}$. If $\inst$ is satisfiable, then
$\inst_{\mathbb Z}$ is satisfiable as pointed out earlier.

We continue with the reduction $\csp(\D^{\mathbb Z}_{a,k}) \leq_0 \csp(\D^{\leq}_{a,k+1})$.
Let $R_{\mathbb Z}$ denote
a relation in $\D^{\mathbb Z}_{a,k}$.
We define a relation $R$ over ${\mathbb Q}$ as follows: $R$ has
the same definition as $R_{\mathbb Z}$
but every literal that is a strict inequality $x-y < c$
is replaced by $x-y \leq c-1$.
Observe that $R_{\mathbb Z} \subseteq R$
and $R$ is a member of $\D^{\leq}_{a,k+1}$.
Let $\inst_{\mathbb Z}=(V,C_{\mathbb Z})$ denote an arbitrary instance
of $\D^{\mathbb Z}_{a,k}$ and let
$\inst=(V,C)$
where $C=\{R(x_1,\dots,x_k) : R_{\mathbb Z}(x_1,\dots,x_k) \in C_{\mathbb Z}\}$.
If $\inst_{\mathbb Z}$ has a solution $f_{\mathbb Z}$, then
this solution is a solution to $\inst$, too, since
$R_{\mathbb Z} \subseteq R$ for every
$R_{\mathbb Z} \in \D^{\mathbb Z}_{a,k}$.
If $\inst$ is satisfiable, then
it has an integer solution (as discussed earlier)
and
this solution witnesses the satisfiability of $\inst_{\mathbb Z}$.
\end{proof}

%\begin{lemma} \label{lem:integer-zero-solution}
%If an instance $I$ of \csp$(\D_{a,0})$ has a solution, then it has
%a solution over $\integers$, too.
%\end{lemma}
%\begin{proof}
%Arbitrarily choose a satisfiable instance $\inst$ of %CSP$(\D_{a,0})$. 
%Assume that
%$f:V \rightarrow {\mathbb Q}$ is a solution
%to $\inst$. Observe that the function $f_c(x)=c \cdot f(x)$ is a %solution
%to $\inst$ whenever $c \neq 0$.
%We assume that $V=\{x_1,\dots,x_n\}$ and
%$f(x_i)=a_i/b_i$, $1 \leq i \leq n$, where
%$a_i$ and $b_i \neq 0$ are integers. Let $c=
%b_1 \cdot \ldots \cdot b_n$ and note
%that $f_c$ is a function from $V$ to ${\mathbb Z}$.
%Thus, a satisfiable instance $\inst$ of CSP$(\D_{a,0})$
%always has an integer solution.
%\end{proof}

\begin{lemma}
$\csp(\D_{a,0}) \leq_0 \csp(\D^{\mathbb Z}_{a,0}) \leq_0 \csp(\D_{a,0})$
\end{lemma}
\begin{proof}
We first verify that
if an instance $I$ of \csp$(\D_{a,0})$ has a solution, then it has
a solution over $\integers$, too.
Arbitrarily choose a satisfiable instance $\inst$ of CSP$(\D_{a,0})$. 
Assume that
$f:V \rightarrow {\mathbb Q}$ is a solution
to $\inst$. Observe that the function $f_c(x)=c \cdot f(x)$ is a solution
to $\inst$ whenever $c \neq 0$.
We assume that $V=\{x_1,\dots,x_n\}$ and
$f(x_i)=a_i/b_i$, $1 \leq i \leq n$, where
$a_i$ and $b_i \neq 0$ are integers. Let $c=
b_1 \cdot \ldots \cdot b_n$ and note
that $f_c$ is a function from $V$ to ${\mathbb Z}$.
Thus, a satisfiable instance $\inst$ of CSP$(\D_{a,0})$
always has an integer solution.

Let us now consider the reduction $\csp(\D_{a,0}) \leq_0 \csp(\D^{\mathbb Z}_{a,0})$.
Given a relation $R \in \D_{a,0}$,
we let $R_{\mathbb Z}$ denote $R$ restricted to
the integers.
Let $\inst_{\mathbb Z}$ denote an arbitrary instance
of $\D^{\mathbb Z}_{a,0}$ and let
$\inst=(V,C)$
where $C=\{R(x_1,\dots,x_k) : R_{\mathbb Z}(x_1,\dots,x_k) \in C_{\mathbb Z}\}$.
If $\inst_{\mathbb Z}$ has a solution, then
$\inst$ has a solution, too.
If $\inst$ has a solution, then
$\inst_{\mathbb Z}$ has a solution as pointed out earlier.
The other reduction is analogous.
\end{proof}

These reductions imply
that all results in Table~\ref{tb:time-summary}
hold for $\D^{\mathbb Z}_{a,k}$.
Similarly, the results in Table~\ref{tb:summary}
hold since the reductions do not change the
primal and incidence graphs of a given instance.
The results for \probfont{DL-Sat} (Theorem~\ref{thm:dl-sat-bounds}) also hold in the integer case. The
reductions show that
$\csp(\D^{\mathbb Z})$ is solvable in
$2^{O(n(\log n + \log k))}$ time but not in
$2^{o(n(\log n + \log k))}$ (under the ETH).
The proof of Theorem~\ref{thm:dl-sat-bounds}
shows that these results immediately carry over
to \probfont{DL-Sat} over the integers.

\section{Conclusion and Future Work}
\label{sec:discussion}

We have initiated a fine-grained complexity analysis of
the satisfiability problem for DL. We have studied
the time complexity of $\csp(\D_{a,k})$ and obtained
closely matching bounds
for almost all choices of $a,k \in \naturals \cup \{\infty\}$.
We have studied the parameterized complexity
of $\csp(\D_{a,k})$ (with parameters
primal and incidence treewidth) and obtained a comprehensive
picture for all choices of $a$ and $k$. We have considered
generalizations where arbitrary formulas are allowed and where 
variable domains are the integers; 
many of our results survive such generalizations.

A future research direction is to close the
gaps between lower bounds and upper bounds 
for time complexity. This boils down to a better understanding
of the time complexity of $\csp(\D_{2,k})$.
There is a lack of natural problems that can be solved in $2^{O(n\log\log n)}$ time but 
do not admit a single-exponential-time algorithm.
This may point in the direction that $\csp(\D_{2,k})$ is solvable
in single-exponential time but it may equally well indicate a need
for new lower bound techniques.
We remark that the running time of the bounded-span algorithm 
(Lemma~\ref{lem:w-disjtemp}) is the dominant term 
in the time complexity of our algorithm for $\csp(\D_{2,k})$ so
improving this part would reduce the overall time complexity.

Our work on parameterized complexity have focused on the parameters
primal and incidence treewidth
One possible way forward is to study other structural parameters.
The notion of treewidth captures the fact that trees are structurally simple, but fails to
do this for cliques since the treewidth of an $n$-clique is $n-1$.
An alternative graph decomposition 
with a corresponding quality measure (known as {\em clique-width})
was introduced and analyzed in a series of articles~\cite{Courcelle:etal:jcss93,Courcelle:Olariu:dam2000,Wanke:dam94}.
This decomposition
captures the structure of both sparse graphs (such as trees) and dense graphs (such as cliques), and
it is known to have algorithmic properties that are
similar to those of bounded treewidth graphs. It may thus be highly relevant in connection
with DL.

Algorithms for deciding the truth of DL
formulas containing universal quantifiers is a natural step forward.
Theorem~\ref{lem:compact-assign} suggests a straightforward but incorrect approach.
Consider a formula $Q_1x_1 \dots Q_nx_n.\phi$ where
$Q_i \in \{\forall,\exists\}$ and $\phi$ is quantifier-free.
Let $D=CD(n,k)$ be the set of values needed for $\phi$ via Lemma 6.
If $Q_1=\forall$, then we assign the values from $D$
to variable $x_1$ and recursively check that all assignments
leads to satisfiability. If $Q_1=\exists$, then we check that
at least one assignment leads to satisfiability. 
%This recursive approach works
%leads to a $2^{O(n (\log{n} + \log{k}))}$ time algorithm.
However, such an algorithm does not work as intended: 
the formula $\forall x \exists y.y-x \geq 1$
is false when interpreted over any finite $D \subseteq {\mathbb Q}$
while it is true when interpreted over ${\mathbb Q}$.
This implies that another algorithmic approach is needed for handling
quantified DL formula.

\section*{Acknowledgements}

The second and the fourth author were supported by
the Wallenberg AI, Autonomous Systems and Software Program (WASP) funded
by the Knut and Alice Wallenberg Foundation. In addition, the second
author was partially supported by the Swedish Research Council (VR)
under grant 2021-0437. The third author was supported by the Engineering and Physical Sciences Research Council (EPSRC) (Project EP/V00252X/1).


\begin{thebibliography}{10}

\bibitem{AignerZiegler18}
Martin Aigner and G{\"{u}}nter~M. Ziegler.
\newblock {\em Proofs from {THE} {BOOK}}.
\newblock Springer, 6th edition, 2018.

\bibitem{allen1983maintaining}
James~F. Allen.
\newblock Maintaining knowledge about temporal intervals.
\newblock {\em Communications of the ACM}, 26(11):832--843, 1983.

\bibitem{Alur:cav99}
Rajeev Alur.
\newblock Timed automata.
\newblock In {\em Proc. 11th International Conference on Computer Aided
  Verification (CAV-1999)}, pages 8--22, 1999.

\bibitem{Arnborg:etal:sijmaa87}
Stefan Arnborg, Derek Corneil, and Andrzej Proskurowski.
\newblock Complexity of finding embeddings in a $k$-tree.
\newblock {\em SIAM Journal on Matrix Analysis and Applications},
  8(2):277--284, 1987.

\bibitem{Audhya:etal:wcmc2011}
Goutam~K. Audhya, Koushik Sinha, and Sasthi~C. Ghosh.
\newblock A survey on the channel assignment problem in wireless networks.
\newblock {\em Wireless Communications \& Mobile Computing}, 11(5):583--609,
  2011.

\bibitem{barber2000reasoning}
Federico Barber.
\newblock Reasoning on interval and point-based disjunctive metric constraints
  in temporal contexts.
\newblock {\em Journal of Artificial Intelligence Research}, 12:35--86, 2000.

\bibitem{Barrett:etal:SMT}
Clark~W. Barrett, Roberto Sebastiani, Sanjit~A. Seshia, and Cesare Tinelli.
\newblock Satisfiability modulo theories.
\newblock In Armin Biere, Marijn Heule, Hans van Maaren, and Toby Walsh,
  editors, {\em Handbook of Satisfiability}, volume 185 of {\em Frontiers in
  Artificial Intelligence and Applications}, pages 825--885. {IOS} Press, 2009.

\bibitem{Bertele:Brioschi:NDP72}
Umberto Bertelé and Franscesco Brioschi.
\newblock {\em Nonserial Dynamic Programming}.
\newblock Academic Press, 1972.

\bibitem{Bezem:etal:lpar2008}
Marc Bezem, Robert Nieuwenhuis, and Enric Rodr{\'{\i}}guez{-}Carbonell.
\newblock The max-atom problem and its relevance.
\newblock In {\em Proc. 15th International Conference on Logic for Programming,
  Artificial Intelligence, and Reasoning (LPAR-2008)}, pages 47--61, 2008.

\bibitem{Bhargava:Williams:aamas2019}
Nikhil Bhargava and Brian~C. Williams.
\newblock Multiagent disjunctive temporal networks.
\newblock In {\em Proc. 18th International Conference on Autonomous Agents and
  MultiAgent Systems (AAMAS-2019)}, pages 458--466, 2019.

\bibitem{DBLP:journals/jair/BliemMMW20}
Bernhard Bliem, Michael Morak, Marius Moldovan, and Stefan Woltran.
\newblock The impact of treewidth on grounding and solving of answer set
  programs.
\newblock {\em Journal of Artificial Intelligence Research}, 67:35--80, 2020.

\bibitem{Bodirsky:book}
Manuel Bodirsky.
\newblock {\em Complexity of Infinite-Domain Constraint Satisfaction}.
\newblock Cambridge University Press, 2021.

\bibitem{Bodirsky:Dalmau:jcss2013}
Manuel Bodirsky and V{\'{\i}}ctor Dalmau.
\newblock Datalog and constraint satisfaction with infinite templates.
\newblock {\em Journal of Computer and System Sciences}, 79(1):79--100, 2013.

\bibitem{Bodirsky:Jonsson:jair2017}
Manuel Bodirsky and Peter Jonsson.
\newblock A model-theoretic view on qualitative constraint reasoning.
\newblock {\em Journal of Artificial Intelligence Research}, 58:339--385, 2017.

\bibitem{Bodirsky:Mamino:survey}
Manuel Bodirsky and Marcello Mamino.
\newblock Constraint satisfaction problems over numeric domains.
\newblock In {\em The Constraint Satisfaction Problem: Complexity and
  Approximability}, volume~7 of {\em Dagstuhl Follow-Ups}, pages 79--111.
  Schloss Dagstuhl - Leibniz-Zentrum f{\"{u}}r Informatik, 2017.

\bibitem{Bodlaender:sicomp96}
Hans~L. Bodlaender.
\newblock A linear-time algorithm for finding tree-decompositions of small
  treewidth.
\newblock {\em {SIAM} Journal on Computing}, 25(6):1305--1317, 1996.

\bibitem{bodlaender1996efficient}
Hans~L. Bodlaender and Ton Kloks.
\newblock Efficient and constructive algorithms for the pathwidth and treewidth
  of graphs.
\newblock {\em Journal of Algorithms}, 21(2):358--402, 1996.

\bibitem{Boerkoel:Durfee:aaai2013}
James~C. Boerkoel and Edmund~H. Durfee.
\newblock Decoupling the multiagent disjunctive temporal problem.
\newblock In {\em Proc. 27th {AAAI} Conference on Artificial Intelligence
  (AAAI-2013)}, 2013.

\bibitem{Calabro:etal:iwpec2009}
Chris Calabro, Russell Impagliazzo, and Ramamohan Paturi.
\newblock The complexity of satisfiability of small depth circuits.
\newblock In {\em Proc. 4th International Workshop on Parameterized and Exact
  Computation (IWPEC-2009)}, pages 75--85, 2009.

\bibitem{Candeago:etal:sat2016}
Lorenzo Candeago, Daniel Larraz, Albert Oliveras, Enric
  Rodr{\'{\i}}guez{-}Carbonell, and Albert Rubio.
\newblock Speeding up the constraint-based method in difference logic.
\newblock In {\em Proc. 19th International Conference on the Theory and
  Applications of Satisfiability Testing (SAT-2016)}, pages 284--301, 2016.

\bibitem{Carbonnel:Cooper:constraints2016}
Cl{\'{e}}ment Carbonnel and Martin~C. Cooper.
\newblock Tractability in constraint satisfaction problems: A survey.
\newblock {\em Constraints}, 21(2):115--144, 2016.

\bibitem{Cook:stoc71}
Stephen~A. Cook.
\newblock The complexity of theorem-proving procedures.
\newblock In {\em Proc. 3rd Annual {ACM} Symposium on Theory of Computing
  (STOC-1971)}, pages 151--158, 1971.

\bibitem{Courcelle:etal:jcss93}
Bruno Courcelle, Joost Engelfriet, and Grzegorz Rozenberg.
\newblock Handle-rewriting hypergraph grammars.
\newblock {\em Journal of Computer and System Sciences}, 46(2):218--270, 1993.

\bibitem{Courcelle:Olariu:dam2000}
Bruno Courcelle and Stephan Olariu.
\newblock Upper bounds to the clique width of graphs.
\newblock {\em Discrete Applied Mathematics}, 101(1--3):77--114, 2000.

\bibitem{Dabrowski:etal:kr2020}
Konrad~K. Dabrowski, Peter Jonsson, Sebastian Ordyniak, and George Osipov.
\newblock Fine-grained complexity of temporal problems.
\newblock In {\em Proc. 17th International Conference on Principles of
  Knowledge Representation and Reasoning (KR-2020)}, pages 284--293, 2020.

\bibitem{Dabrowski:etal:aaai2021}
Konrad~K. Dabrowski, Peter Jonsson, Sebastian Ordyniak, and George Osipov.
\newblock Disjunctive temporal problems under structural restrictions.
\newblock In {\em Proc. 35th {AAAI} Conference on Artificial Intelligence
  (AAAI-2021)}, pages 3724--3732, 2021.

\bibitem{Dabrowski:etal:ai2023}
Konrad~K. Dabrowski, Peter Jonsson, Sebastian Ordyniak, and George Osipov.
\newblock Solving infinite-domain {CSP}s using the patchwork property.
\newblock {\em Artificial Intelligence}, 317:103880, 2023.

\bibitem{Dechter:etal:ai91}
Rina Dechter, Itay Meiri, and Judea Pearl.
\newblock Temporal constraint networks.
\newblock {\em Artificial intelligence}, 49(1-3):61--95, 1991.

\bibitem{DowneyFellows13}
Rodney~G. Downey and Michael~R. Fellows.
\newblock {\em Fundamentals of Parameterized Complexity}.
\newblock Springer, 2013.

\bibitem{erdos1941problem}
Paul Erd\H{o}s and P{\'a}l Tur{\'a}n.
\newblock On a problem of {S}idon in additive number theory, and on some
  related problems.
\newblock {\em Journal of the London Mathematical Society}, 1(4):212--215,
  1941.

\bibitem{Eriksson:Lagerkvist:ijcai2021}
Leif Eriksson and Victor Lagerkvist.
\newblock Improved algorithms for {A}llen's interval algebra: a dynamic
  programming approach.
\newblock In {\em Proc. 30th International Joint Conference on Artificial
  Intelligence (IJCAI-2021)}, pages 1873--1879, 2021.

\bibitem{book/FlumG06}
J{\"o}rg Flum and Martin Grohe.
\newblock {\em Parameterized Complexity Theory}.
\newblock Springer, 2006.

\bibitem{DBLP:journals/algorithmica/GanianKO21}
Robert Ganian, Fabian Klute, and Sebastian Ordyniak.
\newblock On structural parameterizations of the bounded-degree vertex deletion
  problem.
\newblock {\em Algorithmica}, 83(1):297--336, 2021.

\bibitem{DBLP:journals/ai/GanianO18}
Robert Ganian and Sebastian Ordyniak.
\newblock The complexity landscape of decompositional parameters for {ILP}.
\newblock {\em Artificial Intelligence}, 257:61--71, 2018.

\bibitem{DBLP:journals/algorithmica/GanianOR23}
Robert Ganian, Sebastian Ordyniak, and C.~S. Rahul.
\newblock Group activity selection with few agent types.
\newblock {\em Algorithmica}, 85(5):1111--1155, 2023.

\bibitem{gj79}
Michael~R. Garey and David~S. Johnson.
\newblock {\em Computers and Intractability: A Guide to the Theory of
  {NP}-Completeness}.
\newblock W.H. Freeman and Company, 1979.

\bibitem{Gaspers:ETA}
Serge Gaspers.
\newblock {\em Exponential Time Algorithms - Structures, Measures, and Bounds}.
\newblock {VDM}, 2010.

\bibitem{Gerevini:etal:jair2006}
Alfonso Gerevini, Alessandro Saetti, and Ivan Serina.
\newblock An approach to temporal planning and scheduling in domains with
  predictable exogenous events.
\newblock {\em Journal of Artificial Intelligence Research}, 25:187--231, 2006.

\bibitem{Golumbic:etal:aam94}
Martin~Charles Golumbic, Haim Kaplan, and Ron Shamir.
\newblock On the complexity of {DNA} physical mapping.
\newblock {\em Advances in Applied Mathematics}, 15:251--261, 1994.

\bibitem{gpw06}
Georg Gottlob, Reinhard Pichler, and Fang Wei.
\newblock Bounded treewidth as a key to tractability of knowledge
  representation and reasoning.
\newblock In {\em Proc. 21st National Conference on Artificial Intelligence
  (AAAI-2006)}, pages 250--256, 2006.

\bibitem{Gottlob:etal:ai2002}
Georg Gottlob, Francesco Scarcello, and Martha Sideri.
\newblock Fixed-parameter complexity in {AI} and nonmonotonic reasoning.
\newblock {\em Artificial Intelligence}, 138(1-2):55--86, 2002.

\bibitem{Halberstam:Roth:sequencebook}
Heini Halberstam and Klaus Roth.
\newblock {\em Sequences}.
\newblock Springer, 1966.

\bibitem{Hodges:1997:SMT:262326}
Wilfrid Hodges.
\newblock {\em A Shorter Model Theory}.
\newblock Cambridge University Press, New York, NY, USA, 1997.

\bibitem{hong1991comparison}
Hoon Hong.
\newblock Comparison of several decision algorithms for the existential theory
  of the reals.
\newblock Technical Report 91-41, RISC Linz, 1991.

\bibitem{Huang:etal:ai2013}
Jinbo Huang, Jason~Jingshi Li, and Jochen Renz.
\newblock Decomposition and tractability in qualitative spatial and temporal
  reasoning.
\newblock {\em Artificial Intelligence}, 195:140--164, 2013.

\bibitem{impagliazzo2001problems}
Russell Impagliazzo, Ramamohan Paturi, and Francis Zane.
\newblock Which problems have strongly exponential complexity?
\newblock {\em Journal of Computer and System Sciences}, 63(4):512--530, 2001.

\bibitem{Jonsson:Backstrom:ai98}
Peter Jonsson and Christer B{\"{a}}ckstr{\"{o}}m.
\newblock A unifying approach to temporal constraint reasoning.
\newblock {\em Artificial Intelligence}, 102(1):143--155, 1998.

\bibitem{Jonsson:Lagerkvist:mfcs2018}
Peter Jonsson and Victor Lagerkvist.
\newblock Why are {CSP}s based on partition schemes computationally hard?
\newblock In {\em Proc. 43rd International Symposium on Mathematical
  Foundations of Computer Science (MFCS-2018)}, pages 43:1--43:15, 2018.

\bibitem{Jonsson:Lagerkvist:algorithmica2023}
Peter Jonsson and Victor Lagerkvist.
\newblock General lower bounds and improved algorithms for infinite-domain
  {CSP}s.
\newblock {\em Algorithmica}, 85(1):188--215, 2023.

\bibitem{Jonsson:Loow:ai2013}
Peter Jonsson and Tomas L{\"{o}}{\"{o}}w.
\newblock Computational complexity of linear constraints over the integers.
\newblock {\em Artificial Intelligence}, 195:44--62, 2013.

\bibitem{Kl94}
Ton Kloks.
\newblock {\em Treewidth: Computations and Approximations}, volume 842 of {\em
  LNCS}.
\newblock Springer, 1994.

\bibitem{Kolaitis:Vardi:jcss2000}
Phokion~G. Kolaitis and Moshe~Y. Vardi.
\newblock Conjunctive-query containment and constraint satisfaction.
\newblock {\em Journal of Computer and System Sciences}, 61(2):302--332, 2000.

\bibitem{Koubarakis:tcs2001}
Manolis Koubarakis.
\newblock Tractable disjunctions of linear constraints: basic results and
  applications to temporal reasoning.
\newblock {\em Theoretical Computer Science}, 266(1-2):311--339, 2001.

\bibitem{Kral:dam2005}
Daniel Kr{\'{a}}l'.
\newblock An exact algorithm for the channel assignment problem.
\newblock {\em Discrete Applied Mathematics}, 145(2):326--331, 2005.

\bibitem{Lierler:Susman:tplp2017}
Yuliya Lierler and Benjamin Susman.
\newblock On relation between constraint answer set programming and
  satisfiability modulo theories.
\newblock {\em Theory and Practice of Logic Programming}, 17(4):559--590, 2017.

\bibitem{lokshtanov2018slightly}
Daniel Lokshtanov, D{\'a}niel Marx, and Saket Saurabh.
\newblock Slightly superexponential parameterized problems.
\newblock {\em SIAM Journal on Computing}, 47(3):675--702, 2018.

\bibitem{Lutz:Milicic:jar2007}
Carsten Lutz and Maja Mili\v{c}i\'{c}.
\newblock A tableau algorithm for description logics with concrete domains and
  general tboxes.
\newblock {\em Journal of Automated Reasoning}, 38(1-3):227--259, 2007.

\bibitem{Niebert:etal:ftrtft2002}
Peter Niebert, Moez Mahfoudh, Eugene Asarin, Marius Bozga, Oded Maler, and
  Navendu Jain.
\newblock Verification of timed automata via satisfiability checking.
\newblock In {\em Proc. 7th International Symposium on Formal Techniques in
  Real-Time and Fault-Tolerant Systems (FTRTFT-2002)}, pages 225--244, 2002.

\bibitem{book/Niedermeier06}
Rolf Niedermeier.
\newblock {\em Invitation to Fixed-Parameter Algorithms}.
\newblock Oxford University Press, 2006.

\bibitem{Niemela:amai2008}
Ilkka Niemel{\"{a}}.
\newblock Stable models and difference logic.
\newblock {\em Annals of Mathematics and Artificial Intelligence},
  53(1-4):313--329, 2008.

\bibitem{Nieuwenhuis:etal:jacm2006}
Robert Nieuwenhuis, Albert Oliveras, and Cesare Tinelli.
\newblock Solving {SAT} and {SAT} modulo theories: From an abstract
  {Davis--Putnam--Logemann--Loveland} procedure to {DPLL(\emph{T})}.
\newblock {\em Journal of the {ACM}}, 53(6):937--977, 2006.

\bibitem{Obryant:ejc2004}
Kevin O'Bryant.
\newblock A complete annotated bibliography of work related to {S}idon
  sequences.
\newblock {\em Electronic Journal on Combinatorics}, Dynamic Survey(11), 2004.

\bibitem{Oddi:Cesta:ecai2000}
Angelo Oddi and Amedeo Cesta.
\newblock Incremental forward checking for the disjunctive temporal problem.
\newblock In {\em Proc. 14th European Conference on Artificial Intelligence
  (ECAI-2000)}, pages 108--112, 2000.

\bibitem{Papadimitriou:Steiglitz:CO}
Christos~H. Papadimitriou and Kenneth Steiglitz.
\newblock {\em Combinatorial Optimization: Algorithms and Complexity}.
\newblock Prentice-Hall, 1982.

\bibitem{pe1997satisfiability}
Itsik Pe'er and Ron Shamir.
\newblock Satisfiability problems on intervals and unit intervals.
\newblock {\em Theoretical Computer Science}, 175(2):349--372, 1997.

\bibitem{Peintner:etal:cp2007}
Bart Peintner, Kristen~Brent Venable, and Neil Yorke{-}Smith.
\newblock Strong controllability of disjunctive temporal problems with
  uncertainty.
\newblock In {\em Proc. 13th International Conference on Principles and
  Practice of Constraint Programming (CP-2007)}, pages 856--863, 2007.

\bibitem{Pietrzak03}
Krzysztof Pietrzak.
\newblock On the parameterized complexity of the fixed alphabet shortest common
  supersequence and longest common subsequence problems.
\newblock {\em Journal of Computer and System Sciences}, 67(4):757--771, 2003.

\bibitem{renegar1992computational}
James Renegar.
\newblock On the computational complexity and geometry of the first-order
  theory of the reals. {P}art {I}: Introduction. preliminaries. {T}he geometry
  of semi-algebraic sets. {T}he decision problem for the existential theory of
  the reals.
\newblock {\em Journal of Symbolic Computation}, 13(3):255--299, 1992.

\bibitem{Robertson:Seymour:jctb84}
Neil Robertson and Paul~D. Seymour.
\newblock Graph minors. {III.} {P}lanar tree-width.
\newblock {\em Journal of Combinatorial Theory, Series B}, 36(1):49--64, 1984.

\bibitem{Samer:Szeider:jcss2010}
Marko Samer and Stefan Szeider.
\newblock Constraint satisfaction with bounded treewidth revisited.
\newblock {\em Journal of Computer and System Sciences}, 76(2):103--114, 2010.

\bibitem{Schutt:Stuckey:informs2010}
Andreas Schutt and Peter~J. Stuckey.
\newblock Incremental satisfiability and implication for {UTVPI} constraints.
\newblock {\em {INFORMS} Journal on Computing}, 22(4):514--527, 2010.

\bibitem{Seshia:etal:jsbmc2007}
Sanjit~A. Seshia, K.~Subramani, and Randal~E. Bryant.
\newblock On solving boolean combinations of {UTVPI} constraints.
\newblock {\em Journal on Satisfiability, Boolean Modeling and Computation},
  3(1-2):67--90, 2007.

\bibitem{sidon1932satz}
Simon Sidon.
\newblock Ein {S}atz {\"u}ber trigonometrische {P}olynome und seine {A}nwendung
  in der {T}heorie der {F}ourier-{R}eihen.
\newblock {\em Mathematische Annalen}, 106(1):536--539, 1932.

\bibitem{socala2016tight}
Arkadiusz Soca{\l}a.
\newblock Tight lower bound for the channel assignment problem.
\newblock {\em ACM Transactions on Algorithms}, 12(4):48, 2016.

\bibitem{Stergiou:Koubarakis:ai2000}
Kostas Stergiou and Manolis Koubarakis.
\newblock Backtracking algorithms for disjunctions of temporal constraints.
\newblock {\em Artificial Intelligence}, 120(1):81--117, 2000.

\bibitem{Stockmeyer:Meyer:stoc73}
Larry~J. Stockmeyer and Albert~R. Meyer.
\newblock Word problems requiring exponential time: Preliminary report.
\newblock In {\em Proc. 5th Annual {ACM} Symposium on Theory of Computing
  (STOC-1973)}, pages 1--9, 1973.

\bibitem{traxler2008time}
Patrick Traxler.
\newblock The time complexity of constraint satisfaction.
\newblock In {\em Proc. 3rd International Workshop on Parameterized and Exact
  Computation (IWPEC-2008)}, pages 190--201. Springer, 2008.

\bibitem{Tsamardinos:Pollack:ai2003}
Ioannis Tsamardinos and Martha~E. Pollack.
\newblock Efficient solution techniques for disjunctive temporal reasoning
  problems.
\newblock {\em Artificial Intelligence}, 151(1-2):43--89, 2003.

\bibitem{BrentVenable:ijcai2005}
Kristen~Brent Venable and Neil Yorke{-}Smith.
\newblock Disjunctive temporal planning with uncertainty.
\newblock In {\em Proc. 19th International Joint Conference on Artificial
  Intelligence (IJCAI-2005)}, pages 1721--1722, 2005.

\bibitem{Vilain:Kautz:aaai86}
Marc~B. Vilain and Henry~A. Kautz.
\newblock Constraint propagation algorithms for temporal reasoning.
\newblock In {\em Proc. 5th National Conference on Artificial Intelligence
  (AAAI-1986)}, pages 377--382, 1986.

\bibitem{Wanke:dam94}
Egon Wanke.
\newblock $k$-{NLC} graphs and polynomial algorithms.
\newblock {\em Discrete Applied Mathematics}, 54(2-3):251--266, 1994.

\bibitem{Zavatteri:etal:constraints2023}
Matteo Zavatteri, Alice Raffaele, Dario Ostuni, and Romeo Rizzi.
\newblock An interdisciplinary experimental evaluation on the disjunctive
  temporal problem.
\newblock {\em Constraints}, 28(1):1--12, 2023.

\end{thebibliography}
\end{document}